\documentclass[11pt]{article}
\usepackage{theapa, rawfonts}

\usepackage[margin = 1.5in]{geometry}
\usepackage{amsmath}
\usepackage{amssymb}
\usepackage{appendix}
\usepackage{amsthm}
\usepackage[ruled]{algorithm2e}
\usepackage{subfig}
\usepackage{paralist}
\usepackage{thmtools}
\usepackage{mathtools}
\usepackage{graphicx}
\usepackage{authblk}

\renewcommand{\cal}[1]{\mathcal{#1}}
\newcommand{\me}{\mathrm{e}}
\DeclareMathOperator{\eps}{\varepsilon}

\newcommand{\Th}{\textit{Th}}
\newcommand{\R}{\mathbb{R}}
\newcommand{\E}{\mathbb{E}}
\newcommand{\I}{\mathbb{I}}

\newcommand{\ER}{\mathit{ER}}
\newcommand{\True}{\mathit{True}}
\newcommand{\False}{\mathit{False}}
\newcommand{\Alloc}{\mathit{Alloc}}
\newcommand{\Des}{\mathit{D}}
\newcommand\iid{\stackrel{\mathclap{\mbox{\tiny{iid}}}}{\sim}}

\newcounter{int}

\makeatletter
\newcommand{\citenp}[1]{\setcounter{int}{0}\@for\tmp:=#1\do{%
\ifnum \value{int}>0; \fi%
\setcounter{int}{1}%
\citeauthor{\tmp}, \citeyear{\tmp}}}
\makeatother

\declaretheorem[name=Theorem,numberwithin=section]{theorem}

\theoremstyle{definition}

\theoremstyle{definition}

\newtheorem{example}{Example}[section]

\DeclareMathOperator*{\argmax}{arg\,max}
\DeclareMathOperator*{\argmin}{arg\,min}

\begin{document}


\title{The Price is (Probably) Right: Learning Market Equilibria from Samples}

\author[1]{Vignesh Viswanathan}
\author[2]{Omer Lev}
\author[3]{Neel Patel}
\author[1]{Yair Zick}
\affil[1]{University of Massachusetts, Amherst}
\affil[2]{Ben-Gurion University of the Negev}
\affil[3]{University of Southern California}
\date{}


\maketitle

\begin{abstract}
Equilibrium computation in markets usually considers settings where player valuation functions are known. We consider the setting where player valuations are \emph{unknown}; using a PAC learning-theoretic framework, we analyze some classes of common valuation functions, and provide algorithms which output direct PAC equilibrium allocations, not estimates based on attempting to learn valuation functions.
Since there exist trivial PAC market outcomes with an unbounded worst-case efficiency loss, we lower-bound the efficiency of our algorithms.
While the efficiency loss under general distributions is rather high, we show that in some cases (e.g., unit-demand valuations), it is possible to find a PAC market equilibrium with significantly better utility.
\end{abstract}

\maketitle
\section{Introduction}\label{sec:intro}
Do markets admit equilibrium allocations? This question has been extensively studied for many years \shortcite{varian1974markets}; more recently, the econ/CS community devoted significant effort to understanding when one can \emph{efficiently compute} market equilibria.  Much of this literature  assumes that one has full access to player valuations over bundles of goods, an unrealistic assumption in many instances: combinatorial valuations are often difficult to elicit (especially for large markets), precluding any possibility of running full-information market algorithms. Machine learning techniques offer a compromise -- assuming access to a \emph{partial} dataset, we can learn player valuations, and use the learned valuations as a proxy. However, this approach raises several issues too: market valuations are often complex, and require a large number of samples to learn without overfitting. Moreover, even if we assume that player valuations have a simple structure, it is not immediately obvious that an \emph{exact} equilibrium for the approximate valuations acts as an \emph{approximate} equilibrium for the exact valuations; as we shall show, this may not be the case. 

Our work explores a relatively new paradigm: instead of learning valuations, we focus on directly learning \emph{market equilibria} from data. We build upon the framework of \shortciteA{JZ20}, and adopt the \emph{PAC solution learning} framework. \shortciteA{JZ20} show that in order to ensure that a market outcome (i.e., an allocation of items to players, as well as item prices) is likely to be a market equilibrium, it suffices to show that it is consistent with the data. That is, the prices and item allocation they induce are such that no player has a sample in the data they can afford and would rather have over their allocation. Finding a consistent market outcome is trivial: setting the price of all the goods to infinity would ensure consistency. However, this outcome would be very inefficient. Our goal is thus to learn \emph{approximately efficient} PAC market equilibria.

We make two main contributions:
\begin{itemize}
\item We suggest direct algorithms to find ``good'' equilibria for several families of utility functions, and show that they guarantee reaching the best theoretical bound possible on the equilibria they find.
\item Working on a dataset derived from real-world settings, we see how in simple utility functions and few samples, learning valuations can outperform our algorithm, yet with more elaborate utility functions, we are able to find a good equilibrium approximations (given enough samples) \footnote{The code can be found at \texttt{https://github.com/vignesh-viswanathan/Learning\_Market\_Equilibria}}.
\end{itemize}

\subsubsection*{Paper Structure}\label{subsec:contrib}
We study Fisher markets with indivisible goods under different classes of valuation functions, and propose algorithms which output an efficient PAC market equilibrium. That is, each player receives, with high probability, their most preferred affordable bundle of goods. 
We examine a variety of valuation classes: \emph{unit-demand} (Section \ref{sec:unitdemand}), \emph{single minded} (Section \ref{sec:singleminded}), \emph{additive} (Section \ref{sec:additive}) and \emph{submodular} (Section \ref{sec:submodular}) valuations. For each class, we provide a tight, distribution-independent, efficiency bound. We also show that, under more favorable distributions, we can achieve far better efficiency guarantees for unit-demand and additive valuations. We end with an empirical examination of our results on realistic-seeming data (Section~\ref{sec:expts}).

\subsection{Related Work}\label{subsec:relwork}
There is a rich body of classical literature on market equilibria with indivisible goods  \shortcite{eisenberg1961aggregation,varian1974markets,kc82,lemefastequilibria,gul99}, exchange economies \shortcite{bikhchandani} and Fisher markets \shortcite{babaioff2017additive,BLS16,SegalHalevi2017CompetitiveEF,branzeisingleminded,branzei2015unitdemand}. 
In recent years there has been a significant renewed interest in computing market outcomes, such as fair allocation  \shortcite{Kurokawa2016,farhadi2019}, optimal pricing \shortcite{guruswamiprices,renatoprices}, approximate equilibria \shortcite{budish2011,barman2018additive} and markets with divisible goods \shortcite{devanur2008market}. 
However, the above do not address learning approximately efficient market solutions from data. 

There exists a fast-growing body of literature on learning game-theoretic solutions from data: in cooperative games \shortcite{balkanski2017costsharing,sliwinski2017hedonic,balcan2015learning,igarashi2019stable,JZ20}, auctions \shortcite{psomas2016auctions,brero2018auctions,cole2014,jamie2015nips} and optimization \shortcite{balkanski2017minops,balkanski2017opt,rosenfield2018}. Some of this literature propose methods to learn market outcomes as well: \shortciteA{murray2020robust} and \shortciteA{kroer2019computing} examine the simpler case with divisible goods and additive valuations, \shortciteA{renatoprices} examine markets with a single item and \shortciteA{greenwald2020} propose a method to learn market outcomes indirectly from noisy valuations. However, to the best of our knowledge, there exists no prior work that attempts to learn market outcomes in combinatorial markets from samples. 

\section{Model and Preliminaries}\label{sec:model}
We study the \emph{Fisher market} model; there is a set of \emph{players}, $N = \{1,2,\dots,n\}$ and a set of \emph{goods}, $G = \{g_1,g_2,\dots,g_k\}$. 
Each player $i$ has a \emph{budget} $b_i \in \R_+$ and a \emph{valuation function} $v_i:2^G \rightarrow\R_+ \cup \{0\}$ which assigns a value $v_i(S)$ for each bundle of goods $S \subseteq G$.
We assume that no two players have the same budget, and that $b_1>\ldots> b_n$. 
This is a standard assumption, and is not a significant loss of generality: it is mostly done to induce some priority order among players, and ensure that equilibria exist. When budgets are equal, one can introduce small perturbations (as also done in \shortciteA{budish2011}). 
An \emph{allocation} in such a market is a tuple $(\cal A,\vec{p})$, where $\cal A = \{A_1, A_2, \dots, A_n\}$ ($A_i\subseteq G$ for $1\leq i\leq n$; and $\cal A$ is pairwise disjoint) is the allocation vector and $\vec{p} = \{p_1, p_2, \dots, p_k\}$ is the price vector. Some $A_i$s may be empty; i.e., if $A_i = \emptyset$ then player $i$ receives nothing. 
We define the \emph{affordable set} $D_i(\vec{p},b_i)$ as the set of affordable bundles for player $i$ given a price vector $\vec{p}$: 
\begin{align}
D_i(\vec{p},b_i) = \{ S \subseteq G : \sum_{g_j \in S} p_j \le b_i\}\notag
\end{align}
%
An allocation is a \emph{Walrasian} equilibrium (or simply an equilibrium) if all players are allocated the best possible set of goods they can afford, i.e., 
$A_i \in D_i(\vec p,b_i)$ for all $i \in N$, and for all $i \in N$:
\begin{align}
v_i(A_i) \in \argmax\{v_i(S):S \in D_i(\vec p,b_i)\} \notag
\end{align} 
\shortciteA{JZ20} define a learning-theoretic equilibrium notion based on the \emph{probably approximately correct} (PAC) framework \shortcite{anthony1999learning} called \emph{PAC Equilibria}. 
An allocation $(\cal A, \vec{p})$ is a PAC Equilibrium if it is unlikely that a sample from a distribution $\cal D$ (over bundles in $G$), under the same prices, is both better than the current allocation and affordable for some player $i$. 
It is often easier to discuss learning-theoretic notions in terms of their \emph{expected loss}; here, the loss is a function of player valuations $v$ and budgets $\vec b$, a bundle $S \subseteq G$, and the proposed outcome $(\cal A,\vec p)$:
\begin{align}
    L_{v,\vec b}(S,\cal A,\vec p) = \begin{cases}
        1 & \mbox{if }\exists i \in N, v_i(A_i) < v_i(S) \\
          &\land S \in D_i(\vec{p},b_i)\\
        0 & \mbox{otherwise.}
    \end{cases}\label{eq:lossdefinition}
\end{align}
We omit the $v$ and $\vec b$ subscripts when they are clear from context. An allocation is an $\varepsilon$-PAC equilibrium with respect to $\cal D$ if its expected loss, denoted $L_{\cal D}(\cal A,\vec p)$, is lower than $\varepsilon$
\begin{align}
L_{\cal D}(\cal A,\vec p) \triangleq \E_{S \sim \cal D}[L_{v,\vec b}(S,A,\vec p)] < \varepsilon\label{eq:pac-equilibrium}
\end{align}
$\varepsilon$-PAC equilibria are somewhat similar to $\varepsilon$-PAC approximations \shortcite{anthony1999learning}: given a function $u:2^G \to \R$, a $\bar{u}$ is an $\varepsilon$-PAC approximation of $u$ w.r.t. $\cal D$ if 
$\Pr_{S \sim \cal D}[u(S) \ne \bar{u}(S)] < \varepsilon$.

We follow a standard model of learning from samples: we are given players' budgets $b_1, b_2, \dots, b_n$, as well as $m$ input samples $S_1, S_2, \dots, S_m$ drawn i.i.d. from a distribution $\cal D$, and player valuations over the samples: $v_i(S_j)$ for all $i \in N$ and $j \in [m]$. 
Our goal is to find algorithms, whose input is a set of i.i.d. sampled bundles and valuations over them, that output a PAC Equilibrium (as per Equation~\eqref{eq:pac-equilibrium}) with probability $\ge 1 - \delta$ (over the randomization of sampling $m$ i.i.d. samples from $\cal D$). In other words, if $(\cal A,\vec p)$ is the output of some learning algorithm, then the PAC guarantee is
\begin{align*}
    \Pr_{S_1,\dots,S_m\iid \cal D}\left[L_{\cal D}(\cal A,\vec p) < \varepsilon\right]\ge 1 - \delta
\end{align*}
The number of samples needed, $m$, should be polynomial in the number of players, the number of goods, and in $\frac1\varepsilon,\log\frac1\delta$. As mentioned in Section~\ref{sec:intro}, PAC equilibria are not guaranteed to be efficient; in what follows we explore market \emph{stability}, \emph{envy} and allocative \emph{efficiency}.

An allocation $(\cal A, \vec{p})$ is said to be \emph{envy free} if all the players prefer their bundle to every other player's bundle that they can afford, i.e., an allocation is envy free if for all $i,j\in N$:
\begin{align}
    v_i(A_i) \ge v_i(A_j) \lor A_j \notin D_i(\vec{p}, b_i) \label{eq:EF-allocation} 
\end{align}
Note that if $(\cal A,\vec p)$ is a Walrasian equilibrium, then \eqref{eq:EF-allocation} is trivially true. 
The \emph{efficiency ratio} of an allocation $\cal A$ is the ratio of the total welfare (or utility) of $\cal A$ to that of the optimal equilibrium allocation, i.e.,

\begin{align}
\ER_v(\cal A) = \frac{\sum_{i=1}^{n} v_i(A_i)}{\sum_{i=1}^{n} v_i(A_i^*)} 
\end{align}

where $\cal A^*$ is a welfare-maximizing equilibrium allocation; Unlike simpler settings (e.g. rent division \shortcite{gal2017rent}), market outcomes need not maximize social welfare in Fisher markets with indivisible goods.\\

Some of the proofs have been omitted or replaced by proof sketches for conciseness. 
The full proofs can be found in the appendix.
\subsection{A Short Primer on PAC Learning}\label{sec:PAClearning}
In this section we briefly introduce the theory of PAC learning. The familiar reader may skip this section, or refer to \cite{anthony1999learning,kearnsvazirani1994learning}. 
Probably Approximately Correct learning, or PAC learning, is a formal treatment of the number of samples needed in order to learn a \emph{model} from samples. Let $\cal D$ be a distribution over a sample space $\cal X$. A \emph{hypothesis class} $\cal H$, is a class of functions $f:\cal X \to \cal Y$, where $\cal Y$ is a label space. For example, if $\cal Y = \{\pm1\}$ then $\cal H$ is a class of binary classifiers. $\cal H$ can be any class of potential learners, e.g. linear functions or deep neural networks. 
Suppose that there is some model $f \in \cal H$ that labels elements in $\cal X$; can we recover a model $\hat f$ that well approximates $f$? More formally, we define a loss of $\hat f$ over samples from $\cal D$ as follows: 
$$L_{\cal D}^f(\hat f) = \E_{\vec x \sim \cal D}\left[\I[f(\vec x) \ne \hat f(\vec x)]\right].$$
Here $\I(f(\vec x) \ne \hat f(\vec x)]$ equals $1$ when $f$ and $\hat f$ do not agree, and is $0$ otherwise. We are interested in finding a hypothesis $\hat f \in \cal H$ that exhibits low loss with respect to $\cal D$, i.e. $L_{\cal D}^f(\hat f) < \eps$. In particular, we are interested in \emph{learning} $\hat f$ from a set of samples $\vec x_1,\dots,\vec x_m \iid \cal D$. Let $\cal A$ be an algorithm that takes as input a set of $m$ i.i.d. samples from $\cal D$, and outputs a hypothesis $\hat f\in \cal H$. We say that $\cal A$ PAC learns $\cal H$ if for every $f\in \cal H$ and every $\eps,\delta >0$, $\cal H$ outputs a hypothesis $\hat f$ satisfying 
\begin{align*}
    \Pr_{\vec x_1,\dots,\vec x_m\iid \cal D}[L_{\cal D}^f(\hat f) \ge \eps]]< \delta
\end{align*}

\section{Computing PAC Equilibria}\label{sec:theory}
We first discuss some sufficient conditions for finding a PAC Equilibrium from samples, starting with a simple observation:
if we are able to approximate player valuation functions $v$ using an underestimate $\bar{v}$, then any \emph{exact} equilibrium for $\bar{v}$ is a PAC equilibrium for $v$. 
\begin{restatable}{prop}{proplowerboundequilibrium}\label{prop:lowerbound-equilibrium}
Let $v_1,\dots,v_n:2^G \to \R$ be a player valuation profile; let $(\bar{v}_i)_{i \in N}$ be $\frac{\varepsilon}{n}$-PAC approximations of $(v_i)_{i \in N}$ w.r.t. $\cal D$, such that for all $i \in N$ and all $S \subseteq G$, $
\bar{v}_i(S) \le v_i(S)$. If $(\cal A,\vec p)$ is a market equilibrium under $\bar{v}$, then $(\cal A,\vec p)$ is an $\varepsilon$-PAC equilibrium for $v$ w.r.t. $\cal D$. 
\end{restatable}
\begin{proof}
By the union bound, 
$\Pr_{S\sim \cal D} [ \exists i \in N, v_i(S) \ne 
\bar{v}_i(S)] \le \varepsilon$.

To show $(\cal A,\vec p)$ is an $\varepsilon$-PAC equilibrium for $v$, we will bound:
\begin{align}
&\Pr_{S \sim \cal D}[\exists i, (v_i(S) > v_i(A_i))\land (S \in D_i(\vec{p},b_i))] \notag\\
&=\Pr_{S \sim \cal D}[\exists i, (v_i(S) > v_i(A_i))\land (S \in
D_i(\vec{p},b_i))| \forall i, v_i(S) = \bar{v}_i(S)] \Pr_{S \sim \cal D}[\forall i, v_i(S) = \bar{v}_i(S)]  \notag\\
&\quad + \Pr_{S \sim \cal D}[\exists i, (v_i(S) > v_i(A_i))\land (S \in
D_i(\vec{p},b_i))| \exists i, v_i(S) \ne \bar{v}_i(S)] \Pr_{S \sim \cal D}[\exists i, v_i(S) \ne \bar{v}_i(S)] \notag\\
&\le \Pr_{S \sim \cal D}[\exists i, (\bar{v}_i(S) > v_i(A_i))\land (S \in D_i(\vec{p},b_i))]\Pr_{S \sim \cal D}[\forall i, v_i(S) = \bar{v}_i(S)] + \varepsilon \notag\\
&\le \Pr_{S \sim \cal D}[\exists i, (\bar{v}_i(S) > \bar{v}_i(A_i))\land (S \in D_i(\vec{p},b_i))]\Pr_{S \sim \cal D}[\forall i, v_i(S) = \bar{v}_i(S)] + \varepsilon\label{eq:lowerbound-v}\\
&=  \varepsilon \label{eq:Ap-equilibrium}
\end{align}
The transition to \eqref{eq:lowerbound-v} is due to the theorem assumption that $v_i(S) \ge \bar{v}_i(S)$; the transition to \eqref{eq:Ap-equilibrium} is because $(\cal A,\vec p)$ is a market equilibrium for $\bar{v}$, so there can be no set $S$ which any player $i$ can get and values more than $A_{i}$.

\end{proof}

\shortciteA{JZ20} prove that a PAC equilibrium can be directly learned using only $\cal O(k)$ samples if one can efficiently compute a \emph{consistent solution}, that is, a market outcome that has zero loss on the samples. More precisely, we say that a mechanism $\cal M$ outputs a \emph{consistent solution} if for any given set of samples $\cal S = \{S_1,\dots,S_m\}$, $\cal M$ outputs $(\cal A,\vec{p})$ such that the \emph{empirical loss} $\hat L(\cal A,\vec p)$ is $0$:
\begin{align*}
    \hat L (\cal A,\vec p) \triangleq \frac1m\sum_{j = 1}^m L(S_j,\cal A,\vec p) = 0
\end{align*}

\begin{theorem}[\shortciteA{JZ20}]\label{thm:consistency}
Suppose that an algorithm $\mathcal{M}$ takes as input a set of $m$ samples of goods $\cal S$ drawn i.i.d. from an unknown distribution $\cal D$, and outputs a consistent equilibrium allocation. If $m \in \mathcal{O}\Big( \frac{1}{\varepsilon}\big( k\log \frac{1}{\varepsilon} + \log \frac{1}{\delta} \big) \Big)$ then the allocation output by $\mathcal{M}$ is an $\varepsilon$-PAC Equilibrium w.p. $ \ge 1- \delta$.
\end{theorem}
Proposition~\ref{prop:lowerbound-equilibrium} and Theorem~\ref{thm:consistency} provide two paths to computing PAC market equilibria: either compute an equilibrium for a PAC from an underestimate of the valuations, or directly learn market outcomes from samples. As we mention above, our objective is finding market outcomes with provable social welfare guarantees, with respect to the \emph{true} player valuation profile.  

\section{Unit Demand Markets}\label{sec:unitdemand}
We begin our exploration with a fundamental class of market valuations: unit-demand markets. In a unit demand market, the value of each bundle $S\subseteq G$ is the value of the most valuable good in $S$, i.e., for all $i \in N$, $v_i(S) = \max_{g \in S} v_i(\{g\})$. 
We make the standard assumption that players have distinct values for goods, i.e., that $v_i(\{g\}) \ne v_i(\{g'\})$ if $g\ne g'$; this is mostly done to break ties (see \shortciteA{budish2011}).
Unit demand markets correspond to settings such as room/housing allocation scenarios where each tenant can only stay in a single room/buy a single home \shortcite{alkan1991allocation,aragones1995rent,gal2017rent}, or to gaming ``loot boxes'', in which players care mainly about the most valuable item. 

The standard data-driven approach is to PAC learn the valuation functions, and output an equilibrium allocation for the learned valuations. 
We refer to this method as \emph{indirect learning}, and to outcomes computed in this manner as indirectly learned outcomes. For unit demand markets this can be done quite easily, by estimating the value of each item as the value of the least valuable sample that contains it (creating a PAC approximation for the valuations), and then allocating the items first to the player with the largest budget, who gets their most valued item; then the player with the second largest budget, who gets their most valued item which is still available, and so on.
Algorithm \ref{algo:unit-demand-indirect} follows this method for unit-demand markets.
\begin{algorithm}[tbh]

\DontPrintSemicolon
\SetAlgoLined
\LinesNumbered
\SetNlSty{}{}{:}
\KwIn{A set of samples $\cal S$; player valuations over samples $(v(\cal S))$; player budgets $b_1 > \dots >b_n$}

\ForEach {$i\in N;g\in G$}{
$\bar{v}_i(\{g\}) \leftarrow +\infty$\;
$\bar{v}_i(\{g\}) \leftarrow \min_{S \in \cal S: g \in S} v_i(S)$\; \label{line:val-learn}
}
$R_1 \gets \bigcup_{S \in \cal S}S$\; \label{line:indequib-start}
\For{$i \gets 1 \textbf{ to } n$}{
    $g^* \gets \argmax_{g \in R_i} \bar{v}_i(\{g\})$\;
    Allocate $g^*$ to player $i$; set the price of $g^*$ to $b_i$\;
    $R_{i+1} \leftarrow R_i \setminus \{g^*\}$\; 
    \If {$R_{i+1}=\emptyset$}{ \tcc{All sampled goods have been allocated. We now allocate the rest.}
        $R_{i+1}\leftarrow G\setminus R_1$\;
        $R_1\leftarrow G$\;
    }\label{line:indequib-end}
}
Allocate leftover goods to player $n$ at price $0$\;

\caption{Indirectly Learning Outcomes for Unit Demand Markets}
\label{algo:unit-demand-indirect}
\end{algorithm}

Algorithm \ref{algo:unit-demand-indirect} first learns a consistent valuation profile in line \ref{line:val-learn} (and thus serves as a PAC approximation for the true valuations as per classic PAC learning results \shortcite{anthony1999learning}). Next, it iterates over the set of players in decreasing order of budgets, assigning each player $i$ their most preferred unallocated good at a price of $b_i$. 
Lines \ref{line:indequib-start} to \ref{line:indequib-end} compute an exact equilibrium for the learned valuations: players always pick their most preferred available item, and cannot afford any previously allocated item, as players are chosen in a decreasing order of budgets.

In terms of efficiency, Algorithm \ref{algo:unit-demand-indirect} guarantees an efficiency inversely proportional to the disparity in player valuations.
\begin{restatable}{prop}{propUDindirecteff}\label{prop:UD-indirect-eff}
If $(\cal A,\vec p)$ is the output of Algorithm \ref{algo:unit-demand-indirect}, then $\ER_v(\cal A)\ge \frac{1}{\sigma}$ where $\sigma = \max_{i \in N} \frac{\max_{g \in G}v_i(\{g\})}{\min_{g \in G}v_i(\{g\})}$, the maximal ratio between a players valuation for two different items.
\end{restatable}
\begin{proof}
If a good is valued at $0$ by some player, then $\sigma$ is undefined. So, we only consider the case where all goods have a non-zero valuation.

Let $\cal A^*$ be the socially optimal allocation. When $n \ge k$, in the optimal equilibrium allocation, the top $k$ players budget-wise get one good each and the rest of the players get nothing. Indirect learning also allocates a good to each of the top $k$ players budget wise. 
If $g_i$ is the good allocated to player $i$ in the algorithm, then the value to player $i$ in the optimal equilibrium allocation is bounded by $\sigma v_i(g_i)$. This gives us the efficiency $\ER_v(\cal A)$,
\begin{align*}
    \frac{\sum_{i = 1}^{n} v_i(A_i)}{\sum_{i = 1}^{n} v_i(A_i^*)} 
    = \frac{\sum_{i = 1}^{k} v_i(A_i)}{\sum_{i = 1}^{k} v_i(A_i^*)} 
    \ge \frac{\sum_{i = 1}^{k} v_i(A_i)}{\sum_{i = 1}^{k} \sigma v_i(A_i)} 
    = \frac{1}{\sigma}
\end{align*}
Similarly, when $n \le k$, every player gets one good and some goods may be left unallocated. Indirect learning also allocates a good to each player. 
If $g_i$ is the good allocated to player $i<n$ in the algorithm, then the value to player $i$ in the optimal equilibrium allocation is bounded by $\sigma v_i(g_i)$. This gives us
\begin{align*}
    \ER_v(\cal A)
    = \frac{\sum_{i = 1}^{n} v_i(A_i)}{\sum_{i = 1}^{n} v_i(A_i^*)} 
    \ge \frac{\sum_{i = 1}^{n} v_i(A_i)}{\sum_{i = 1}^{n} \sigma v_i(A_i)} 
    = \frac{1}{\sigma}
\end{align*}
\end{proof}

The main drawback with such an algorithm is that it does not output a PAC Equilibrium. 
Consider the example below:
\begin{example}\label{ex:unitdemandindirect}
Consider a setting where $N =  \{1,2\}$ and $G = \{g_1,g_2,g_3\}$. Player budgets are $b_1 = 2,b_2 = 1$. Player valuations satisfy
\begin{align}
 v_1(\{g_1, g_2\}) = 5;\; &v_1(\{g_3\}) = 3 \notag \\
v_2(\{g_1, g_2\}) = 4 ;\; & v_2(\{g_3\}) = 2 \notag 
\end{align}
We observe a distribution $\cal D$ which samples uniformly at random two sets: $\{g_1,g_2\}$ and $\{g_3\}$. We can thus reasonably assume that we observe both bundles with high probability after a small number of i.i.d. samples. 
Approximating preferences would yield:
\begin{align*}
    \bar{v}_1(\{g_1\}) = \bar{v}_1(\{g_2\}) = 5;\; & \bar{v}_1(\{g_3\}) = 3\\
    \bar{v}_2(\{g_1\}) = \bar{v}_2(\{g_2\}) = 4;\; & \bar{v}_2(\{g_3\}) = 2
\end{align*}
A valuation-approximating algorithm allocates one item from $g_1,g_2$ to player 1 and the other to player 2, and allocates $g_3$ to player 2. We set the price of $g_1$ to $2$ and the price of $g_2$ to $1$. Assume w.l.o.g. that $g_1$ is assigned to player 1; it is possible that $v_1(\{g_1\}) = 0$ and $v_1(\{g_2\}) = 5$, in which case player 1 demands $g_3$. In that case, the probability of observing a sample (namely $\{g_3\}$) which player~$1$ demands is $\frac12$, not an arbitrarily low $\varepsilon>0$, so this approach does not yield an $\varepsilon$-PAC equilibrium.
\end{example}
The bad result in Example \ref{ex:unitdemandindirect} is not due to some intrinsic failure of the valuation-approximating algorithm; it is impossible to learn a consistent underestimate of a unit demand valuation. Consider again the setting in Example \ref{ex:unitdemandindirect}: it is impossible to determine whether $v_1(\{g_1\}) = 5$ or $v_1(\{g_2\}) = 5$; indeed, the only viable underestimate sets both items' values to $0$. However, doing so yields $\bar{v}_1(\{g_1,g_2\}) = 0 < v_1(\{g_1,g_2\})$, an inconsistency.
To conclude, the indirect approach does not yield a PAC Equilibrium.

Let us turn our attention to directly learning PAC market outcomes from samples. 
We refer to this method as \emph{direct solution learning}, and any outcome computed from this method as a directly learned equilibrium. 
Algorithm \ref{algo:unit-demand-market} directly learns a PAC equilibrium in the unit-demand setting. It iterates over all players in decreasing order of budget, and allocates the smallest bundle of goods from all available goods with the highest possible value. 
We use two properties of unit demand valuations, formalized in the following lemma.
\begin{restatable}{lemma}{lemunitdemandproperties}\label{lem:unitdemand-properties}
Given two bundles of goods $S,T \subseteq G$ and some player $i \in N$ with unit demand valuations, if no two goods have the same value for $i$ then
\begin{enumerate}
    \item\label{STequal} If $v_i(S) = v_i(T) = c$ then $v_i(S\cap T) = c$ as well.
    \item\label{STdifferent}If $v_i(S) > v_i(T)$ then $v_i(S) = v_i(S\setminus T)$
\end{enumerate}
\end{restatable}
\begin{proof}
Since player $i$ has a unit-demand valuation, 
\begin{align}
    v_i(S) = &\max\{v_i(S\cap T),v_i(S\setminus T)\}\\
    v_i(T) = &\max\{v_i(S\cap T),v_i(T\setminus S)\};
\end{align} 
since all items have different values, it must be the case that $v_i(S)>v_i(S\cap T)$ or $v_i(S) > v_i(S\setminus T)$. 
Suppose that $v_i(S) = v_i(T) = c$ and $c > v_i(S\cap T)$; then it must be the case that $v_i(S) = v_i(S\setminus T) =c$ and $v_i(T) = v_i(T\setminus S) = c$. However, this implies that there are two disjoint goods: $g\in S\setminus T$ and $g' \in T \setminus S$ that are equally valued by $i$, a contradiction. We have thus proven Item~\ref{STequal}. 

Similarly, if $v_i(S) > v_i(T)$ and $v_i(S) = v_i(S\cap T)$, we get that
$$
v_i(S) > v_i(T) \ge v_i(S\cap T) = v_i(S), 
$$
a contradiction. Therefore $v_i(S) = v_i(S\setminus T)$ which proves Item~\ref{STdifferent}.
\end{proof}

Using Lemma~\ref{lem:unitdemand-properties}, we identify the smallest most valued bundle $B^1_i$ for player $i$, and allocate it to the player if it contains no previously allocated items; otherwise, we remove all such samples from $\cal S$, since we know such items are already priced out of their budget by previous players, and we cannot use them to get information on the next most valued set of goods for this player. 
We continue to identify the next most valued bundle of minimal size for player $i$. 
We repeat this process until we identify the smallest subset of most valued items among unallocated items. 
If we allocate a bundle to $i$ after $t$ steps, we denote it as $B^t_i$; we then price the items in $B_i^t$ such that their total price is $b_i$. 
Note that all samples that contain $B^t_i$ have a price of $\geq b_i$, which guarantees that no player $i'>i$ can afford them.

We repeat this procedure for all players. At the end of the \textbf{\textit{for loop}} (Algorithm~\ref{algo:unit-demand-market}, line~\ref{algo-2:forloop}), we allocate any leftover goods to player $n$ for free, and assign any good which is not present in the sample set to player $1$ at a price of $0$. 

We first show that Algorithm~\ref{algo:unit-demand-market} outputs a consistent outcome. 


 \begin{algorithm}[tbh]
 \DontPrintSemicolon
 \SetAlgoLined
 \LinesNumbered
 \SetNlSty{}{}{:}
 \KwIn{A set of samples $\cal S$; player valuations over samples $v(\cal S)$ and budgets $b_1 > \dots>b_n$}
 $\Alloc \gets \emptyset$\;
 Allocate unobserved goods to player $1$ at price 0\;
 \For{$i \gets 1 \textbf{ to } n$}{\label{algo-2:forloop}
     $\cal S^1_i\gets \mathcal{S}$; $c \gets \False$; $t\gets 1$\;
     \While {$c=\False$} {\label{line7}
         $C^t_i \gets$ some set in $\argmax_{T \in \cal S^t_i} v_i(T)$\;
         $\cal L^t_i\gets\{T\in \cal S^t_i|v_i(T)=v_i(C^t_i)$\}\;
         $B^t_i \gets \bigcap_{T \in \cal L^t_i} T$\;
         $B^t_i \gets B^t_i \setminus \bigcup_{T \in \cal S \mid v_i(T)<v_i(C^t_i)}T$\;
         \If {$B^t_i\cap \Alloc\neq \emptyset $} {\label{line-12}
             $t \gets t + 1$; $\cal S^t_i \gets  \cal S^{t-1}_i\setminus{\cal L^t_i}$\;\label{lineSChange}
             }
         \Else {
             $c\gets \True$; $\Alloc\gets \Alloc\cup B^t_i$\;
             $A_i \gets B^t_i$ and price of each $g\in B^t_i$ is $\frac{b_i}{|B_i|}$\;
             }
         }
     
 }
 Allocate the leftover goods to player $n$ at price $0$\;
 \caption{Directly Learning Equilibria for Unit Demand Valuations}
\label{algo:unit-demand-market}
\end{algorithm}

\begin{theorem}\label{thm:unitdemand-consistency}
Algorithm \ref{algo:unit-demand-market} outputs a consistent market outcome.
\end{theorem}
\begin{proof}
Let the output of Algorithm \ref{algo:unit-demand-market} be $(\cal A,\vec p)$. 
Let us assume there is a sample $S\in \cal S$ such that for some player $i$, $v_i(S) > v_i(A_i)$; we need to show that $S \notin D_i(\vec p,b_i)$. 
Consider the items not available to player $i$ when it is their turn to select a bundle, referred to as $\Alloc$ in Algorithm \ref{algo:unit-demand-market}. If $S \cap \Alloc \ne \emptyset$, then $S$ must contain some previously allocated bundle $A_{i'}$, where $b_{i'} > b_i$; thus the price of $S$ is greater than $b_i$, and $S$ is not demanded by $i$. 
If $S$ can be allocated to player $i$ and is one of the most valued bundles at the time, player $i$ selects their bundle (i.e., $S \in \cal L^t_i$), then $B^t_i \subseteq S$; in particular, $v_i(S) = v_i(B^t_i)$. Otherwise, $v_i(B^t_i) > v_i(S)$ therefore $v_i(A_i)\geq v_i(S)$ and $i$ would not demand $S$.
\end{proof}

While Algorithm \ref{algo:unit-demand-market} outputs a consistent outcome, it offers an efficiency guarantee of $\frac{1}{\min\{n, k\}}$, under the minor assumption that player valuations are normalised with respect to their budget (i.e., $\max_{g \in G} v_i(\{g\}) = b_i$ for all $i \in N$).
\begin{restatable}{prop}{propefficiencyarbitrarydistribution}\label{prop:efficiencyarbitdistribution}
If for all $i \in N$, $\max_{g \in G} v_i(\{g\}) = b_i$, Algorithm \ref{algo:unit-demand-market} outputs an allocation $(\cal A, \vec{p})$ with $\ER_v(\cal A) \ge \frac{1}{\min\{n, k\}}$.
\end{restatable}
\begin{proof}
Irrespective of the samples and the distribution, Algorithm \ref{algo:unit-demand-market} ensures that the player with the highest budget (Player 1) gets their best possible allocation i.e. $v_1(A_1) = b_1$. 
When $n \le k$, Given the normalisation w.r.t. player budgets, the utility of the optimal equilibrium allocation has to be less than the sum of all the budgets i.e. $\sum_{i\in N}v_i(A^*_i) \le \sum_{i \in N}b_i \le nb_1$. When $n > k$, only $k$ players can get a good which means the upper bound on the utility of the optimal equilibrium allocation will be $\le kb_1$. From this the upper bound on the utility of the optimal equilibrium allocations will be $\le \min\{k, n\}b_1$
Therefore $\ER_v(\cal A)$ is 
\begin{align}\frac{\sum_{i\in N}v_i(A_i)}{\sum_{i\in N}v_i(A^*_i)} \ge \frac{b_1}{\sum_{i\in N}v_i(A^*_i)} \ge \frac{b_1}{\min\{n, k\}b_1} = \frac{1}{\min\{n,k\}}\end{align}
\end{proof}

Proposition \ref{prop:efficiencyarbitdistribution} offers a rather weak bound: the same efficiency ratio can be achieved by allocating all goods to the player with the highest budget. 
However, the bound is tight, and is an outcome of ``bad" distributions. We show that there exists sample sets for which no allocation can guarantee an efficiency greater than $\frac{1}{\min\{n, k\}}$. 
\begin{restatable}{theorem}{thmUDinfobounds}\label{thm:UD-info-bounds}
Let $\cal S, v(\cal S)$ be a set of samples along with its valuations; let $\cal V$ be the set of unit demand valuation profiles consistent with the set of samples and are budget normalised (i.e., $\max_{g \in G} v_i(\{g\}) = b_i$ for all players $i \in N$) and let $\cal B \subset \R_+^n$ be the set of all feasible budgets, i.e., the set of all budgets in $\mathbb{R}^n_{+}$ such that $b_1 > b_2 > \dots > b_n$. Then, we have 
\begin{align*}
    \min_{v \in \cal V} \max_{\cal A} \min_{\cal S \subseteq 2^G, \vec{b} \in \cal B} \ER_v(\cal A) \le \frac{1}{\min\{n, k\} - \delta}
\end{align*}
for any $\delta \in (0, n)$ where $\cal A$ is a consistent allocation with respect to the samples.
\end{restatable}
\begin{proof}
Consider a market with $n$ players and $k$ goods. Define a set of unit demand valuation function profiles $\cal V'$ as follows: each player has one good for which $v_i(g) = b_i$ and every other good has value $0$ for this player. We refer to the good with non-zero valuation as the favourite good of player $i$. Also, let no two players in the top $\min\{n, k\}$ players budget wise have the same favourite good. This set of valuations profiles satisfies our budget normalisation condition. 

Define the budget vector $\{b_1, b_2, \dots, b_n\}$ as follows: for every player $b_i = b_1 - \delta_i$ where $\delta_1 = 0$, $0 < \delta_2 < \delta_3 < \dots < \delta_n$ and $\sum_{i \in N} \delta_i = \delta b_1$. Let us call this vector of budgets $\vec{b'}$.

Now, suppose the only sample we have is the set of goods $G$ ($\cal S = \{G\}$) and $v_i(G) = b_i$ for all $i \in N$. This satisfies our budget normalisation condition and is consistent with all the valuation function profiles in $\cal V'$.

Note that for any valuation profile in $v \in \cal V'$, the best equilibrium allocation is where the top $\min\{n, k\}$ players get their favourite good. This allocation gives us a total value of 
\begin{align}
    \sum_{i \in N} v_i(A^*_i) = \sum_{i = 1}^{\min\{n, k\}} b_i  &= \min\{n, k\}b_1 - \sum_{i = 1}^{\min\{n, k\}} \delta_i \notag \\
    &\ge \min\{n, k\}b_1 - \delta b_1 \label{eqn:UD-info-bound-optimal-utility}
\end{align}

Suppose the allocation $\cal A$ allocates all the goods to one player. The maximum total welfare that $\cal A$ can guarantee is $b_1$ and this arises when the entire bundle is allocated to player $1$. Allocating the entire bundle to any other player will give us a strictly lower utility since all other players have a lower budget. This allocation gives us an efficiency 
\begin{align*}
    \min_{v \in \cal V} ER_v(\cal A) \le \frac{b_1}{\min\{n, k\}b_1 - \delta b_1} = \frac{1}{\min\{n, k\} - \delta}
\end{align*}
since the maximum utility that the optimal equilibrium allocation can obtain among all the valuation function profiles consistent with $\cal S$ is lower bounded by Equation \eqref{eqn:UD-info-bound-optimal-utility}.

If this is not the case and $\cal A$ allocates goods to more than one player, then we show that the maximum utility that $\cal A$ can guarantee is $0$. Let $\cal A$ allocate non-empty bundles to all players in $\{i_1, i_2, \dots, i_{n'}\}$. Therefore, the bundles $\{A_{i_1}, A_{i_2}, \dots, A_{i_{n'}}\}$ are non-empty. There exists a valuation function profile in $\cal V'$ such that the favourite good of $i_1$ is in $A_{i_2}$, the favourite good of $i_2$ is in $A_{i_3}$ and so on till finally, the favourite good of $i_{n'}$ is in $A_{i_1}$. All the players in $\{i_1, i_2, \dots, i_{n'}\}$ have different favourite goods here implying that all the players in $\{i_1, i_2, \dots, i_{n'}\}$ which are in the top $\min\{n, k\}$ players budget wise have different favourite goods. For those players in the top $\min\{n, k\}$ budget wise who are not allocated any goods, we can set their favourite good such that no two players in the top $\min\{n, k\}$ budget wise have the same favourite good. This valuation profile is in $\cal V'$ and is consistent with $\cal S$. The optimal equilibrium utility in this case is non-zero trivially and therefore the efficiency guaranteed by this allocation is $0$.

This means, given the set of samples and the set of budgets as defined above, we cannot guarantee an efficiency greater than $\frac{1}{\min\{n, k\} - \delta}$. This means that 
\begin{align*}
    \min_{v \in \cal V} \max_{\cal A} \min_{\cal S \subseteq 2^G, b \in \cal B} \ER_v(\cal A) \le 
    \min_{v \in \cal V} \max_{\cal A} \min_{\cal S = \{G\}, \vec{b} = \vec{b'}} \ER_v(\cal A) \le \frac{1}{\min\{n, k\} - \delta}
\end{align*}
\end{proof}

While Algorithm~\ref{algo:unit-demand-market} offers no reasonable welfare guarantees for general distributions, its performance guarantees improve significantly under certain distributional assumptions. 
Specifically, this holds true if $\cal D$ is a product distribution with a bounded probability of sampling each good. 
A product distribution $\cal D$ over $G$ is a distribution for which there exist values $p_1,\dots,p_k \in [0,1]$ such that for every $S \subseteq G$, $\Pr_{\cal D}[S] = \prod_{g_j \in S}p_j$. 
Product distributions offer more amenable welfare guarantees for two reasons: first, by definition, the presence of a particular good in the sample is independent of the presence of any other good (offering us a better chance of observing players' valuations for individual items); second, goods are sampled with non-zero probability (thus we observe all goods in some bundle with high probability). 
 Theorem~\ref{thm:prod_improved_bound} shows that Algorithm~\ref{algo:unit-demand-market} outputs a PAC equilibrium with an efficiency ratio of 1 with exponentially high probability, when samples are drawn i.i.d. from a product distribution; the proof requires that player preference orders over items are sufficiently distinct. 
 Before we prove Theorem~\ref{thm:prod_improved_bound}, we present two technical results -- Lemma~\ref{lem:best_eqm} and Lemma~\ref{lem:prodD-expbound} -- which we use to prove Theorem~\ref{thm:prod_improved_bound}. 
 \begin{restatable}{lemma}{lembestegm}\label{lem:best_eqm}
In unit demand markets with unequal budgets and strict preferences over items, any equilibrium allocation assigns player $i$ the best possible available good, i.e., $\{g_i^*\}$ equals $\argmax_{g\in \cal G_i} v_i(g)$ ($\cal G_1 = G$ and for $i > 1$, $\cal G_i = G \setminus \{\bigcup_{l=1}^{i-1}\cal \{g_l^*\}\}$). Moreover, all equilibria have the same social welfare $\sum_{i}v_i(g_i^*)$.
\end{restatable}
(Proof in Appendix~\ref{apdx:unit-demand})\\

In Lemma~\ref{lem:best_eqm}, we show that the social welfare for any equilibrium for unit demand players is unique and each player $i$ gets the good $g_i^*$. 
Therefore to show that the efficiency of Algorithm~\ref{algo:unit-demand-market} is 1 with high probability, it is sufficient to show that Algorithm~\ref{algo:unit-demand-market} assigns $g_i^*$ for all $i$ with high probability.

We now present Lemma~\ref{lem:prodD-expbound}, in which we prove that for any player $i$, if $\cal S^t_i$ at $t-$th iteration of the \textbf{\textit{while loop}} in Algorithm~\ref{algo:unit-demand-market} contains more than $k^2$ samples then the corresponding $B_i^t$ contains only the best available good for player $i$ in $\bigcup \limits_{S\in \cal S^t_i }S$, with high probability.    
\begin{restatable}{lemma}{lemprodDexpbound}\label{lem:prodD-expbound}
Suppose that $\cal D$ is a product distribution such that for all $g \in G$, $1 - \sqrt{2\me^{-1/k}-1} < \Pr_{S \in \cal D}(g \in S) < \frac{1}{2} + \frac{\sqrt{2\me^{-1/k}-1}}{2}$. If $|\cal S^t_i|\geq k^2$ (at the $t-$th iteration of the \textbf{\textit{while loop}} in Algorithm~\ref{algo:unit-demand-market} for player $i$), the corresponding $B_i^t$ equals $\{\hat{g}_i\}$ to player $i$ with at least $ 1 - \me^{-\frac k2}$ probability, where
\begin{equation*}
 \hat{g}_i \in \argmax\{v_i(\{g\}):g \in \bigcup \limits_{S\in \cal S^t_i }S\}
\end{equation*}
\end{restatable}
(Proof in Appendix~\ref{apdx:unit-demand})\\

We are now ready to prove Theorem~\ref{thm:prod_improved_bound}. We show that when we assume agent preferences sufficiently differ -- the good any agent gets in the optimal equilibrium allocation is in one of their top $\mathcal{O}(\log(\max\{n,k\}))$) goods -- Algorithm \ref{algo:unit-demand-market} is optimal with high probability.

\begin{restatable}{theorem}{thmprodimprovedbound}\label{thm:prod_improved_bound}
Suppose that $\cal D$ is a product distribution, such that $\Pr_{S \sim \cal D}[g\in S] \in [\alpha,\beta]$. Assume that for every agent $i$, $|\{g\in G: v_i(g) > v_i(g_i^*) \}|< \frac{\max\{\log n,\log k\}}{\log (\frac{1}{1-\beta})}$ \footnote{$g^*_i$ is defined as in Lemma~\ref{lem:best_eqm}: $\{g_i^*\}= \argmax_{g\in \cal G_i} v_i(g)$ ($\cal G_1 = G$ and for $i > 1$, $\cal G_i = G \setminus \{\bigcup_{l=1}^{i-1}\cal g_l^*\}$.}.

If $k>3$, $1 - \sqrt{2\me^{-1/k}-1}\leq \alpha$ and $\beta\leq\frac{1}{2} + \frac{\sqrt{2\me^{-1/k}-1}}{2}$, the output of Algorithm \ref{algo:unit-demand-market}, $(\cal A,\vec p)$, satisfies
\begin{equation*}
    \Pr[\ER_v(\cal A) = 1] \geq 1 - \frac{2n \max\{\log n,\log k\}}{\log
    \big(\frac{1}{ (1-\beta)}\big)}\me^{-\frac{k}{4}}
\end{equation*}
when $|\cal S| \ge \max \{k^2 n^2, k^4\}$
\end{restatable}
\begin{proof}
Let $\alpha = \min_{g \in G} \Pr_{S \sim \cal D}(g \in S)$ and $\beta = \max_{g \in G} \Pr_{S \sim \cal D}(g \in S)$. 
For simplicity, let $\Psi = \max\{\log n,\log k\}/\log (\frac{1}{1-\beta})$

We claim that with $\max \{k^4,n^2k^2\}$ samples, for all $i \in N$,
\begin{equation*}
   \Pr[\forall i' \leq i\; A_{i'} = \{g_{i'}^*\}]\geq 1- 2i \cdot \Psi \me^{-\frac{k}{4}} 
\end{equation*}

We prove our claim by induction on $i$. For player 1, the probability of not observing $g_1^*$ in $k^2$ samples is upper bounded by
\begin{align}
& \le (1-\alpha)^{k^2} \le (\sqrt{2\me^{-1/k}-1})^{k^2} = ({2^{}\me^{-1/k}-1})^{{k^2}/2} \label{eq:good-present} \\
& = 2^{k^2/2}\me^{-k/2} \bigg ( 1 - \frac{1}{2\me^{-1/k}}\bigg )^{{k^2}/2} \notag \\ 
& \leq  2^{k^2/2}\me^{-k/2} \bigg ( \frac{1}{2}\bigg )^{{k^2}/2} =  \me^{-k/2} \quad  (\text{ when } k \geq 3 ) \notag
\end{align}
Lemma~\ref{lem:prodD-expbound} shows that, once one of the samples contains $g_1^*$, we require $\geq k^2$ samples in $\cal S$ for Algorithm~\ref{algo:unit-demand-market} to allocate $\{ g_1^* \}$ to player 1 with a probability of at least $1-e^{-\frac{k}{2}}$. Taking a union bound, Algorithm \ref{algo:unit-demand-market} allocates $\{g_1^*\}$ to player 1 with a probability of at least $1 - 2e^{-k/2} \ge 1 - 2e^{-k/4}$.

By the inductive hypothesis, we assume that our claim is true for the first $i-1$ players. In other words, when we compute $A_i$, the set of already allocated items ($\Alloc$) is $\{g_1^*,\dots , g_{i-1}^*\}$ with high probability. 
For player $i$, let $g_i^*$ be her $t-$th preferred good. By Lemma~\ref{lem:best_eqm}, we know that $g_i^* \in \argmax_{g\in G\setminus \Alloc} v_i(g)$, and hence that player $i$'s $t-1$ most favorite goods are in $\Alloc$. 

Let $g_i^{t'}$ be the $t'$-th most preferred good for player $i$. Let $B_i^{t'}$ correspond to $B_i^{t'}$ at the $t'$ iteration of the \textbf{\textit{while loop}} of Algorithm~\ref{algo:unit-demand-market} (line \ref{line7}). Now the probability that Algorithm~\ref{algo:unit-demand-market} assigns good $\{g_i^*\}$ to player $i$ is at least

\begin{align}
    \Pr[A_i=\{g^*_i\}] &\ge \Pr\Big [\text{The set } B_i^{t'} = \{ g_i^{t'} \} \text{ for all } 1\leq t'\leq t\} \Big ] \label{eqn:rank_prob} \\
    &= \prod_{t' = 1}^{t} \Pr\big[B_i^{t'} = \{ g_i^{t'} \} \mid \forall l < t'\; B_i^{l} = \{ g_i^{l}  \}  \big]\notag    
\end{align}

The right hand side of Equation \eqref{eqn:rank_prob} refers to the event where the Algorithm first tries to allocate $i$'s favourite good to her and then seeing as the good is allocated, tries to allocate $i$'s second favourite good to her and so on till it tries to allocate the $t$'th favourite good to her. Seeing as this good is unallocated, the Algorithm allocates this good to $i$ resulting in $A_i = \{g_i^*\}$.

Consider the case when $k>n$.  
The probability that a sample set of size $k^2$ has at least one sample which does not contain $g_i^1,g_1^2,\ldots,g_i^{t'-1}$  (that will therefore remain in $\cal{S}^{t'}_i$) is at least $1 - ( 1 - {(1 - \beta)}^{t'-1} )^{k^2}$. Since we assume $t' \leq t \leq \max\{\log n,\log k\}/\log (\frac{1}{1-\beta}) $, the probability that no sample will remain in $\cal S_i^{t'}$ from a sample set of size at least $k^2$ is:

\begin{equation*}
  \leq ( 1 - {(1 - \beta)}^{t'}  )^{k^2} \leq  \bigg ( 1 - \frac{1}{k} \bigg )^{k^2} < \me^{-k} 
\end{equation*}

Therefore, with $|\cal S| \geq  k^4$ samples (which can be viewed as $k^2$ different sets of samples, each of size $k^2$), using the union bound, the probability that there are less than $k^2$ samples in $\cal S^{t'}_i$ is $\leq k^2 \me^{-k}\leq \me^{-k/4}$ (for $k>3$). When there are $k^2$ samples in $|\cal S_i^{t'}|$ and each good is present in the sample with a probability of at least $\alpha$, then (using Equation \eqref{eq:good-present}), the good $g_i^{t'}$ is present in at least one sample with a probability of at least $1 - e^{-k/2}$. Combined with Lemma \ref{lem:prodD-expbound}, this implies that for each $t' \leq t$;

\begin{align*}
     \Pr\big[B_i^{t'} = \{ g_i^{t'} \} &\mid  \forall l < t'\; B_i^{l} = \{ g_i^{l}  \} \big]\\& \geq
     (1 - \me^{-k/4})(1 - \me^{-k/2})\geq 1-2\me^{-k/4}
\end{align*}

Similarly when $n>k$, with $n^2k^2$ many samples, $\Pr\big[B_i^{t'} = \{ g_i^{t'} \} \mid  \forall l < t'\; B_i^{l} = \{ g_i^{l}  \}  \big] \geq 1 - 2\me^{-k/4}$

Therefore the probability that Algorithm~\ref{algo:unit-demand-market} assigns good $\{g_i^*\}$ to player $i$ (using Equation \eqref{eqn:rank_prob}) is 
\begin{align*}
    &\geq \prod_{t' = 1}^{t-1} \Pr\big[B_i^{t'} = \{ g_i^{t'} \} \mid \forall l < t'\; B_i^{l} = \{ g_i^{l}  \}  \big]\notag\\
    & \geq (1 - 2\me^{-\frac{k}{4}})^t \geq 1 - 2\Psi \me^{-\frac{k}{4}}     
\end{align*}

Now, using the union bound for the first $(i-1)$ players and the guarantees for player~$i$ we get,

\begin{align*}
    &\Pr[\forall i' \leq i\; \;A_{i'}=\{g_{i'}^*\} ]\geq 1- 2i\Psi\me^{-\frac{k}{4}}   
\end{align*}

Setting $i=n$ concludes the proof.
\end{proof}

As $\beta $ decreases (provided $\beta > 1 - \sqrt{2\me^{-1/k}-1}$), the condition in Theorem~\ref{thm:prod_improved_bound} on the difference between  players' preferences becomes less stringent. Moreover, if $k$ is large, the  exponential term in the probability guarantee dominates, and Algorithm~\ref{algo:unit-demand-market} is highly likely to output an efficient outcome.
However, if $\beta$ is smaller, the efficiency guarantee is less likely to hold. 
 Note that when $\beta = 1$, i.e., there is a good $g$ that appears in all samples, the performance of Algorithm \ref{algo:unit-demand-market} depends on  which player gets $g$. 
If the most preferred good for all players is $g$, Algorithm~\ref{algo:unit-demand-market} allocates $g$ to player $1$ and will not be able to continue: it is impossible to identify the second preferred good (and beyond). Therefore, Algorithm~\ref{algo:unit-demand-market} has an efficiency $\ge \frac{1}{\rho n}$ (for $\rho = \max_{g \in G}\frac{\max_{i \in N} v_i(g)}{\min_{i \in N}v_i(g)}$, the maximal ratio between the valuation of a single item by different agents) since we can only guarantee that the highest budget player will receive their optimal equilibrium allocation. 

We can generalize the efficiency bound in Theorem~\ref{thm:prod_improved_bound} for any preference order over the items for all players. 
We observe that with at least $\max(k^4,n^2k^2)$ samples, the first $\sim \max (\log k,\log n)$ players will be assigned $g_i^*$ with high probability. We show the efficiency guarantee for Algorithm~\ref{algo:unit-demand-market} for any preference order in Proposition~\ref{prop:proddist-unitD}, and its connection to the disparity in valuation functions between agents. 
\begin{restatable}{prop}{propproddistunitD}\label{prop:proddist-unitD}
If $\cal D$ is a product distribution such that for all $g_j \in G$, $1 - \sqrt{2\me^{-1/k}-1} < \Pr_{S \in \cal D}(g \in S) < \frac{1}{2} + \frac{\sqrt{2\me^{-1/k}-1}}{2}$ and $k>3$. Then, with exponentially high probability, Algorithm \ref{algo:unit-demand-market} allocates goods with an efficiency ratio $\ER_v(\cal A) \ge \frac{\log n}{\rho n \log \big( \frac{1}{1 - \beta} \big )} $ using a polynomial number of samples where $\rho = \max_{g \in G}\frac{\max_{i \in N} v_i(g)}{\min_{i \in N}v_i(g)}$ and $\beta = \max_{g_j \in G} \Pr_{S \in \cal D}(g \in S)$.
\end{restatable}

Furthermore, in Corollary \ref{cor:uniform-unitD}, we show the efficiency bound when each good is sampled i.i.d. w.p. $\frac12$.
\begin{restatable}{corollary}{coruniformunitD}\label{cor:uniform-unitD}
If the distribution $\cal D$ is uniform over the set $2^G$ and \(k > 3\), with exponentially high probability. Algorithm~\ref{algo:unit-demand-market} allocates goods with an efficiency $\ER_v(\cal A) > \frac{\log n}{\rho n}$ where $\rho = \max_{g \in G}\frac{\max_{i \in N} v_i(g)}{\min_{i \in N}v_i(g)}$ 
using a polynomial number of samples.
\end{restatable}
\begin{proof}
This extends directly from Proposition \ref{prop:proddist-unitD}. When $k \ge \frac{1}{\log 1.6}$, the lower bound and upper bound constraints on the probability of sampling each good in Proposition~\ref{prop:proddist-unitD} improve such that the product distribution where $\Pr_{S \in \cal D}({g \in S}) = \frac{1}{2} \quad \forall g \in G$ satisfies the constraints.
This means we can directly apply the results of Proposition~\ref{prop:proddist-unitD}: the uniform distribution is a product distribution with $ \alpha = \beta = \frac{1}{2}$ (which satisfy the constraints specified in Proposition~\ref{prop:proddist-unitD}). 
Thus, the efficiency ratio of the uniform distribution is at least $\frac{\log n}{\rho n \log 2} > \frac{\log n}{\rho n}$.
\end{proof}

\section{Single Minded Markets}\label{sec:singleminded}
In single minded markets, each player has a particular bundle of goods, $\Des_i \subseteq G$ they desire; every bundle that does not contain $\Des_i$ has no value, i.e.,
\begin{equation}
v_i(S) = 
\begin{cases}
    1 & \Des_i \subseteq S \\
    0 & \text{otherwise.}
\end{cases} \notag
\end{equation}
We show that a PAC underestimate for single-minded valuations can be efficiently learned, and an equilibrium for single-minded valuations can be efficiently computed. 
Therefore, using Proposition \ref{prop:lowerbound-equilibrium}, a PAC Equilibrium is computable in polynomial time.

\begin{restatable}{prop}{SMPacLearnable}\label{prop:SM-PAC-learnable}
The class of single minded valuation functions can be efficiently PAC learned, such that the learned valuation function weakly underestimates players' true valuations.
\end{restatable}
\begin{proof}
From a given set of samples $\cal S$, set 
$\bar{\Des}_i = \bigcap_{S \in \cal S : v_i(S)>0} S $. If, for a player $i \in N$, no sample has $v_i(S) > 0$, then set $\bar{\Des}_i = G$. This learned valuation is consistent and weakly lower than the actual valuations since $\Des_i \subseteq \bar{\Des}_i$ (i.e., a sample containing a set of items that is in $\Des_i$ but not all of $\bar{\Des}_i$ will be given a value 0 instead of 1).

The total number of possible valuation functions, i.e., size of the hypothesis class (denoted by $\cal H$) is $2^k$ (the number of possible choices for $\Des_i$). 
Thus, in order to PAC-learn $\Des_i$, we need a number of samples polynomial in $\frac{1}{\varepsilon}$, $\log \frac{1}{\delta}$ and $\log |\cal H| \in \mathcal{O}(k)$ (a classic learning result for finite hypothesis classes, see \shortciteA{anthony1999learning}).
\end{proof}

\shortciteA{branzeisingleminded} present an Algorithm to compute equilibria under equal budgets. We extend this Algorithm to settings where each player has a unique budget.

\begin{theorem}\label{thm:SM-equilibrium}
Algorithm \ref{algo:singleminded} outputs a market equilibrium for single minded players with all different budgets.
\end{theorem}

\begin{algorithm}[!htb]
\SetAlgoLined
\DontPrintSemicolon
\LinesNumbered
\SetNlSty{}{}{:}
\SetKwInOut{Notation}{Notation}
\SetKwFunction{SP}{SetPrice}
\SetKwFunction{UD}{UpdateDemand}
\SetKwProg{Fn}{Function}{:}{}
\KwIn{Player valuations $v$ and budgets $b_1 > \dots>b_n$}
\Notation{$b_i^*$ is the remaining budget for player $i$; prices are represented by $\vec p$.}
$\vec{p} \gets \vec{0}$; $\vec b^* \gets \{b_1, b_2, \dots, b_n\}$\; 
$B_i \gets \Des_i \quad \forall i \in N$\;
\For{each $g_j \in G$}{
    \If{$g_j$ is only demanded by one player}{
        Allocate $g_j$ to that player at $p_j \gets 0$\; \label{line:oneplayerdemand}
    }
    \ElseIf{$g_j$ is demanded by multiple players}{
        $p_j \gets \SP(g_j, \vec b^*, B)$ \; \label{line:multiplayerdemandbegin}
        Allocate $g_j$ to the player that can afford it at price $p_j$\;
        \UD$(B,\vec{p}, \vec b^*)$\; \label{line:multiplayerdemandend}
    }
}
Allocate all unallocated goods to player $n$ at price $0$\;
\Fn{\SP($g_j, \vec b^*, \cal B$)}{
$s \gets \argmax_{i \in N \land g_j \in B_i} b_i^*$\;
$t \gets \argmax_{i \in N \setminus {s} \land g_j \in D_i} b_i^*$\;
$p_j \gets b^*_t + \frac{b_s^* - b_t^*}{n^2}$\;
$b^*_s \gets b^*_s - p_j$ \;
\While{$\exists i \ne s: b^*_i = b^*_s$}{
    $b^*_s \gets b^*_s - \frac{b_s^* - b_t^*}{n^2} $, $p_j \gets p_j + \frac{b_s^* - b_t^*}{n^2}$\;
} 
\Return{$p_j$} \;
}
\Fn{\UD($\cal B, \vec{p}, \vec b^*$)}{
\For{$i \in N$}{
    \If{$(B_i \ne \emptyset) \land (\sum_{g \in B_i}p_g > b_i^*)$}{
        $B_i \gets \emptyset$ \;
    }
}
}
\caption{Competitive Equilibrium for Single Minded Valuations and Different Budgets}
\label{algo:singleminded}
\end{algorithm}

\begin{proof}
Algorithm \ref{algo:singleminded} iteratively allocates goods while keeping track of players' remaining budgets. 
If a good is demanded by multiple players, it is priced such that only one player can afford it, and allocated to that player. The \texttt{SetPrice} function ensures that no two players have the same remaining budget, by slightly increasing the price; this ensures that there are no ties when selecting the next player to allocate a good to. 

All players either get their desired set or a subset of their desired set if it is unaffordable. 
Thus the resulting allocation is an equilibrium: players who do not receive their desired set are not able to afford it. 
\end{proof}

The key difference between our approach and that of \shortciteA{branzeisingleminded} is how over-demanded goods are priced. \shortciteA{branzeisingleminded} assign the good to the player with the smallest desired set at a price equal to their budget. 
In our case, player budgets differ and therefore, ties cannot be broken by desired set size; rather, we instead break ties by remaining budgets.

Computing an equilibrium with total welfare at least $K$ has been shown to be NP-Complete by \shortciteA{branzei2015unitdemand} when players have equal budgets. In Theorem~\ref{thm:SM-Efficiency-NP}, we show this for our setting as well. 
\begin{restatable}{theorem}{thmSMEfficiencyNP}\label{thm:SM-Efficiency-NP}
It is NP-Complete to decide if a single minded market has an equilibrium with total welfare at least $K$
\end{restatable}
(Proof in Appendix~\ref{apdx:single-minded})\\

Theorem \ref{thm:SM-PAC-efficiency} shows that despite this, it is possible to compute a PAC equilibrium with an efficiency ratio $\ge \frac{1}{\min\{n, k\}}$. We now turn to establishing the efficeincy bounds of the algorithm. 

\begin{restatable}{lemma}{lemSMfullinfoefficiency}\label{lem:SM-fullinfo-efficiency}
Algorithm~\ref{algo:singleminded} assigns at least one player its desired set.
\end{restatable}
\begin{proof}
Let us first define a few terms which will help us with the proof. 
At any point in the algorithm, a player is {\em in the running} if they can afford their desired set, and is {\em eliminated} otherwise.
In Algorithm \ref{algo:singleminded}, all players start out in the running and get eliminated as the prices increase. 
Once a player gets eliminated, they will stay that way till the end of the algorithm since prices of goods never decrease and therefore will never be able to afford their bundle again; indeed, Algorithm \ref{algo:singleminded} sets players' demands to $\emptyset$ once they are eliminated.

We prove inductively that before and after any good is allocated, at least one player is still in the running. We assume w.l.o.g. that goods are considered in the order $g_1,\dots,g_k$.

For the first good $g_1$, all players are in the running before the allocation. If at most one player demands $g_1$ then prices remain $0$, and all players are still in the running; if multiple players demand $g_1$, then one player is allocated $g_1$ and the remaining players are eliminated (lines \ref{line:multiplayerdemandbegin}-\ref{line:multiplayerdemandend}); however, the player who received $g_1$ remains in the running.
In both cases, at least one player remains in the running after the good is allocated. 
Now, let us assume this is true for goods $g_1,\dots,g_{i-1}$. For $g_i$, there exists at least one player who is in the running before $g_i$ is allocated by the inductive hypothesis. If no more than one player demands $g_i$ then the price of $g_i$ is 0, and no player is eliminated; otherwise, all players who demand $g_i$ get eliminated, except for the player who receives $g_i$, who is still in the running.
Thus, there is at least one player in the running when we reach $g_k$. This player receives their desired set; otherwise, the allocation is not an equilibrium which contradicts Theorem \ref{thm:SM-equilibrium}. 
\end{proof}

\begin{theorem}\label{thm:SM-PAC-efficiency}
Let $(\cal A, \vec{p})$ be the output of Algorithm \ref{algo:singleminded} on valuations learned as in Proposition \ref{prop:SM-PAC-learnable}; then $\ER_v(\cal A) \ge \frac{1}{\min\{n, k\}}$. 
\end{theorem}
\begin{proof}
From Lemma \ref{lem:SM-fullinfo-efficiency}, we get that at least one player will receive his desired set. This desired set is the learned desired set which is a superset of the actual desired set (see Proposition \ref{prop:SM-PAC-learnable}). Therefore, the player who receives his learned desired set also receives his actual desired set. This means that the total welfare obtained is at least $1$. The maximum welfare any allocation can obtain is $\min\{n,k\}$ since the total number of players getting their desired set is upper bounded by $k$ and $n$. Thus, the efficiency of the computed PAC Equilibrium is $\ge \frac{1}{\min\{n, k\}}$
\end{proof}
Similar to unit demand markets, we show that our result in Theorem \ref{thm:SM-PAC-efficiency} is tight and no algorithm can guarantee a better efficiency.
\begin{restatable}{theorem}{thmSMinfobounds}\label{thm:SM-info-bounds}
Let $\cal S, v(\cal S)$ be a set of samples along with its valuations, $\cal V$ be the set of single minded valuation function profiles which are consistent with the set of samples and $\cal B \subset \R^n_{+}$ be the set of all feasible budgets, i.e., the set of all budgets in $\mathbb{R}^n_{+}$ such that $b_1 > b_2 > \dots > b_n$. Then, we have 
\begin{align*}
    \min_{v \in \cal V} \max_{\cal A} \min_{\cal S \subseteq 2^G, \vec{b} \in \cal B} \ER_v(\cal A) \le \frac{1}{\min\{n, k\}}
\end{align*}
where $\cal A$ is a consistent allocation with respect to the samples.
\end{restatable}

(Proof in Appendix~\ref{apdx:single-minded})\\

We now also show that our learned allocations are envy free. Note that this result about envy is stronger than the probabilistic result implied by the fact that the learned allocation is a PAC Equilibrium. It shows that, irrespective of what the samples are, the allocation output by Algorithm \ref{algo:singleminded} is guaranteed to be envy free.
\begin{restatable}{prop}{propSMenvyfree}\label{prop:SM-envyfree}
Let $(\cal A, \vec{p})$ be the output of Algorithm \ref{algo:singleminded} on valuations learned as in Proposition \ref{prop:SM-PAC-learnable}; then $(\cal A, \vec{p})$ is envy free.
\end{restatable}
\begin{proof}
Assume for contradiction that there exists a player $i$ who prefers the bundle $A_j$ to their own and can afford it. Also assume for now that $j \ne 1$. 

This means that $A_j$ contains the desired set $D_i$. Note that the learned desired set of player $i$ (say $\bar{D}_i$) cannot be equal to $D_i$. If it was, then player $i$ would be able to afford goods in $D_i$ when they are allocated, causing the pricing mechanism to make goods in $D_i$ unaffordable to player $i$ -- a contradiction.

In addition to this, if all the goods in $D_i$ were given to player $j$ and $j$ is only allocated goods in their learned desired set (since $j\ne1$), we have $D_i \subseteq \bar{D}_j$. This means that $v_i(\bar{D}_j) = 1$ and since $\bar{D}_i$ and $\bar{D}_j$ are formed by the intersection of samples in $\cal S$, all the samples in $\cal S$ which are used to compute $\bar{D}_j$ will be used when computing $\bar{D}_i$. This means $\bar{D}_i \subseteq \bar{D}_j$. 

Now, let $G'$ denote the set of goods in $\bar{D}_i \setminus D_i$ that have been allocated before any good in $D_i$ is allocated. $G'$ cannot be empty since if it was, the first good in $D_i$ which gets allocated will either be unaffordable to player $i$ or allocated to player $i$; both alternatives create a contradiction. If any good in $G'$ was allocated to any player other than $j$, $j$'s desired set would have been set to $\emptyset$ and no goods in $D_i$ would have been allocated to $j$ resulting in a contradiction. If, on the other hand, all the goods in  $G'$ were allocated to player $j$, it would be allocated at a price higher than $b_i$ since player $i$ would desire these goods as well resulting in the entire bundle being unaffordable to player $i$ which is also a contradiction.

When $j = 1$, if player $j$ was given any leftovers, the bundle would not be affordable to player $i$ since $b_1 > b_i$. If player $j$ was not given any leftovers, we use an analogous argument to show that $\bar{D}_i \subseteq \bar{D}_j$, leading to a conclusion similar to the one reached above.
\end{proof}

\section{Additive Markets}\label{sec:additive}
In additive markets, each player has additive valuations. The valuation of a bundle is equal to the sum of the valuations of every good in that bundle: $v_i(S) = \sum_{g \in S}v_i(\{g\})$.
While additive valuations are PAC-Learnable, we cannot use Proposition \ref{prop:lowerbound-equilibrium} to learn a PAC-Equilibrium since in a lot of cases, we cannot learn an underestimate of the valuations. This can be seen using Example \ref{ex:unitdemandindirect}. 

Although additive Fisher markets with indivisible goods have recently received a lot of attention, there are still many open questions regarding the efficient computation of a market clearing equilibrium. \shortciteA{babaioff2017additive} examine the specific case where there are only two players and \shortciteA{branzei2015unitdemand} show that it is computationally intractable to decide if a market has a competitive equilibrium when budgets are equal. This dearth of positive algorithmic results means that even if we could accurately learn the valuation of each good (which is not guaranteed and depends on the samples), we may not be able to compute an equilibrium in polynomial time. In this paper, we take a different approach and attempt to learn an equilibrium directly (using Theorem \ref{thm:consistency}); however, our outcome is not necessarily market clearing.

Our approach is described in Algorithm \ref{algo:additive-valuations}. The algorithm has three steps. First, we pre-process the samples to ensure that there are no proper subsets in the samples. This is done to ensure that no sample which is a superset of another sample is allocated. We can remove the supersets and replace them by the set difference between the superset and the subset: we can derive the value of this bundle under additive valuations, as executed in the function PreProcess. 

The second step allocates samples to players. To each player, the algorithm allocates the favourite sample among all the unallocated samples. Here, a sample is unallocated if no good in the sample has been allocated. It then prices each good equally such that the total price is equal to the budget of the player. 

The last step ensures consistency, it checks each of the original samples to see if a player prefers it over their own sample and can afford it. If there exists such a player, the algorithm proceeds to set the price of one of the goods in the sample to infinity to ensure that no player can afford it. This good is chosen as follows: if the sample has an unallocated good, then the unallocated good is chosen. If the sample does not have an unallocated good, the algorithm takes away a good from the sample which belonged to the player with the least budget and then sets its price to infinity. We refer to the act of setting the price of a good to infinity as burning a good. 

It is easy to see because of the third step that the algorithm is always consistent. It also worth noting that as long as we can underestimate the valuation in Line \ref{line:preprocess} in Algorithm \ref{algo:additive-valuations}, we will always end up with a consistent outcome. This means that this algorithm could be modified for any class of valuations to output a consistent outcome. 

\begin{algorithm}
\LinesNumbered
\DontPrintSemicolon
\KwIn{A set of samples $\cal S$; player valuations over samples $v(\cal S)$ and budgets $b_1 > \dots>b_n$}
\SetKwFunction{PP}{PreProcess}
\SetKwProg{Fn}{Function}{:}{}
$\cal S', \tilde{v}(\cal S') \gets $\PP{$\cal S, v(\cal S)$}; $\vec p \gets \vec 0$\;
\For{$i \leftarrow 1$ to $n$}{
    $B_i \gets \text{ some set in }\argmax_{T \in \cal S'}\tilde{v}_i(T)$\;
    Allocate $B_i$ to player $i$, i.e., $A_i \gets B_i, \tilde{v}_i(A_i) \gets \tilde{v}_i(B_i)$\;
    $p_g \gets \frac{b_i}{|B_i|} \; \forall g \in B_i $\;
    $\cal S' \gets \cal S' \setminus \bigcup_{S \in \cal S':S \cap B_i \ne \emptyset}S$\;
}

\While{$\exists i \in N, S \in \cal S$  s.t $\tilde{v}_i(A_i) < v_i(S) \land \sum_{g \in S} p_g \le b_i$}{
    \If{$\exists g \in S$ s.t. $g \notin \bigcup_{i \in N} A_i$}{
        $p_g \gets \infty$\;
    }
    \Else{
        $j \gets \argmin_{i \in N: S \cap A_i \ne \emptyset} b_i$\;
        $g \gets $ any good in $A_j \cap S$\; 
        $p_g \gets \infty$\;
        $A_j \gets A_j \setminus \{g\}$\;
        $p_g \gets \frac{b_j}{|A_j|} \; \forall g \in A_j $\;
        $\tilde{v}_j(A_j) \gets 0$\;
    }
}
If there is an unallocated sample, allocate it to the player with highest value for that sample at price $0$\;
Allocate all leftover goods to player 1 at price $0$\;
\Fn{\PP{$\cal S, v(\cal S)$}}{
    $\cal S' \gets \cal S$\;
    $\tilde{v}(\cal S') \gets v(\cal S)$\;
    \While{$\exists S', S'' \in \cal S'$ s.t. $S' \subsetneq S''$}{
        $\tilde{v}_i(S'') \gets \tilde{v}_i(S'') - \tilde{v}_i(S') \quad \forall i \in N$\; \label{line:preprocess}
        $S'' \gets S'' \setminus S'$\;
    }
    \Return $\cal S', \tilde{v}(\cal S')$\;
}
\caption{Consistent Allocation For Additive Markets}
\label{algo:additive-valuations}
\end{algorithm}

We now prove two efficiency bounds for our algorithm. These bounds hold only for additive valuations. To start with, we show that when the valuations are budget normalised, then the efficiency is inversely related to the number of goods. Before that, we show that no good in player 1's initially allocated sample gets taken away in Lemma \ref{lem:additive-player1}.
\begin{restatable}{lemma}{lemadditiveplayer}\label{lem:additive-player1}
In Algorithm \ref{algo:additive-valuations}, no good in player 1's initially allocated sample gets taken away.
\end{restatable}
\begin{proof}
The bundle that player 1 is allocated is either a sample or a subset of a sample. Let's call this parent sample $S$. The price of $S$ is at least $b_1$ which is unaffordable to all other players and therefore, no other player can demand it. Any sample intersecting with this sample (say $S'$) may be affordable to other players and the algorithm may burn a good from this sample. However, $S'$ will have a good $g \notin A_1$ because the PreProcess step ensures that no samples are allocated which are proper supersets of other samples. This good either remains unallocated or is allocated to a player with lower budget. Either way, it gets burnt first to ensure consistency leaving the goods in $A_1$ untouched.
\end{proof}

\begin{theorem}
When $\forall i \in N, \max_{g \in G} v_i(\{g\}) = b_i $, then Algorithm \ref{algo:additive-valuations} outputs an allocation with $\ER_v(\cal A) \ge \frac{1}{k}$
\end{theorem}
\begin{proof}
Algorithm \ref{algo:additive-valuations} always ensures the first player has a bundle with valuation at least $ b_1$.
If the first player's favourite good is not present in any sample, he receives at price 0 resulting in a valuation of at least $b_1$. 

If the first player's favourite good is present in the samples, then there exists a sample (with the first player's favourite good in it) which is valued at at least $b_1$ by the first player. Since the first player is allocated his favourite sample, he is given a bundle whose value is at least $ b_1$. By Lemma \ref{lem:additive-player1}, none of these goods are taken away and their final utility is at least $b_1$.

Since the largest amount of value a good can give a player is $b_1$. The total utility of any allocation is upper bounded by $kb_1$. This gives us the following bound:
\begin{align*}
    \ER_v(\cal A) = \frac{\sum_{i \in N} v_i(A_i)}{\sum_{i \in N}v_i(A_i^*)} \ge \frac{ v_1(A_1)}{\sum_{i \in N}v_i(A_i^*)} \ge \frac{b_1}{kb_1} = \frac{1}{k}
\end{align*}
\end{proof}
We now show that this bound is tight for general distributions.
\begin{restatable}{theorem}{thmadditiveinfobounds}\label{thm:additive-info-bounds}
Let $\cal S, v(\cal S)$ be a set of samples along with its valuations, $\cal V$ be the set of additive valuation function profiles which are consistent with the set of samples and are budget normalised, i.e., $\max_{g \in G} v_i(\{g\}) = b_i$ for all the players $i \in N$ and $\cal B \subset \R^n_{+}$ be the set of all feasible budgets, i.e., the set of all budgets in $\mathbb{R}^n_{+}$ such that $b_1 > b_2 > \dots > b_n$. Then, we have 
\begin{align*}
    \min_{v \in \cal V} \max_{\cal A} \min_{\cal S \subseteq 2^G, \vec{b} \in \cal B} \ER_v(\cal A) \le \frac{1}{k - \delta}
\end{align*}
for any $\delta \in (0, k)$ where $\cal A$ is a consistent allocation with respect to the samples.
\end{restatable}

(Proof in Appendix~\ref{apdx:additive})\\

Our next bound does not require the valuations to be normalised but imposes conditions on the samples and depends on the disparity in the valuations of goods.

In Proposition \ref{prop:additive-distribution}, we show that when samples are disjoint, the efficiency varies inversely with the disparity in valuations
\begin{restatable}{prop}{propadditivedistribution}\label{prop:additive-distribution}
When all the samples in $\cal S$ are pairwise disjoint, then Algorithm \ref{algo:additive-valuations} outputs an allocation with $\ER_v(\cal A) \ge \frac{1}{\rho}$ where $\rho = \max_{g \in G}\frac{\max_{i \in N}v_i(\{g\})}{\min_{i \in N}v_i(\{g\})}$
\end{restatable}
\begin{proof}
When all the samples are pairwise disjoint, all the goods are allocated and none of them are burnt. This is because, each player gets their favourite sample that has not been allocated yet. All samples that have been allocated to players with higher budgets are unaffordable to this player. 
Therefore, if $v_g$ is the amount of utility gained by the player who is allocated good $g$ in $\cal A$ and $v_g^*$ refers to the same for allocation $A^*$, then
\begin{align*}
    \ER_v(\cal A) = \frac{\sum_{g \in G} v_g}{\sum_{g \in G} v_g^*} \ge \frac{\sum_{g \in G} v_g}{\sum_{g \in G} \rho v_g} = \frac{1}{\rho}
\end{align*}
\end{proof}

\section{Submodular Markets}\label{sec:submodular}
In submodular markets, each player has monotone submodular valuations, i.e., each player's valuation function $v_i: 2^G \mapsto \mathbb{R}^+ \cup \{0\}$ satisfies the following three conditions:
\begin{enumerate}[(a)]
    \item $v_i(\emptyset) = 0$
    \item For any two $S, T \subseteq G$ such that $S \subseteq T$, $v_i(S) \le v_i(T)$.
    \item For any two $S, T \subseteq G$,
    \begin{align}
        v_i(S) + v_i(T) \ge v_i(S \cup T) + v_i(S \cap T) \label{eqn:submodular}
    \end{align}
\end{enumerate}
The class of monotone submodular valuations contains the class of additive valuations, as well as many others. 
This increase in complexity comes with an even greater dearth of positive algorithmic results. 
In addition to this, monotone submodular valuations cannot be efficiently PAC learned \shortcite{balcan2011submod}. 
So, we cannot use Proposition \ref{prop:lowerbound-equilibrium} to learn a PAC Equilibrium. 

We, instead, use a direct learning approach similar to that of additive markets but modify our algorithm slightly due to two reasons. First, the pre-process step that worked for additive valuations will not work for submodular valuations since we cannot accurately determine the value of the bundle that results when you remove a subset from a set. However, we can underestimate it using equation \eqref{eqn:submodular} as follows: given two sets $A, B \subseteq G$ such that $A \subseteq B$, then by substituting $S = B\setminus A$ and $T = A$ in equation \eqref{eqn:submodular} we get
\begin{align*}
    v_i(B \setminus A) \ge v_i(B) - v_i(A)
\end{align*}
Therefore, $v_i(B) - v_i(A)$ gives us an underestimate of $v_i(B \setminus A)$. Furthermore, the inequality does not change if we replace $v_i(B)$ with an underestimate of $v_i(B)$. 

Second, because we have to underestimate valuations, our efficiency guarantee may not hold. In order to prevent this, we modify our algorithm so that it can use extra information about the valuations. This is done using an additional input parameter $c_i$ for all $i \in N$ which specifies an underestimate of the value of the highest valued good, i.e., for all $i \in N$:
$c_i \le \max_{g \in G} v_i(\{g\})$.
Note that when there is no available information about the value of $c_i$, we can set $c_i = 0$.

The algorithm has been described in Algorithm \ref{algo:submod-valuations}. The algorithm has the same three steps as that of Algorithm \ref{algo:additive-valuations} but the first two steps are modified to work for submodular valuations. 

The PreProcess step removes any supersets from the set $\cal S$ and replaces them with the set difference between the superset and the subset. It also computes the set of goods which could have a value $\ge c_i$ and stores it in the set $F_i$. Note that $F_i$ is never empty and has a value of at least $c_i$ to player $i$. The following lemma proves it.
\begin{restatable}{lemma}{lemsubmodfi}\label{lem:submod-fi}
In the set $\cal F$ output by the PreProcess function of Algorithm \ref{algo:submod-valuations}, $F_i \ne \emptyset$ and $v_i(F_i) \ge c_i \quad \forall i \in N$.
\end{restatable}
\begin{proof}
There exists at least one good $g$ such that $v_i(\{g\}) \ge c_i$ by definition. Any sample with $g$ will have value at least $c_i$ by the monotone property.

Refer to the definition of $F_i$ in Line \ref{line:fi-def} in Algorithm \ref{algo:submod-valuations}. If the good $g$ is not present in any of the samples, then this good is included in $F_i$. If this good is present in any of the samples then this good will be present in $\bigcup_{S \in \cal S: v_i(S) \ge c_i}S$ and will not be present in $\bigcup_{S' \in \cal S: v_i(S') < c_i}S'$. Therefore, the good will be included in $F_i$.

Since all goods with $v_i(\{g\}) \ge c_i$ will be present in $F_i$ and there is at least one good such that $v_i(\{g\}) \ge c_i$, the lemma follows immediately.
\end{proof}

We then use this in the second step to give a player a bundle of value at least $c_i$ when no other sample guarantees a value of at least $c_i$. Of course, this is not applicable when an element of $F_i$ has been allocated to some other player.

The third step remains the same and ensures consistency since $\tilde{v}_i$ is an underestimate of $v_i$. So, if for any $S \in \cal S$, $v_i(S) > v_i(A_i)$, then, $ {v}_i(S) > \tilde{v}_i(A_i)$.

\begin{algorithm}
\LinesNumbered
\DontPrintSemicolon
\KwIn{A set of samples $\cal S$, player valuations for these samples $v(\cal S)$, budgets $b_1 > b_2 > \dots > b_n$ and $c_i \le \max_{g \in G}v_i(\{g\}) \forall i \in N$}
\SetKwFunction{PP}{PreProcess}
\SetKwProg{Fn}{Function}{:}{}
$\cal S', \tilde{v}(\cal S'), \cal F \gets$ \PP{$\cal S, v(\cal S)$}\;
\For{$i \leftarrow 1$ to $n$}{
    \If{$\tilde{v}_i(S') < c_i \; \forall S' \in \cal S' \land F_i \cap \bigcup_{j = 1}^{i-1}A_j = \emptyset$}{
        Allocate $F_i$ to player $i$, i.e., $A_i \gets F_i,  \tilde{v}_i({A_i}) \gets c_i$\;
    }
    \Else{
        $B_i \gets \text{ some set in }\argmax_{T \in \cal S'}\tilde{v}_i(T)$\;
        Allocate $B_i$ to player $i$, i.e., $A_i \gets B_i, \tilde{v}_i(A_i) \gets \tilde{v}_i(B_i)$\;
    }
    $p_g \gets \frac{b_i}{|A_i|} \; \forall g \in A_i $\;
    $\cal S' \gets \cal S' \setminus \bigcup_{S \in \cal S':S \cap A_i \ne \emptyset}S$\;
}

\While{$\exists i \in N, S \in \cal S$  s.t $\tilde{v}_i(A_i) < v_i(S) \land \sum_{g \in S} p_g \le b_i$}{
    \If{$\exists g \in S$ s.t. $g \notin \bigcup_{i \in N} A_i$}{
        $p_g \gets \infty$\;
    }
    \Else{
        $j \gets \argmin_{i \in N: S \cap A_i \ne \emptyset} b_i$\;
        $g \gets $ any good in $A_j \cap S$\; 
        $p_g \gets \infty$\;
        $A_j \gets A_j \setminus \{g\}$\;
        $\tilde{v}_j(A_j) \gets 0$\;
    }
}

Allocate all leftover goods to player 1 at price $0$\;
\Fn{\PP{$\cal S, v(\cal S)$}}{
    $\cal S' \gets \cal S$\;
    $\tilde{v}(\cal S') \gets v(\cal S)$\;
    \For{$S' \in \cal S'$}{
        \While{$\exists S \in \cal S$ s.t. $S \subsetneq S'$}{
            $\tilde{v}_i(S') \gets \tilde{v}_i(S') - v_i(S) \quad \forall i \in N$\;
            $S' \gets S' \setminus S$
        }
    }
    \If{$\exists S', S'' \in \cal S'$ s.t. $S' \subsetneq S''$}{
        Remove $S'$ from $\cal S'$
    }
    $\cal F = \{F_1, F_2, \dots, F_n\}$\;
    $F_i \gets \bigg(\bigcup_{S \in \cal S: v_i(S) \ge c_i}S \setminus \bigcup_{S' \in \cal S: v_i(S') < c_i}
    S'\bigg) \cup \bigg(G \setminus \bigcup_{S \in \cal S} S \bigg ) \quad \forall i \in N$\; \label{line:fi-def}
   
    \Return $\cal S', \tilde{v}(\cal S'), \cal F$\;
}
\caption{Submodular Markets Consistent Allocation}
\label{algo:submod-valuations}
\end{algorithm}

We now show that when valuations are budget normalised, then the algorithm has an efficiency of at least $\frac1k$. But before we do that, we show that even in this algorithm, none of player $1$'s goods get taken away. 

\begin{restatable}{lemma}{lemsubmodplayer}\label{lem:submod-player1}
In Algorithm \ref{algo:submod-valuations}, none of player 1's goods get taken away.
\end{restatable}
\begin{proof}
If the first player is not allocated $F_1$, then he is allocated a subset of a sample or a sample. Let's call the parent sample $S$. No player can afford this sample since it has a price of at least $ b_1$. However, some players may be able to afford and prefer a sample (say $S'$) which intersects with the allocated bundle. Since the pre-process step ensures that no supersets are allocated, there will be at least one good in $S'$ which is not allocated to the first player. This good either remains unallocated or is allocated to a player with lower budget. Either way, this good is burnt first to ensure consistency leaving the first player's allocated bundle intact.

If the first player is allocated $F_1$, this means that no set in $\cal S'$ can guarantee a value of at least $ c_1$. 

In such a scenario, any sample $S \in \cal S$ which contains a good $g \in F_1$ also contains a good $g' \notin F_1$. Assume for contradiction that this is not the case. Then there exists at least one sample which is a subset of $F_1$. Let $S$ be minimal such that $S \subseteq F_1$ and $S \in \cal S$. This means (from the way we define $F_i$(Line \ref{line:fi-def})), $v_1(S) \ge c_1$. 

$S$ must have a subset $S' \in \cal S$ such that $v_1(S') < c_1$. This is because if it does not have a subset, then $S$ will be in $\cal S'$ and $v_1(S) \ge c_1$ resulting in a contradiction (since player 1 will not need to be allocated $F_1$). Furthermore, since $S$ is minimal by our assumption, we have for any subset $T$ of $S$,$v_1(T) < c_1$. 

From the way we define $F_i$, since $v_1(S') < c_1$, we get $S' \cap F_1 = \emptyset$. Since $S' \subsetneq S$, there are certain elements in $S$ which are not present in $F_1$ which is a contradiction.

Now, since any sample $S \in \cal S$ which contains a good $g \in F_1$ also contains a good $g' \notin F_1$, the good $g'$ remains unallocated or belongs to a player with lower budget. Either way, it gets burnt first to ensure consistency leaving the goods in $F_1$ intact.
\end{proof}

This brings us to our final proof. When we have budget normalised valuations, then Algorithm \ref{algo:submod-valuations} gives us an allocation with efficiency at least $\frac1k$
\begin{restatable}{theorem}{thmsubmodefficiency}\label{thm:submod-efficiency}
When $\max_{g \in G}{v_i(\{g\})} = b_i$, then Algorithm \ref{algo:submod-valuations} outputs an allocation with efficiency $\ER_v(\cal A) \ge \frac1k$
\end{restatable}
\begin{proof}
When valuations are budget normalised, we can set $c_i = b_i$ for every player. If there exists a sample in $\cal S'$ with utility at least $b_1$, then the first player will get allocated a sample with utility at least $b_1$. If not, then the first player will be allocated $F_1$ which has value at least $b_1$. Using Lemma \ref{lem:submod-player1}, none of these goods are taken away from the first player and so his final utility will be at least $b_1$.

Since the maximum utility a good can give a player is upper bounded by $b_1$, the utility of the optimal equilibrium is upper bounded by $kb_1$. This gives us the following efficiency bound:
\begin{align*}
    \ER_v(\cal A) = \frac{\sum_{i \in N} v_i(A_i)}{\sum_{i \in N}v_i(A_i^*)} \ge \frac{ v_1(A_1)}{\sum_{i \in N}v_i(A_i^*)} \ge \frac{b_1}{kb_1} = \frac{1}{k}
\end{align*}
\end{proof}

Since additive valuations are a subset of monotone submodular valuations, Theorem \ref{thm:additive-info-bounds} applies in this case as well. This means the bound in Theorem \ref{thm:submod-efficiency} is tight.

\section{Experimental Evaluation}\label{sec:expts}
Theorems \ref{thm:UD-info-bounds}, \ref{thm:SM-info-bounds} and \ref{thm:additive-info-bounds} show that it is impossible to prove strong efficiency guarantees for our algorithms (or any algorithms that solve this problem). This is mainly due to the possibility of ``bad" distributions which no algorithm can give good efficiency guarantees for. Therefore, to evaluate our algorithms, we examine realistic markets and datasets that our algorithms could be used on.

In order to test our approach on data we would, ideally, require a dataset consisting of bundles of goods, and users' valuations over these bundles. In addition, we would require a dataset that offers us access to agents' {\em true} valuations, so that we have a baseline for comparison. To our knowledge, there are no such publicly available datasets; that said, it is not unreasonable to assume that companies who collect market data (e.g. consumer analytics, or large-scale movie recommendation systems) have access to such datasets. In order to simulate a dataset that meets our specifications, we use the MovieLens dataset \shortcite{movielens} to simulate a market environment. The MovieLens dataset contains users' (ordinal) rankings over movies. In our setting, the movies serve as goods and the users serve as players. We use this dataset to model a setting where a fixed number of movie screenings is offered to a group via a personalized assignment algorithm. When we vary the number of rooms that can screen movies and the number of different (non-intersecting) time slots that we can use to screen the movies, we get different classes of valuation functions that we can evaluate our algorithms on. 

When there are infinitely many screens but only one time slot, player preferences follow unit demand valuations (Section \ref{subsec:unit-demand-expts}). This is because players can only watch one movie since all of them will be screened at the same time; so the value of a bundle of movies will be equal to the value of the best movie in the bundle. 

When there are infinitely many time slots and only one screen, player preferences follow additive valuations (Section \ref{subsec:additive-expts}). This is because players can watch all the movies in their allocated bundle; therefore, the value of a bundle is equal to the sum of the values of every good in the bundle.

When there are a finite number of screens and a finite number of time slots, player preferences follow submodular valuations (Section \ref{subsec:submodular-expts}). This is because, there is a decreasing marginal utility for every good that you add to a bundle since you cannot watch two movies in the same time slot. More specifically, in Section \ref{subsec:submodular-expts} we consider a setting when there are $10$ time slots and an equal number of movies screened in each time slot. The movies in each time slot are randomly chosen right at the beginning. To make the submodular valuations even more non-trivial, we assume that all the players have a threshold value $\Th$ indicating that after watching $\Th$ movies, they cannot gain any value from additional movies. 

We study three different markets based on different levels of supply and demand: 
\begin{inparaenum}[(a)]
\item a sellers' market where the number of players exceeds the number of goods,
\item a buyers' market where the number of goods exceeds the number of players and
\item a balanced market where the number of goods and the number of players are the same.
\end{inparaenum}
These markets have the sizes: 
\begin{inparaenum}[(a)]
\item $n = 50, k = 30$, \label{market50}
\item $n = 30, k = 50$ and \label{market30}
\item $n = 40, k = 40$ respectively.\label{market40}
\end{inparaenum}
The choice of size for these markets are based on two factors. First, we would like to evaluate as many different kinds of markets as possible in terms of demand and supply, i.e., settings where the number of players is greater, smaller and equal to the number of goods. Second, even though we would like to use larger markets, a lack of efficient baselines for additive and submodular markets makes it infeasible to use markets with a very large number of goods and players ($> 100$).

Following \shortciteA{budish2011}, we assume that players have almost equal budgets: for all $i \in N$, $5 < b_i < 6$. More precisely, we set $b_i = 5 + \texttt{Unif}(0,1)$. We slightly perturb player valuations such that no two valuations are equal (for tiebreaking); perturbed valuations always respect the original rankings, $v_i(j) = r_{ij} + \texttt{Unif}(0,0.1)$, where $r_{ij}$ is the rating given for the movie $j$ by player/user $i$. We also normalize valuations such that $b_i = \max_{g \in G} v_i(\{g\})$. In order to simulate a dataset of bundles, we sample bundles of movies from:
\begin{inparaenum}[(a)]
\item uniform product distributions, with goods sampled w.p. $\frac{1}{2}$;
\item uniform distributions over bundles of constant size $s$, $s \in \{1, 3, 5, 10\}$.
\end{inparaenum}
For each market, for each distribution, we run our algorithm on $5120$ randomly generated samples slowly increasing the number of samples our algorithm uses from $5$ to $5120$. We repeat this procedure a $100$ times and plot graphs on a semilog scale. For each allocation outputted by our algorithm, we also check loss by sampling $1000$ samples from the same distribution and computing {\em empirical loss} using Equation \eqref{eq:lossdefinition}. 
All the graphs we plot have error bars plotted along with the lines but due to the large number of iterations, in most graphs these error bars are not visible. 


\subsection{Unit Demand Markets}\label{subsec:unit-demand-expts}
We compare the two approaches we discuss: indirect learning (Algorithm \ref{algo:unit-demand-indirect}) and direct learning (Algorithm \ref{algo:unit-demand-market}). 
We also compute optimal market outcomes w.r.t. the true preferences, which serve as our baseline. We run both learning algorithms for the different markets and sampling distributions discussed above and evaluate the algorithms in terms of their welfare and their market inconsistency, as measured by their empirical loss.

\subsubsection{Sellers' Market}\label{subsub:UD-sellers}
In the sellers' market, the direct learning approach almost always outperforms the indirect learning approach either by converging to the optimal welfare faster or by doing strictly better than indirect learning when the number of samples is high (see Figure \ref{fig:unitdemand-case1}). 
When the number of samples is low, the indirectly learned outcome allocates one good to each player and therefore allocates goods to more players. This results in a higher utility as compared to the direct learning approach which allocates larger bundles to players when the number of samples are low. These large bundles arise due to the fact that the direct learning approach tries to find the smallest bundle which is sure to contain the highest valued unallocated good. 
As the number of samples increases, the direct learning algorithm learns more, causing the size of the bundles allocated by the direct learning algorithm to decrease; this results in a sharp improvement that allows the direct learning approach to outperform the indirect learning approach even when the size of the samples are high (greater than $5$). 
When each sample is large, as we iterate through the set of players and have only a few goods left to allocate, these goods will very likely not be the best good in any bundle. 
Take for example, a dataset where all samples have size $10$: when the data corresponds to unit demand valuations, we will not be able to make any judgements on the least preferred $9$ goods of any player regardless of what samples we have. 
When faced with this problem, indirect learning allocates these leftover goods to buyers arbitrarily even though it cannot ascertain the exact value of the good it allocates. 
The direct learning approach on the other hand, iterates through the players until it finds players who have one of the unallocated goods as the highest valued good in a sample and then allocates this good to them --- even though these players may have a low budget --- and subsequently increasing welfare. 
Note that while this allocation is, of course, a PAC equilibrium, in this setting it is not an actual equilibrium. Therefore, our allocations might cause envy between the players, since in most allocations higher budget players will not get a good while lower budget players do.
However, it will result in allocations which might exceed the welfare of the optimal equilibrium (see Figure \ref{fig:unitdemand-case1-dc5}).

\begin{figure}[ht]
\subfloat[Sample size 1\label{fig:unitdemand-case1-dc0}]
  {\includegraphics[width=.25\linewidth]{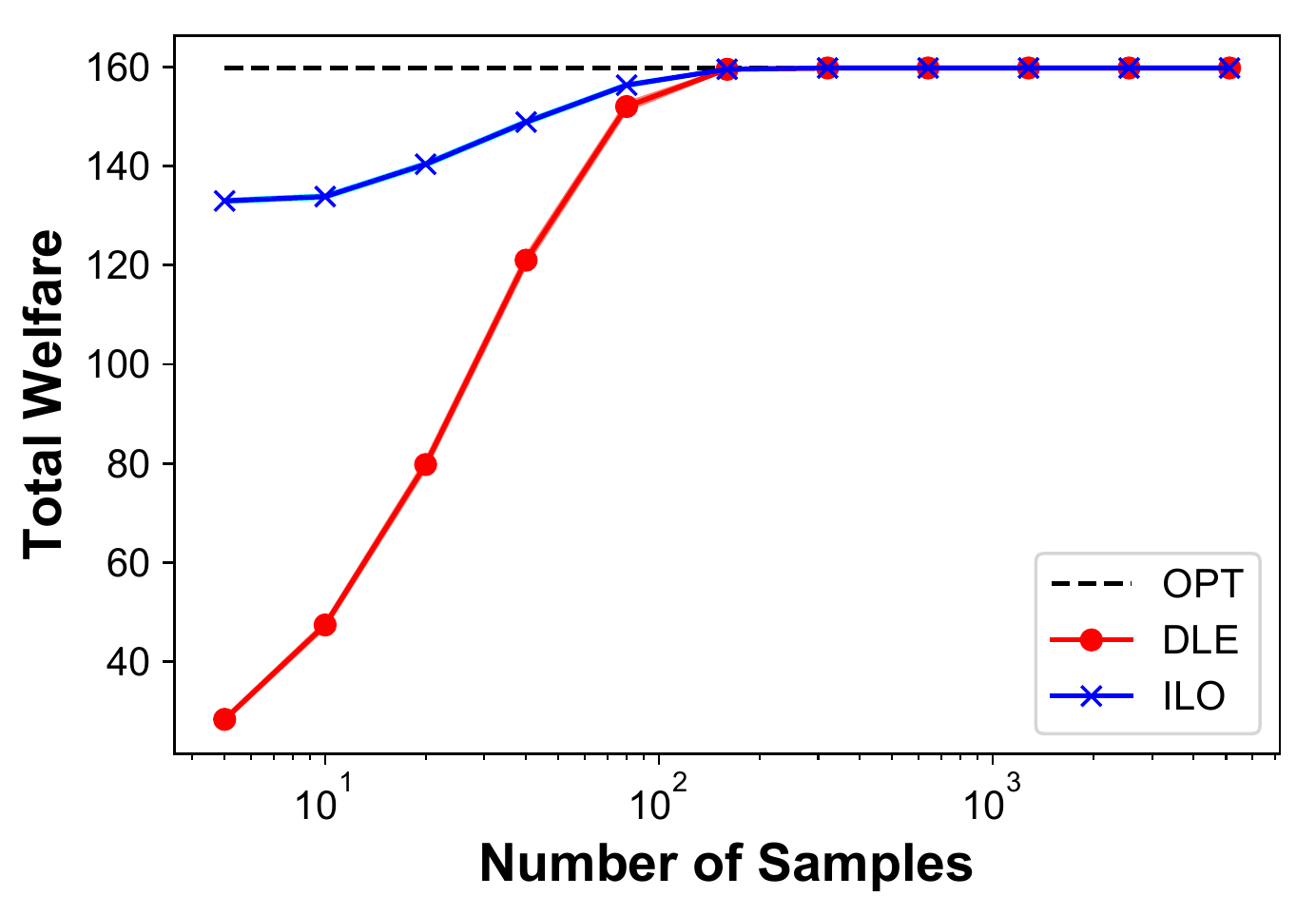}}\hfill
\subfloat[Sample size 5\label{fig:unitdemand-case1-dc2}]
  {\includegraphics[width=.25\linewidth]{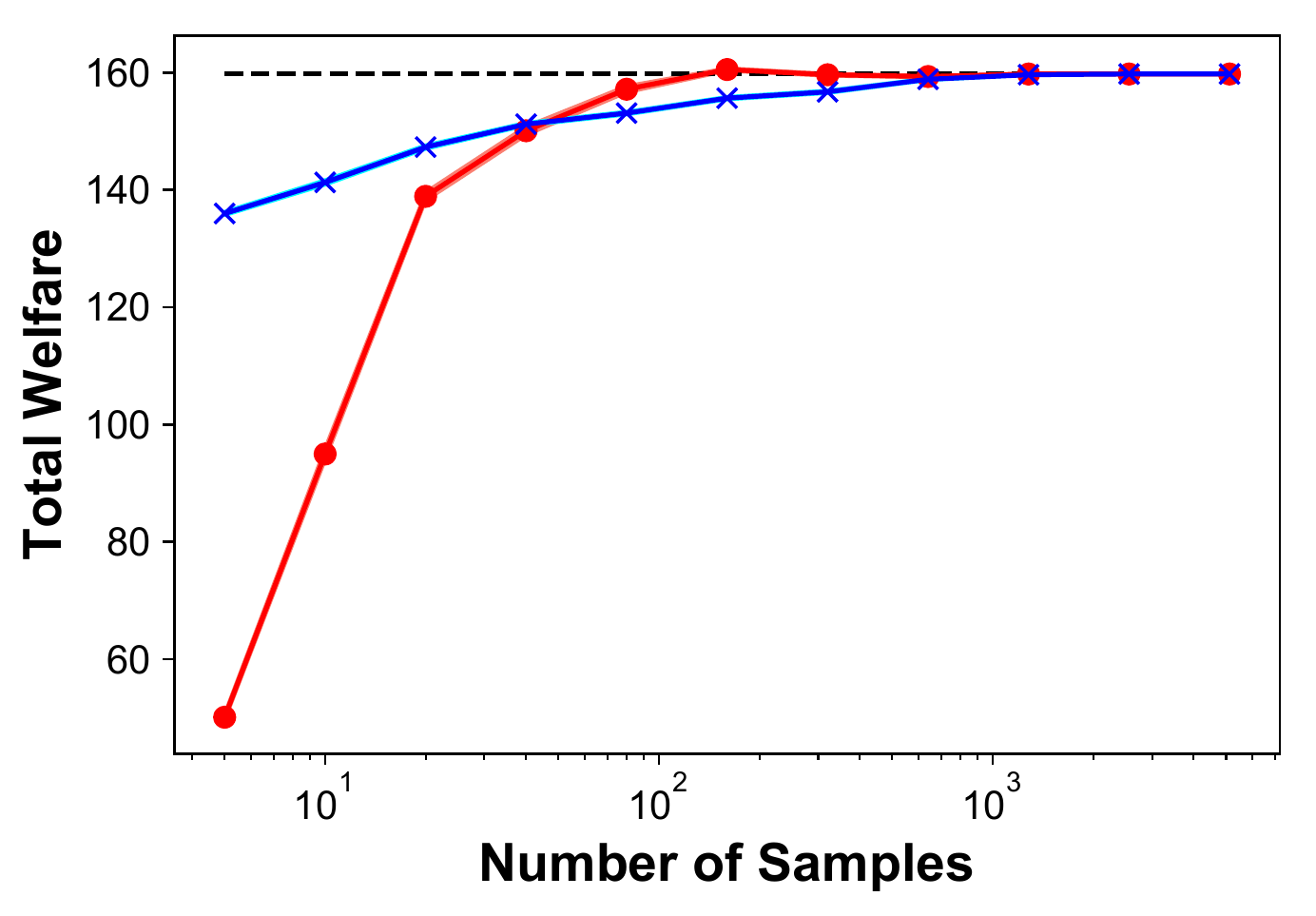}} \hfill
  \subfloat[Sample size 10\label{fig:unitdemand-case1-dc4}]
  {\includegraphics[width=.25\linewidth]{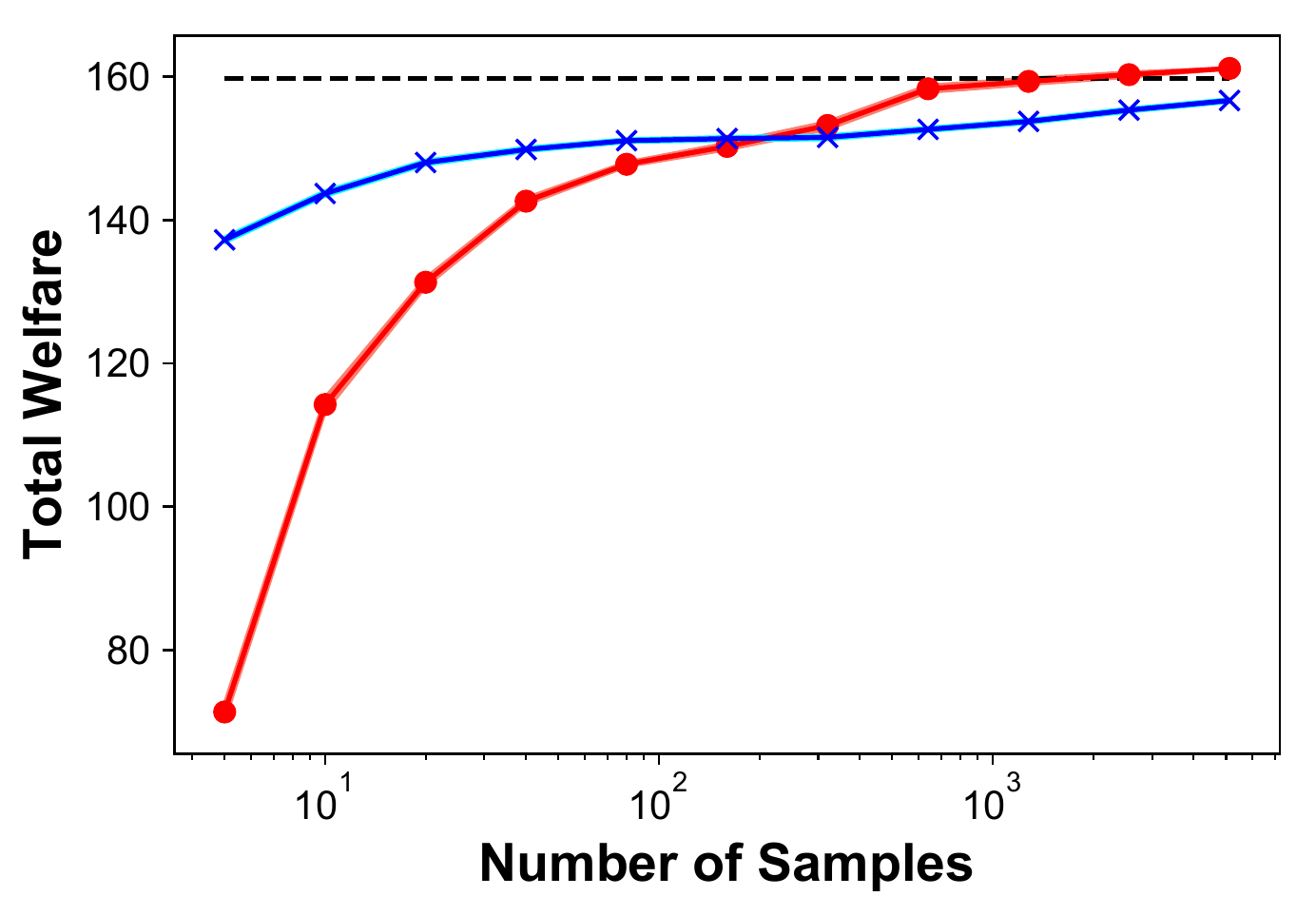}} \hfill
\subfloat[Uniform Product Distribution\label{fig:unitdemand-case1-dc5}]
  {\includegraphics[width=.25\linewidth]{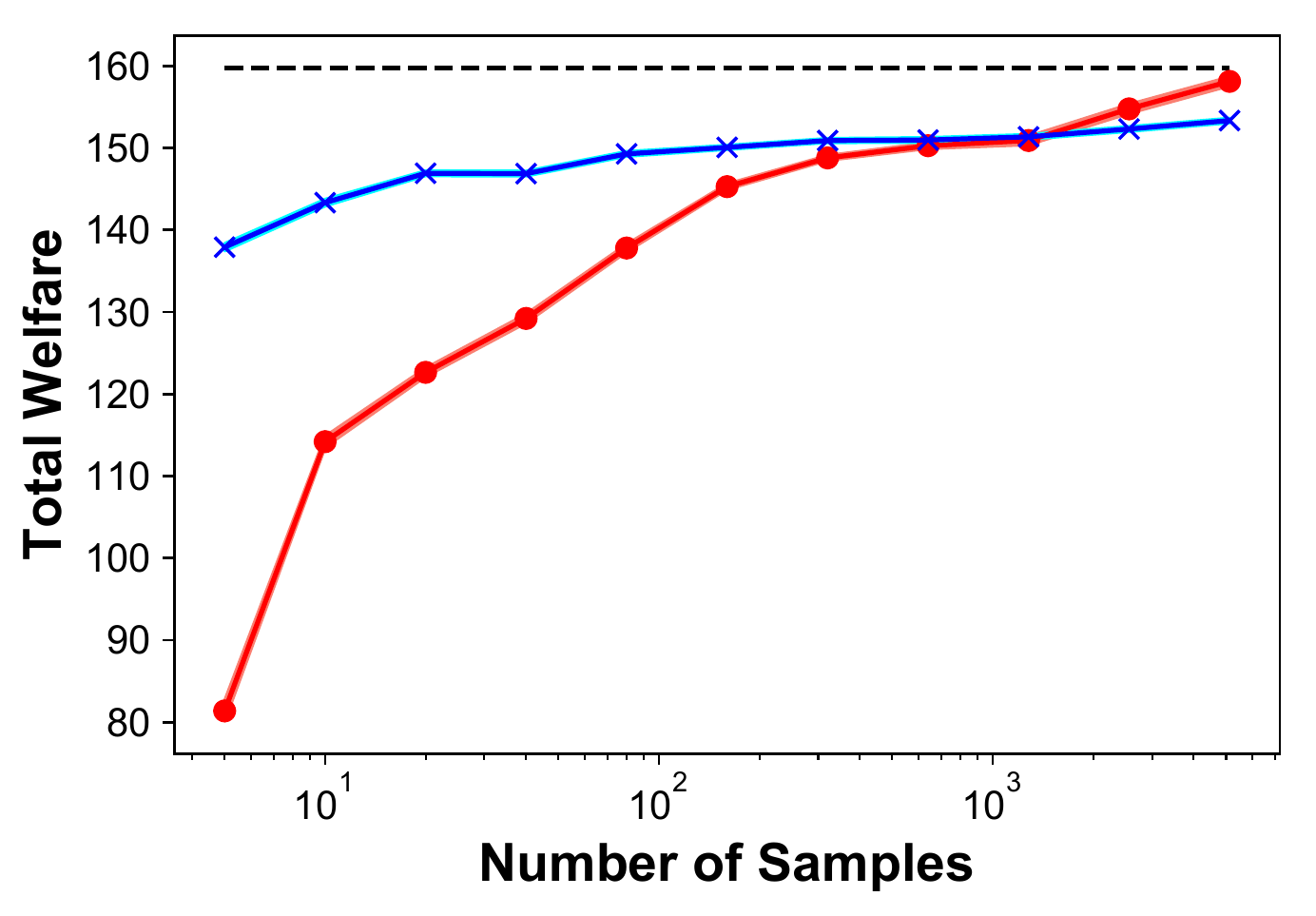}}\hfill
  
\caption{Unit Demand Markets (Sellers' Market): Total welfare vs. \# of samples for various distributions when number of players is greater than the number of goods. \textbf{DLE} represents  directly learned equilibria and \textbf{ILO} represents  indirectly learned outcomes. \textbf{OPT} represents the welfare-maximizing equilibrium allocation.}
\label{fig:unitdemand-case1}
\end{figure}

\subsubsection{Buyers' Market}\label{subsub:UD-buyers}
In the buyers' market, where the number of goods exceeds the number of players, we observe that both algorithms converge to the optimal equilibrium welfare in all the markets for all distributions (see Figure \ref{fig:unitdemand-case2}). When the number of goods is higher than the number of players, with enough samples, both algorithms learn the $n$-highest valued goods for each player and then allocate one of these goods to the player depending on their budget, i.e., depending on what good they would have received in the optimal equilibrium allocation. 
We also observe that in a lot of cases (see Figure \ref{fig:unitdemand-case2-dc2} and \ref{fig:unitdemand-case2-dc4}), the direct learning approach converges to the optimal equilibrium welfare slightly faster than the indirect learning approach. This is mainly because, when we do not have enough samples to learn the valuation function and we have two goods that could be player $i$'s favorite good but we do not know which one it is, indirect learning picks a good from these two goods arbitrarily and allocates it to $i$ whereas the direct learning approach allocates both goods to $i$, guaranteeing a higher utility for $i$. 
When the number of goods exceeds the number of buyers, we do not run out of goods by doing this and so, we obtain a higher welfare allocation.

\begin{figure}[ht]
\subfloat[Sample size 1\label{fig:unitdemand-case2-dc0}]
  {\includegraphics[width=.25\linewidth]{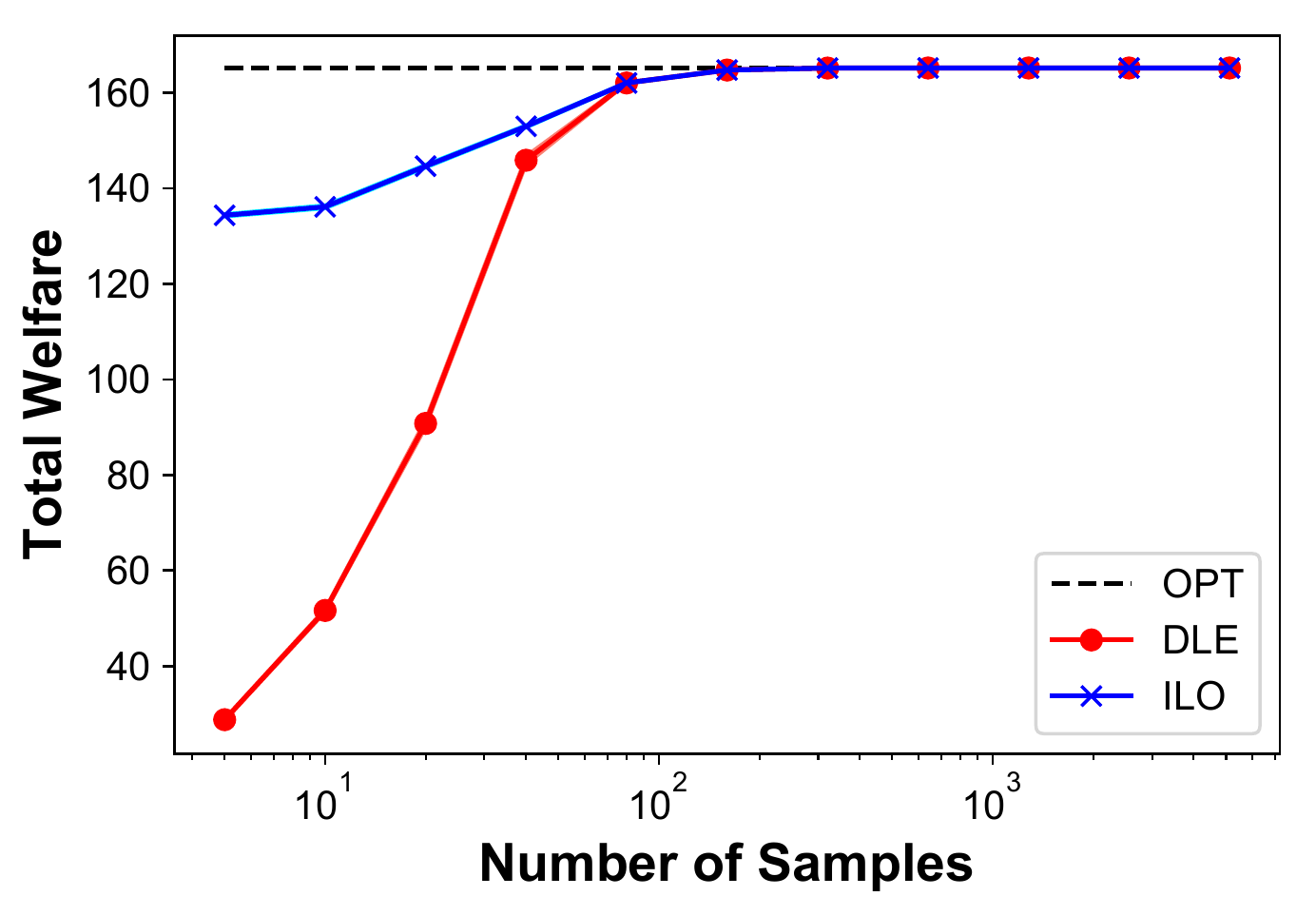}}\hfill
\subfloat[Sample size 5\label{fig:unitdemand-case2-dc2}]
  {\includegraphics[width=.25\linewidth]{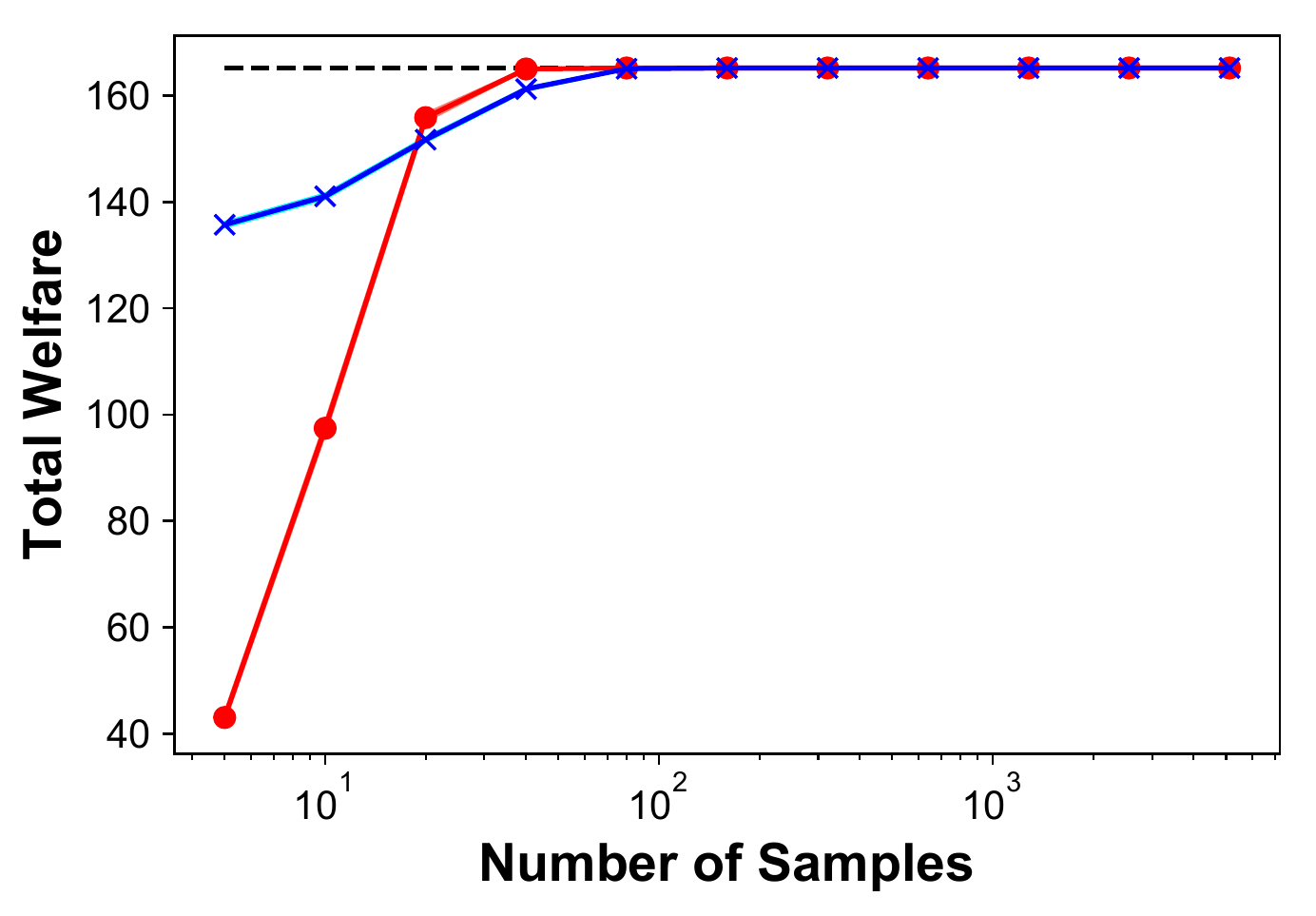}} \hfill
  \subfloat[Sample size 10\label{fig:unitdemand-case2-dc4}]
  {\includegraphics[width=.25\linewidth]{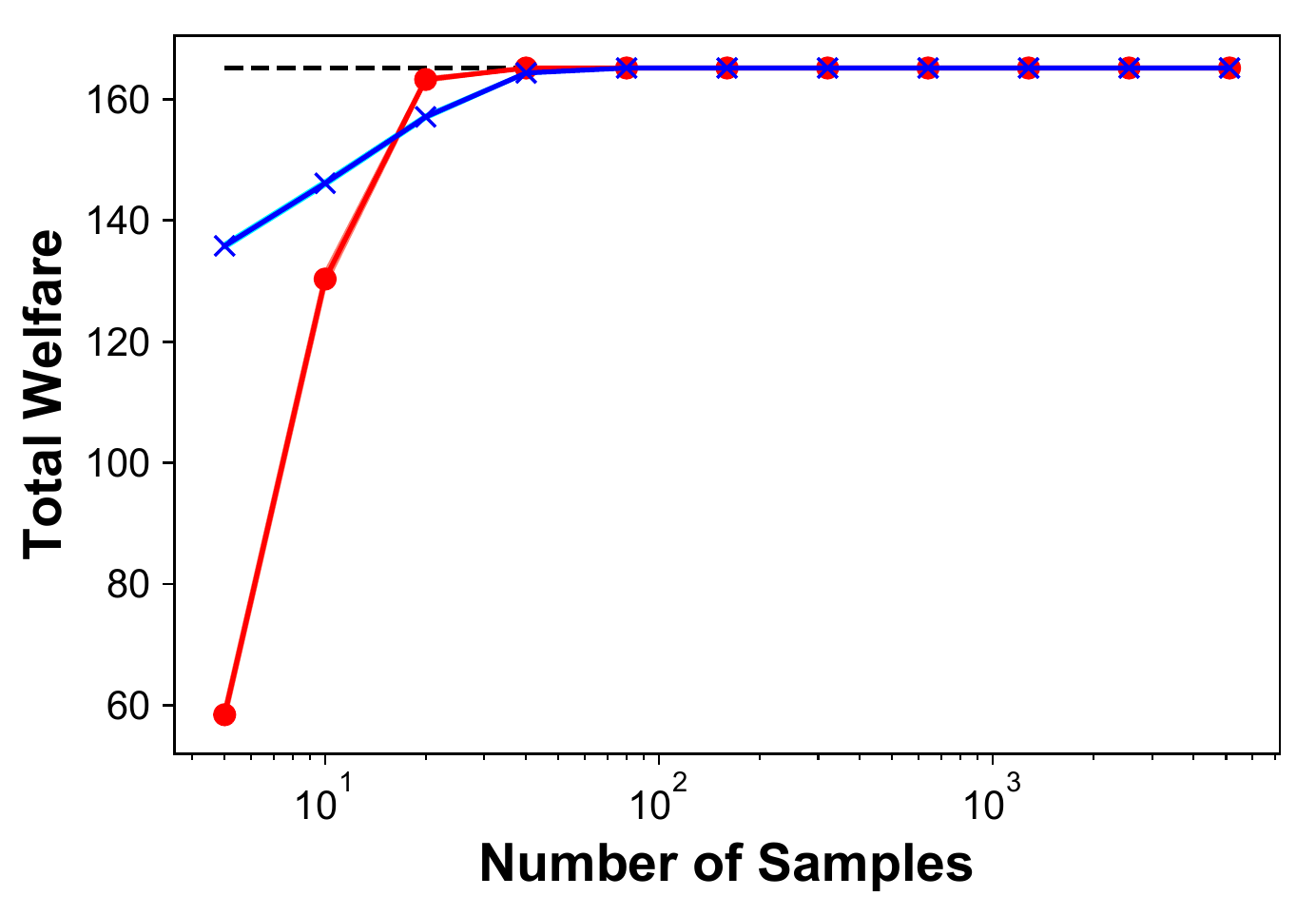}} \hfill
\subfloat[Uniform Product Distribution\label{fig:unitdemand-case2-dc5}]
  {\includegraphics[width=.25\linewidth]{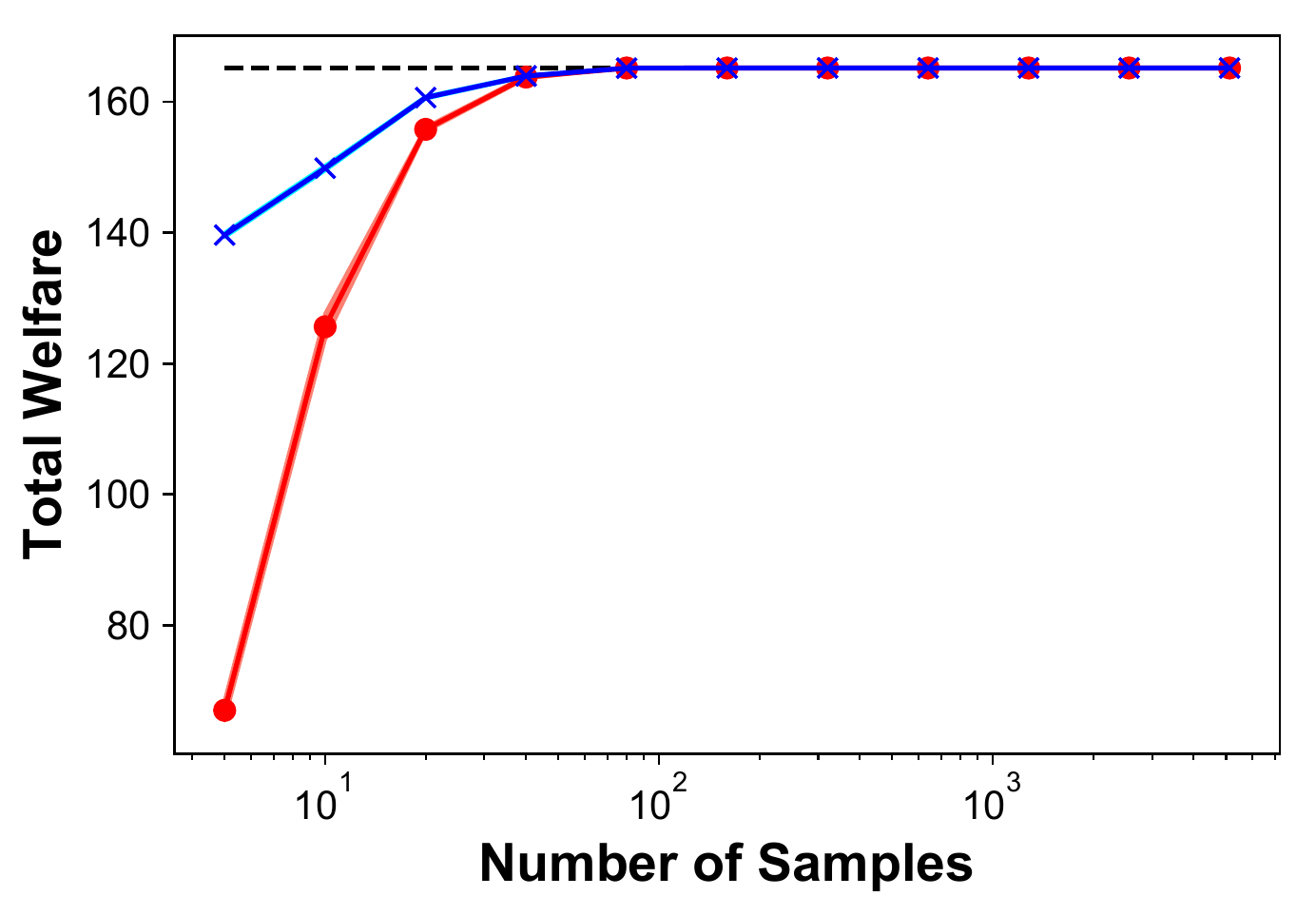}}\hfill
  
\caption{Unit Demand Markets (Buyers' Market): Total welfare vs. \# of samples for various distributions when number of players is less than the number of goods. \textbf{DLE} represents  directly learned equilibria and \textbf{ILO} represents  indirectly learned outcomes. \textbf{OPT} represents the welfare-maximizing equilibrium allocation.}
\label{fig:unitdemand-case2}
\end{figure}

\subsubsection{Balanced Market}\label{subsub:UD-balanced}
In the balanced market, where the number of goods equals the number of buyers, we find that the direct learning approach is outperformed by the indirect learning approach. This is mainly because in this case, the best strategy is to allocate one good to each player and the direct learning algorithm does not do this. It tries to allocate bundles to players to guarantee that they get a high valued bundle; when it cannot find the highest valued unallocated good of any bundle, the algorithm does not allocate anything. This particularly hurts the algorithm when there is not much information in each sample e.g. when the samples are large (see Figure \ref{fig:unitdemand-case3-dc5}).
Indirect learning does exactly what is required and tries to allocates the single best possible good to each buyer. If it cannot find the best good, it allocates an arbitrary good which may not be the best possible good but still results in a fairly high welfare allocation.

\begin{figure}[ht]
\subfloat[Sample size 1\label{fig:unitdemand-case3-dc0}]
  {\includegraphics[width=.25\linewidth]{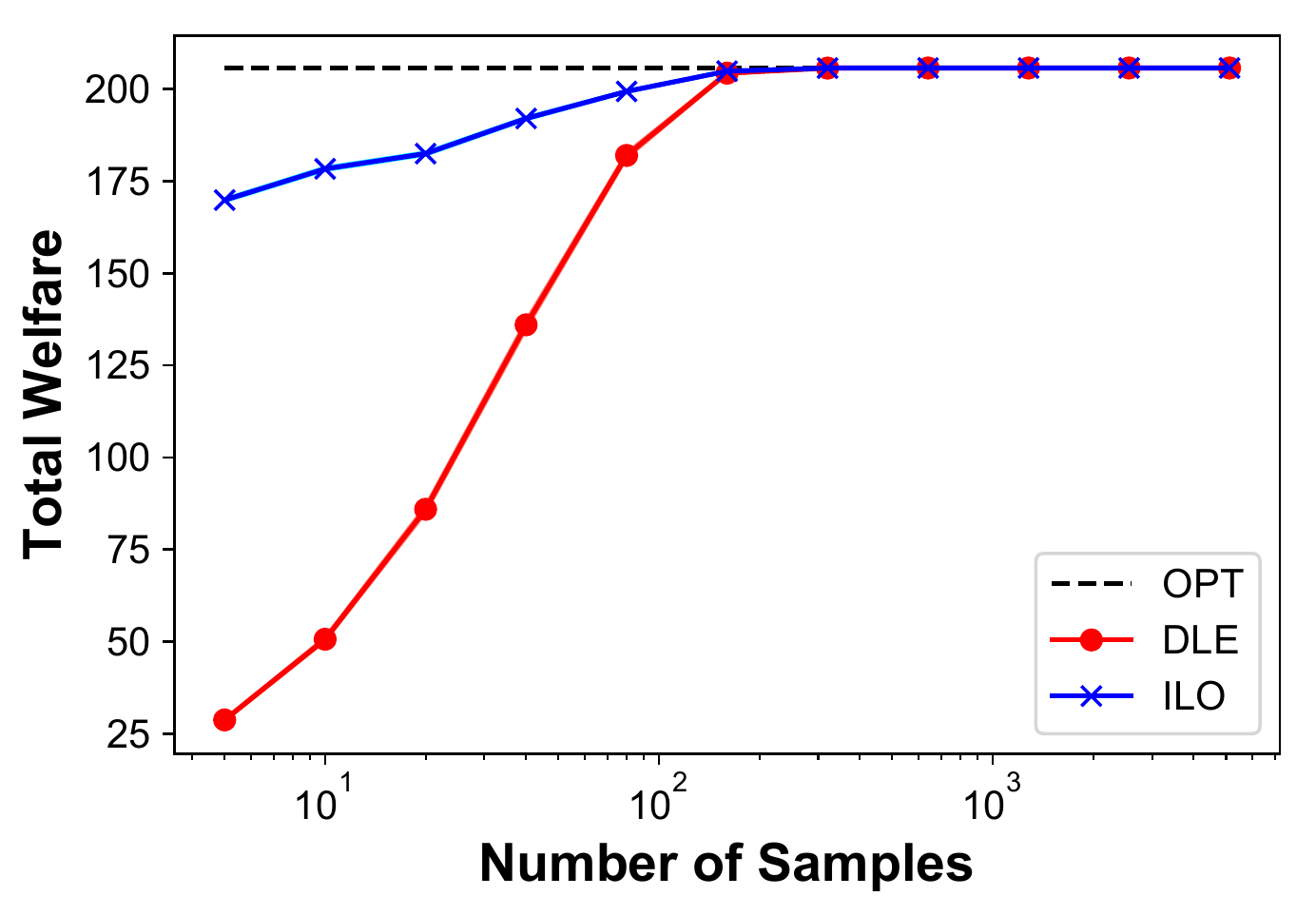}}\hfill
\subfloat[Sample size 5\label{fig:unitdemand-case3-dc2}]
  {\includegraphics[width=.25\linewidth]{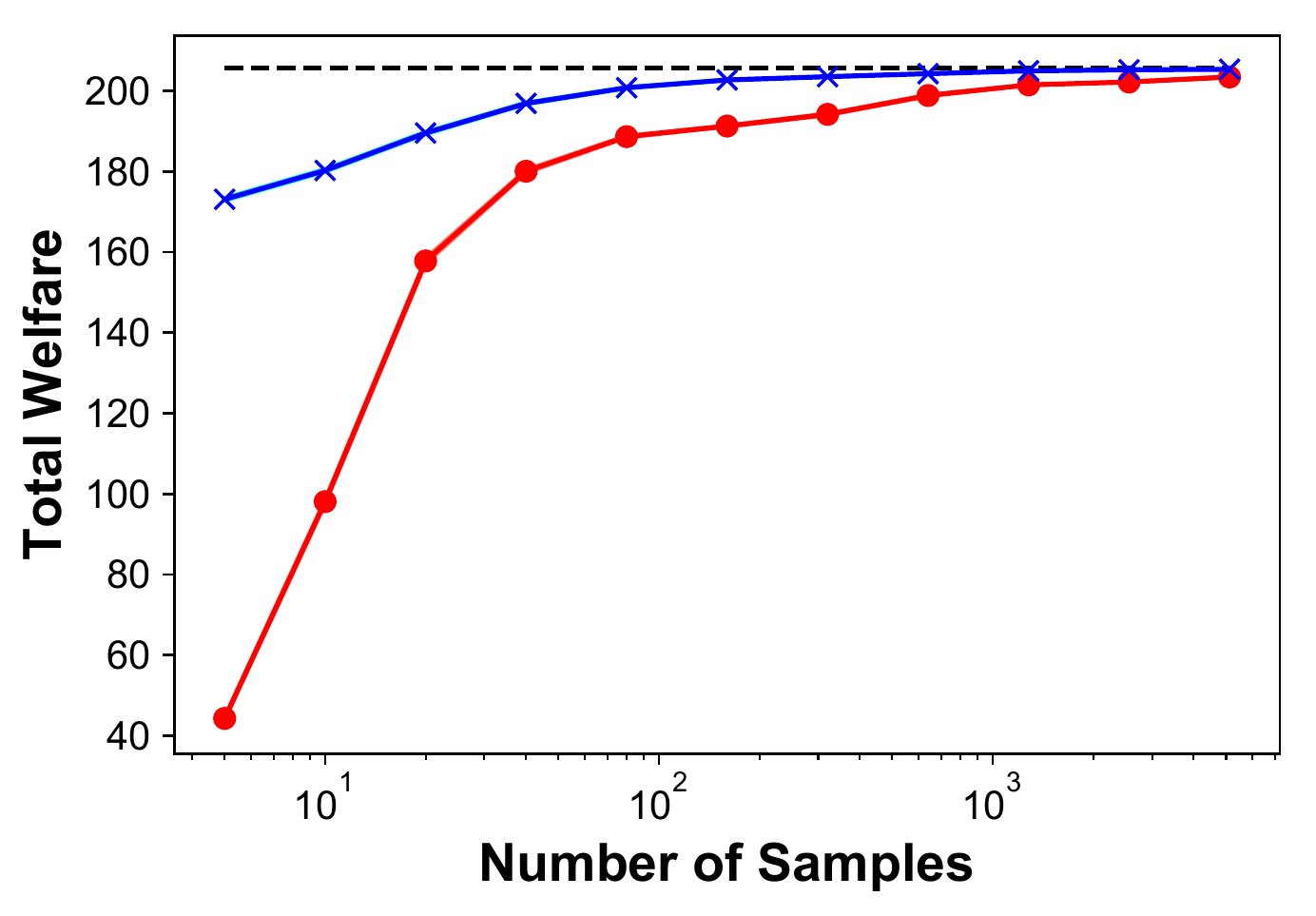}} \hfill
  \subfloat[Sample size 10\label{fig:unitdemand-case3-dc4}]
  {\includegraphics[width=.25\linewidth]{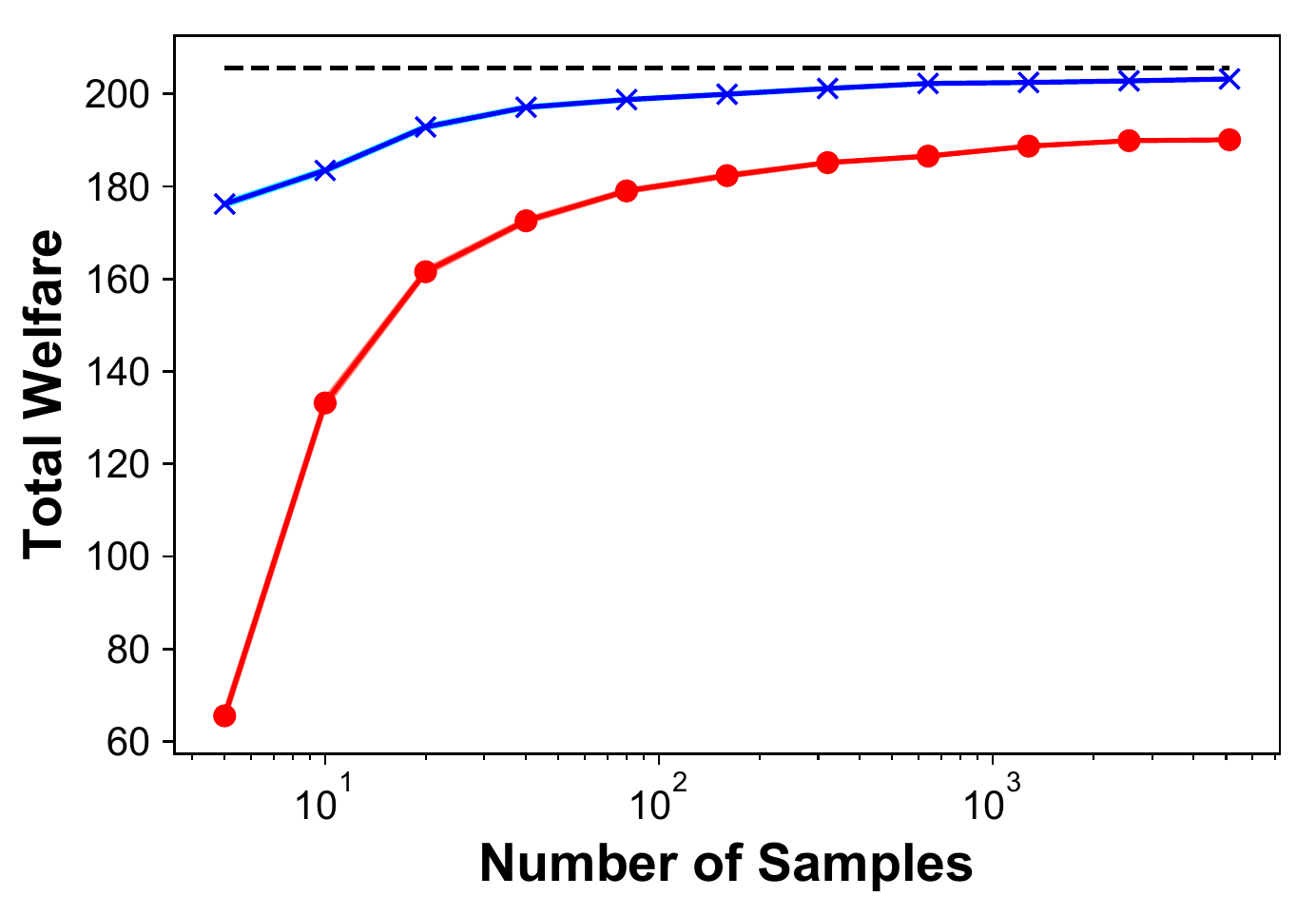}} \hfill
\subfloat[Uniform Product Distribution\label{fig:unitdemand-case3-dc5}]
  {\includegraphics[width=.25\linewidth]{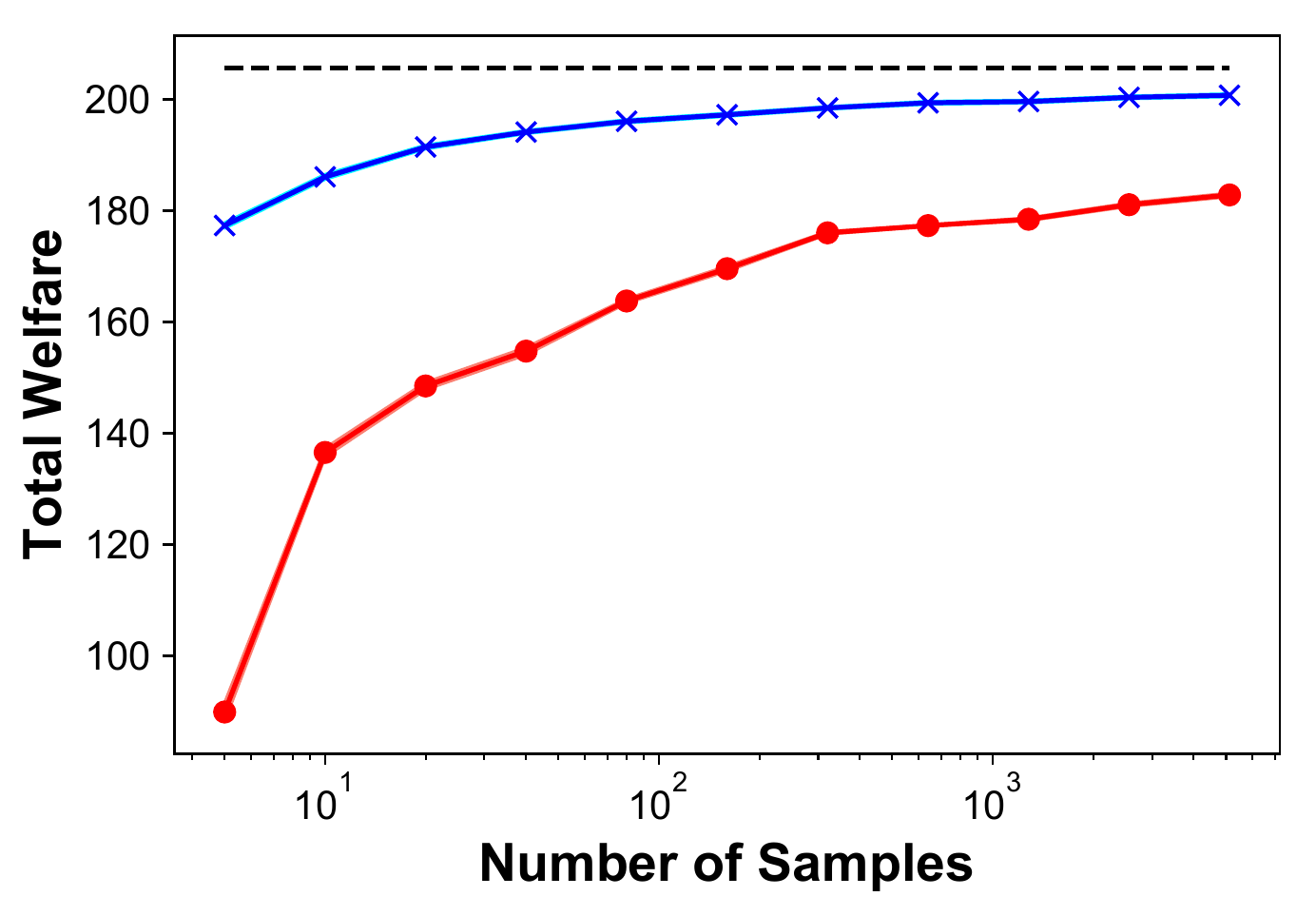}}\hfill
  
\caption{Unit Demand Markets (Balanced Market): Total welfare vs. \# of samples for various distributions when number of players is equal than the number of goods. \textbf{DLE} represents  directly learned equilibria and \textbf{ILO} represents  indirectly learned outcomes. \textbf{OPT} represents the welfare-maximizing equilibrium allocation.}
\label{fig:unitdemand-case3}
\end{figure}

\subsubsection{Empirical Loss Analysis}\label{subsub:UD-loss}
The empirical loss of both algorithms converge to $0$ in $\approx 100$ samples for all the market sizes and distributions we consider. In most cases, the indirect learning approach has a lower empirical loss than that of the direct learning approach. This is mainly because, when the budgets are almost equal, and all the samples have size greater than $1$, any allocation which allocates all the goods such that each player gets one good at a price equal to their budget will be consistent. This consistency arises from the fact that no bundle of size greater than $1$ will be affordable by any player due to budgets being almost equal. For similar reasons, the magnitude of loss reduces as the size of the samples increase for directly learned equilibria as well. As the size of each sample increases, it becomes likelier that this sample cannot be afforded by any player. At the same time, the variance in the empirical loss relative to the expected value also increases since the loss value becomes more sample specific; there are much fewer samples in the support of the distribution that can violate consistency. This can seen in Figure~\ref{fig:unitdemand-loss_case_2_main_paper}: while the magnitude of the empirical loss decreases from Figure \ref{fig:loss11} to \ref{fig:loss14}, the error bars increase in size.

The low loss of indirectly learned allocations, however, does not hold when the number of goods is greater than the number of players (buyers' market) (see Figure~\ref{fig:unitdemand-loss_case_2_main_paper}). In this case, there still are goods which are not allocated and therefore, when the samples have a relatively small size, it is likely to sample a bundle which only has one allocated good and therefore a low price; when the allocations do not allocate the best possible good to each buyer, this bundle will likely violate consistency for some player. This is exactly what happens in Figure \ref{fig:loss11} and \ref{fig:loss12}; when the valuations are only partially learned and not fully learned, direct learning has a much lower loss than indirect learning. When there are enough samples for the valuation functions to be learnt accurately, both direct learning and indirect learning converge to zero loss.
 
We further note that this general trend does not conclude that the empirical loss values are lower in indirect learning than in direct learning. In Example~\ref{ex:unitdemandindirect}, we can observe that the expected loss is $1/2$ for indirect learning, whereas the loss function value for direct learning converges to $0$ in a few samples (as soon as we observe both samples). The lack of a guarantee like Theorem~\ref{thm:consistency} creates uncertainty regarding  whether indirect learning will ever converge to an allocation with low loss whereas Theorem~\ref{thm:consistency} ensures that with enough samples, directly learned allocations always have low loss.

 \begin{figure}[ht]
\subfloat[Sample size 3\label{fig:loss11}]
 {\includegraphics[width=.25\linewidth]{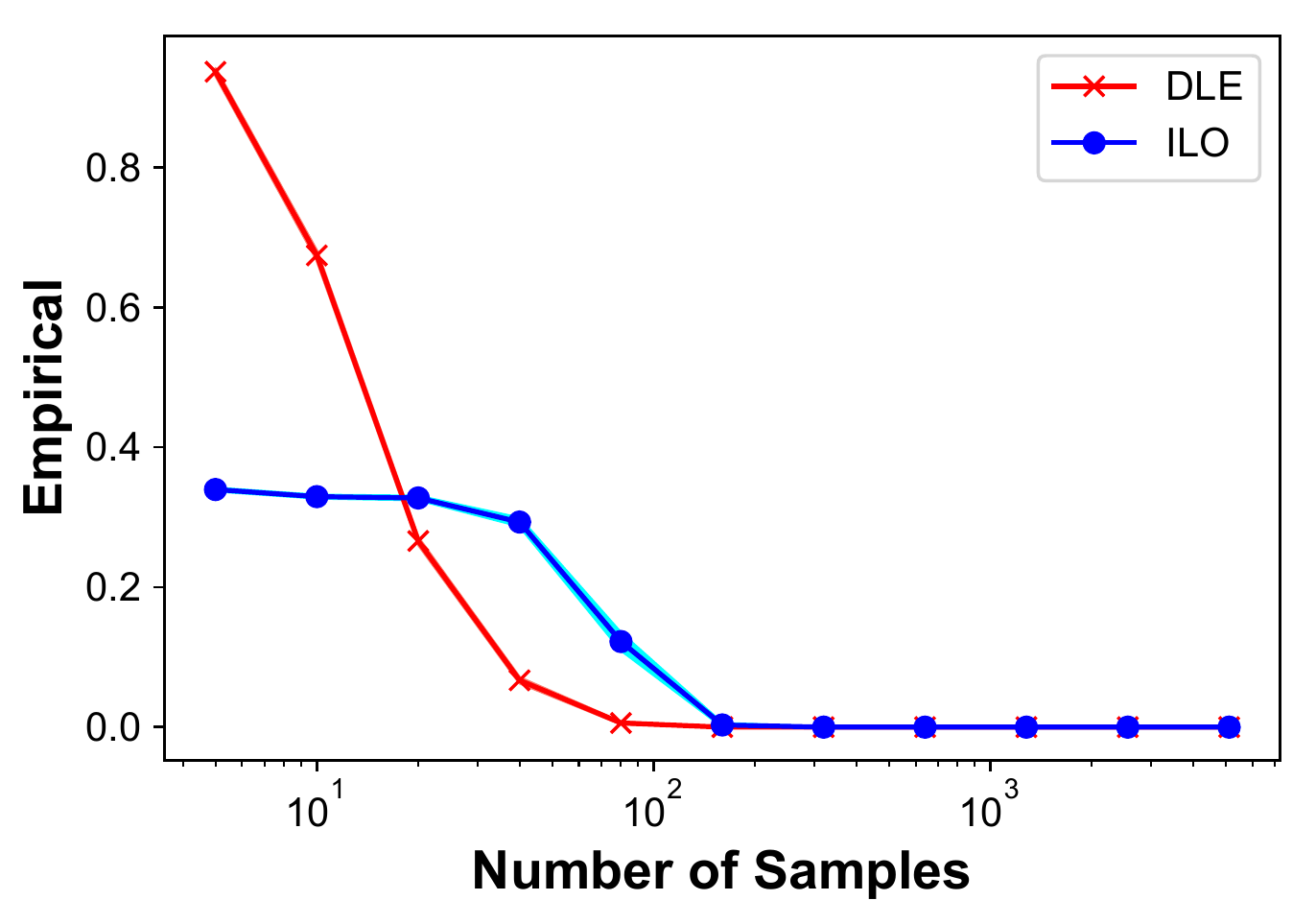}} \hfill
  \subfloat[Sample size 5\label{fig:loss12}]
  {\includegraphics[width=.25\linewidth]{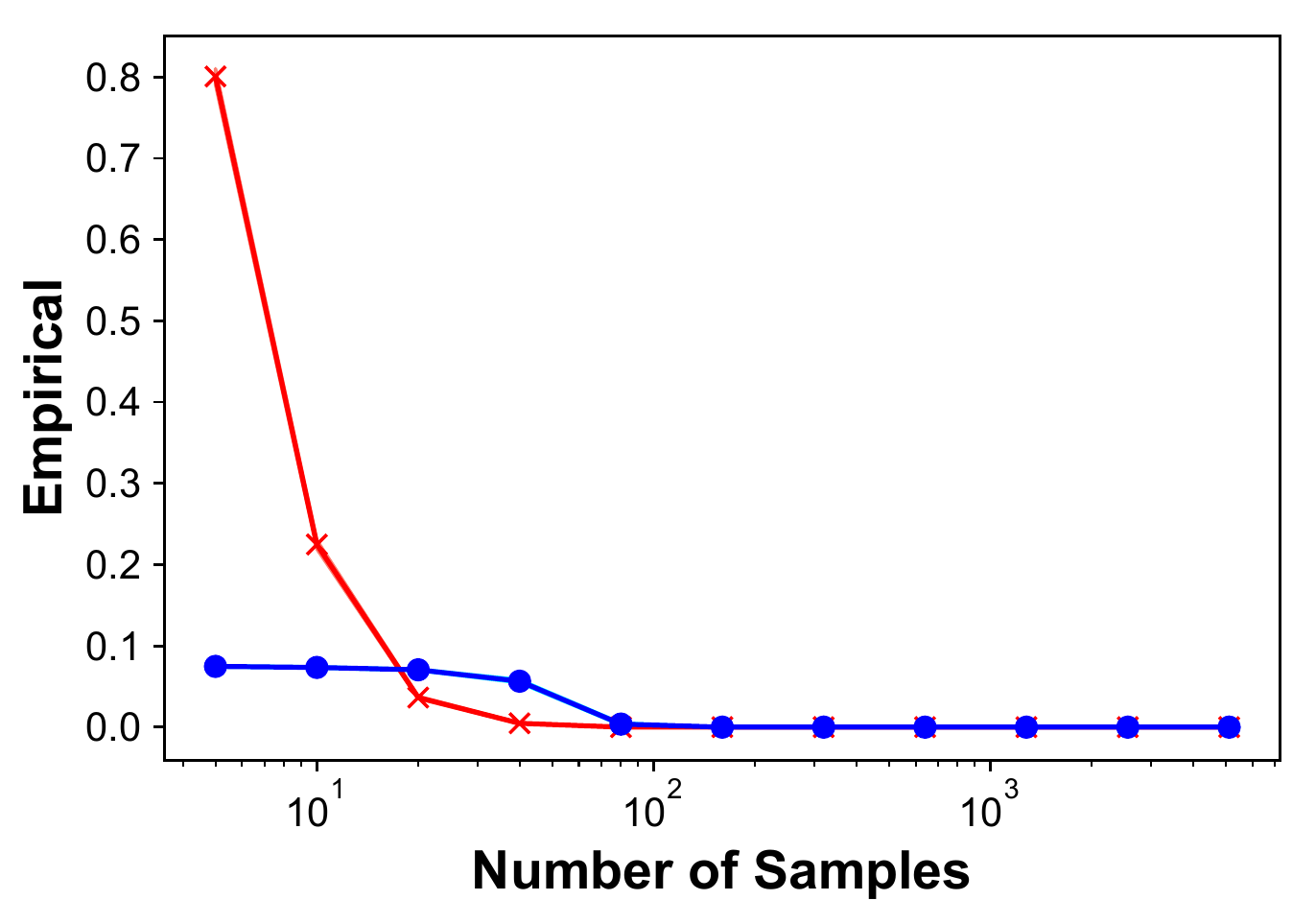}} \hfill
      \subfloat[Sample size 10\label{fig:loss13}]
  {\includegraphics[width=.25\linewidth]{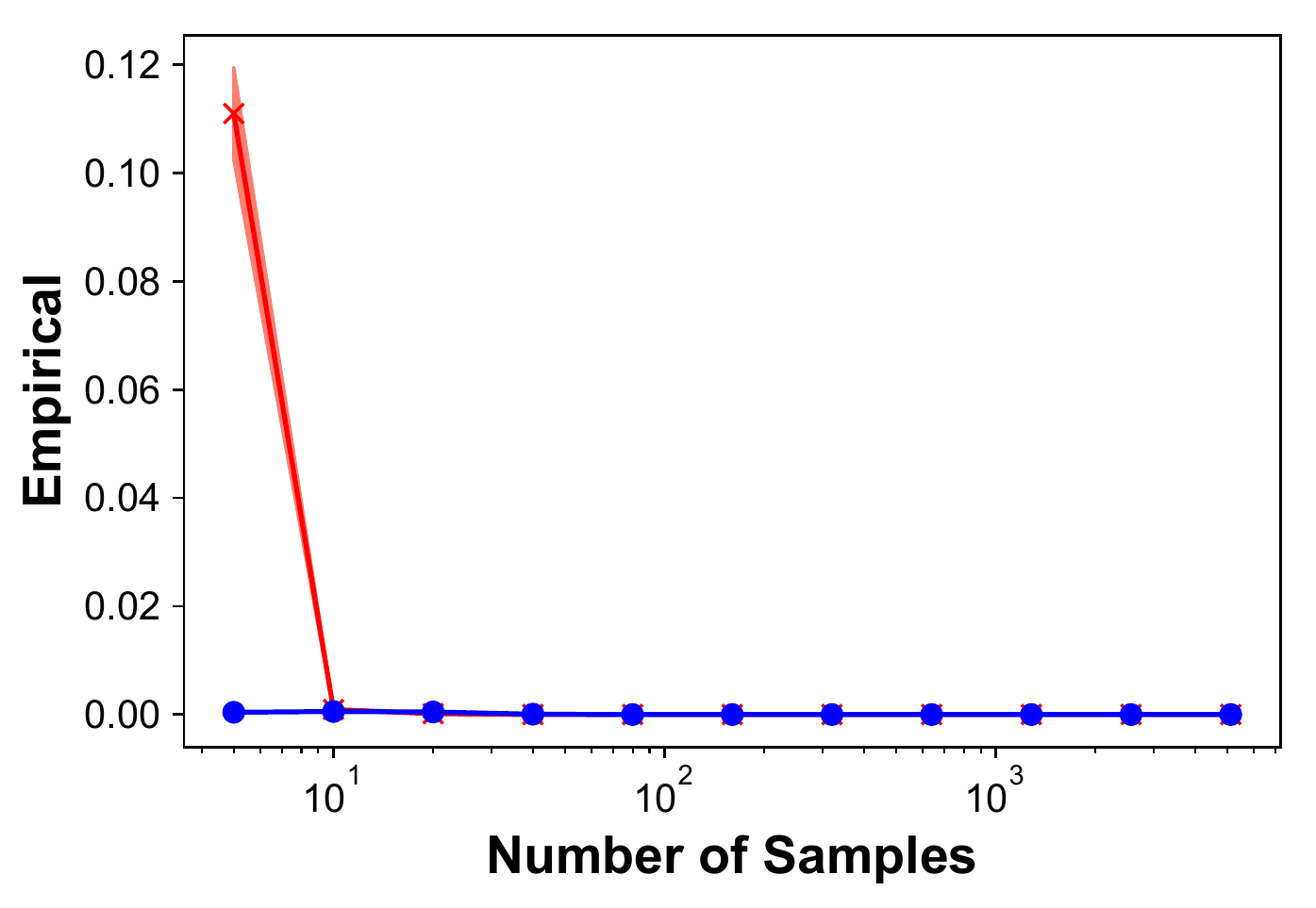}} \hfill
\subfloat[Uniform Product Distribution\label{fig:loss14}]
  {\includegraphics[width=.25\linewidth]{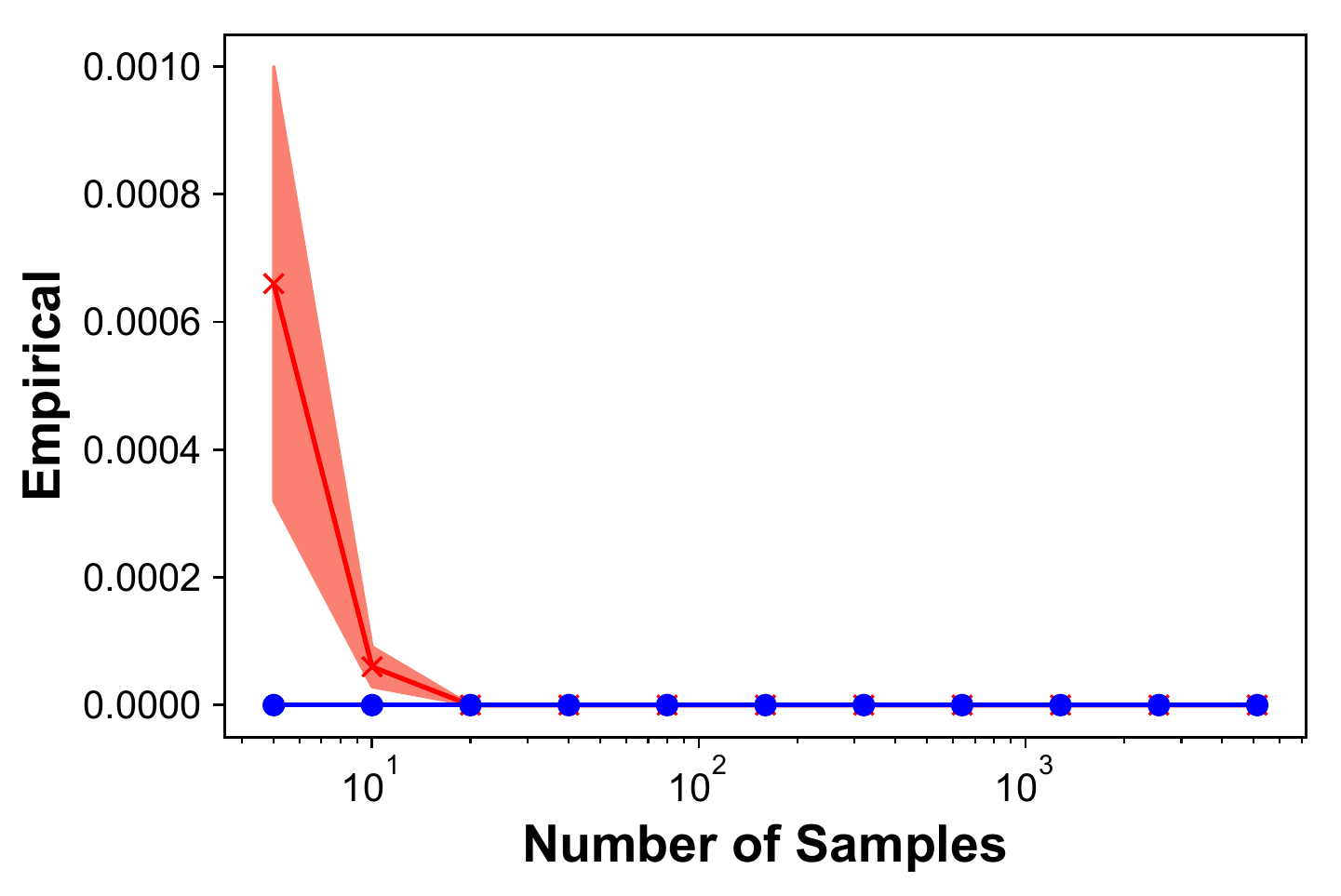}}\hfill
  
\caption{Unit Demand Markets (Buyers' Market): Empirical Loss vs. \# of samples for various distributions when number of goods is greater than the number of players. \textbf{DLE} represents  directly learned equilibria and \textbf{ILO} represents  indirectly learned outcomes.}
\label{fig:unitdemand-loss_case_2_main_paper}
\end{figure}

\subsection{Additive Markets}\label{subsec:additive-expts}
Assuming players have additive valuations, we generate PAC equilibria for the different market sizes and distributions described above using Algorithm \ref{algo:additive-valuations}. Since there is no other algorithm that computes an equilibrium for additive markets, we do not have a straightforward indirect learning approach to compare our algorithm to. Therefore, the indirect learning approach we use assumes that the goods are divisible since equilibria for divisible goods and additive valuations can be computed efficiently. We first learn valuations using regression and then compute an equilibrium assuming the goods are divisible using the algorithm presented in \shortciteA{birnbaum2010new}. Note that our notion of loss has no meaning for this algorithm, so we can only compare the welfare of both algorithms. We also compare the efficiency of our algorithm with the optimal welfare allocation when goods are indivisible (computed using an ILP) and the optimal equilibrium allocation assuming all goods are divisible. 

In addition to this, we evaluate our learning algorithm in terms of its market inconsistency, as measured by the empirical loss and inefficiency with respect to the number of goods burnt by the algorithm. 


\begin{figure}[ht]
\subfloat[Sample Size 1\label{subfig:additive-case1-dc0}]
  {\includegraphics[width=.25\linewidth]{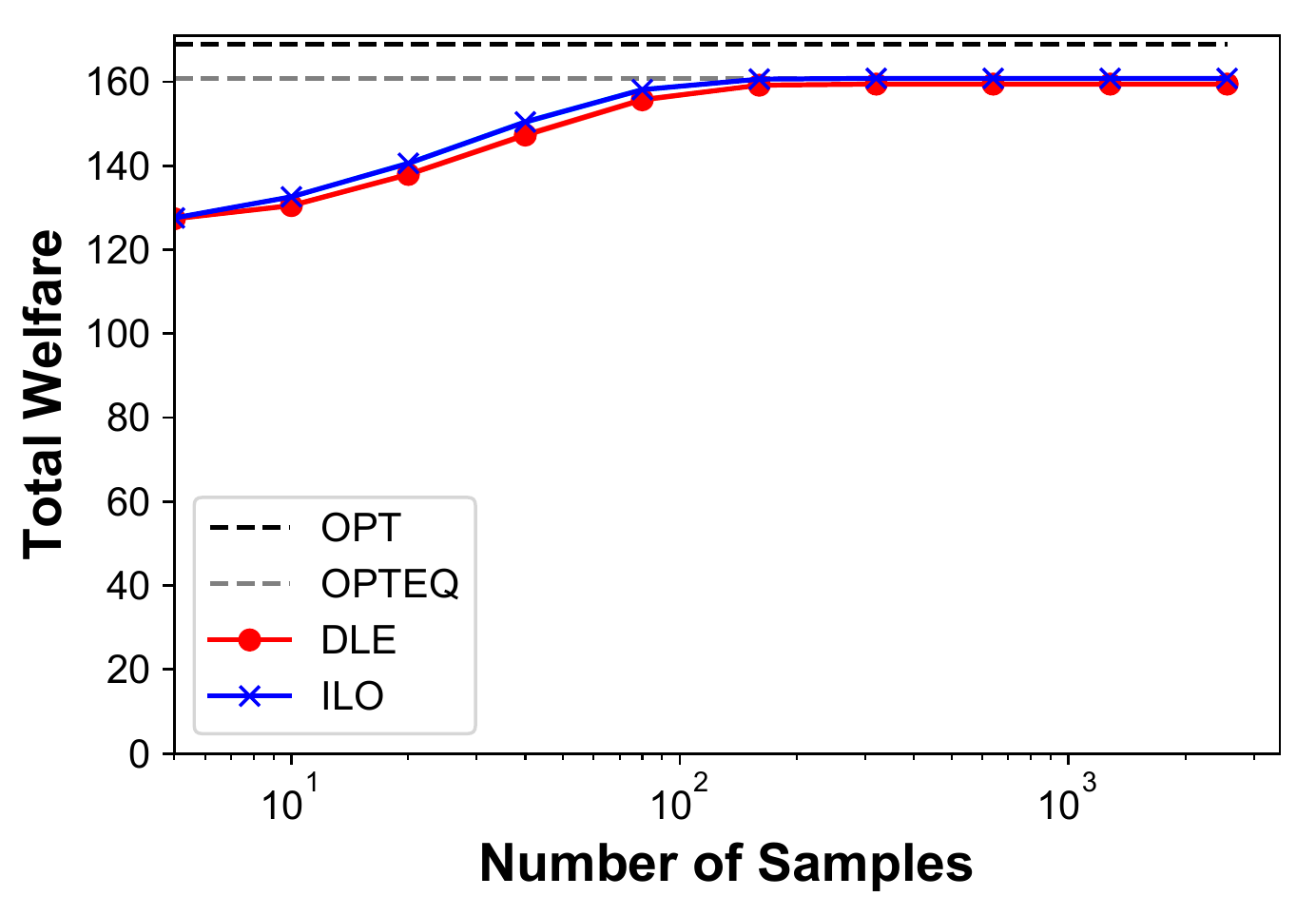}}\hfill
 \subfloat[Sample Size 5\label{fig:additive-case5-dc2}]
  {\includegraphics[width=.25\linewidth]{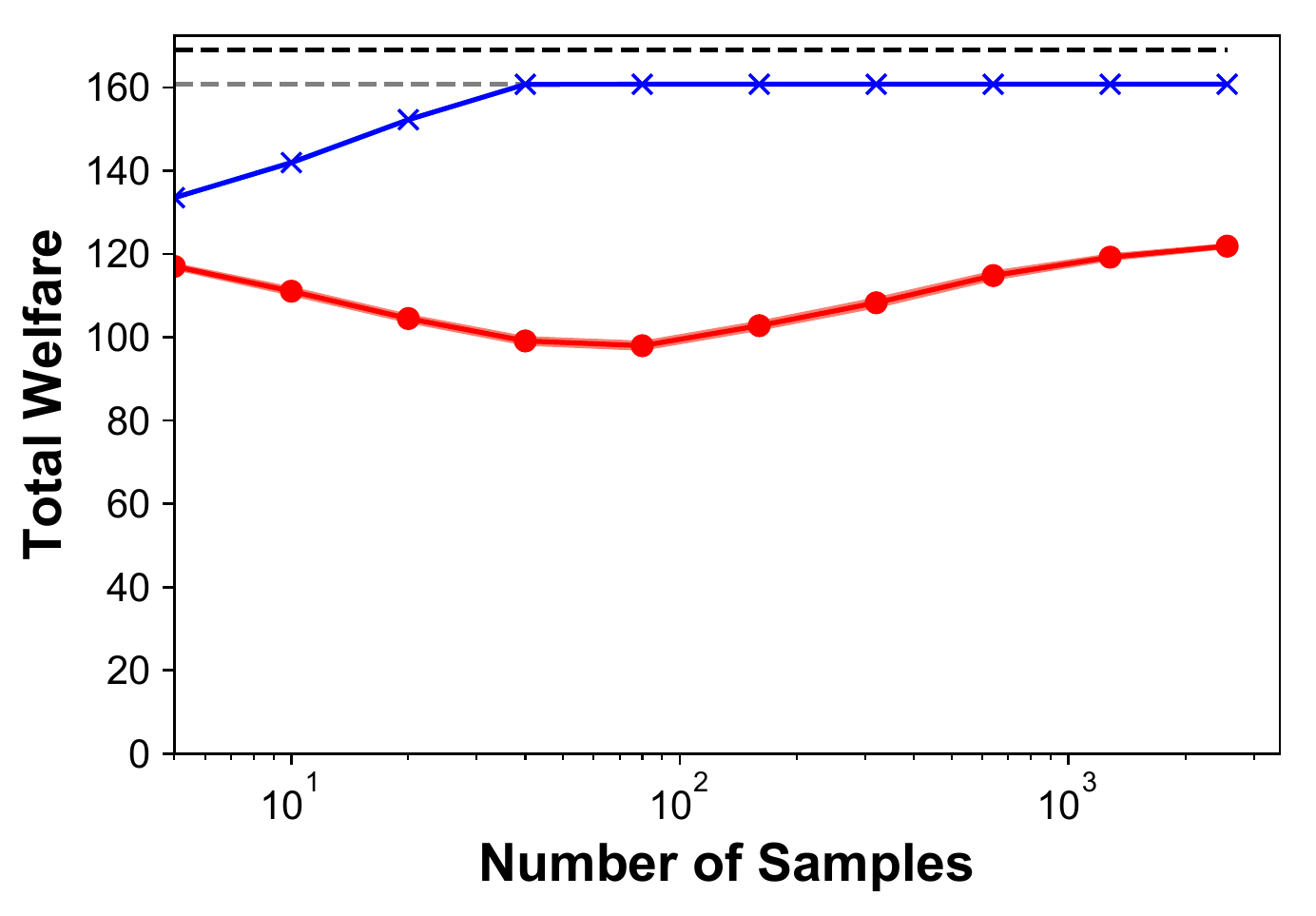}}\hfill
  \subfloat[Sample Size 10\label{fig:additive-case1-dc4}]
  {\includegraphics[width=.25\linewidth]{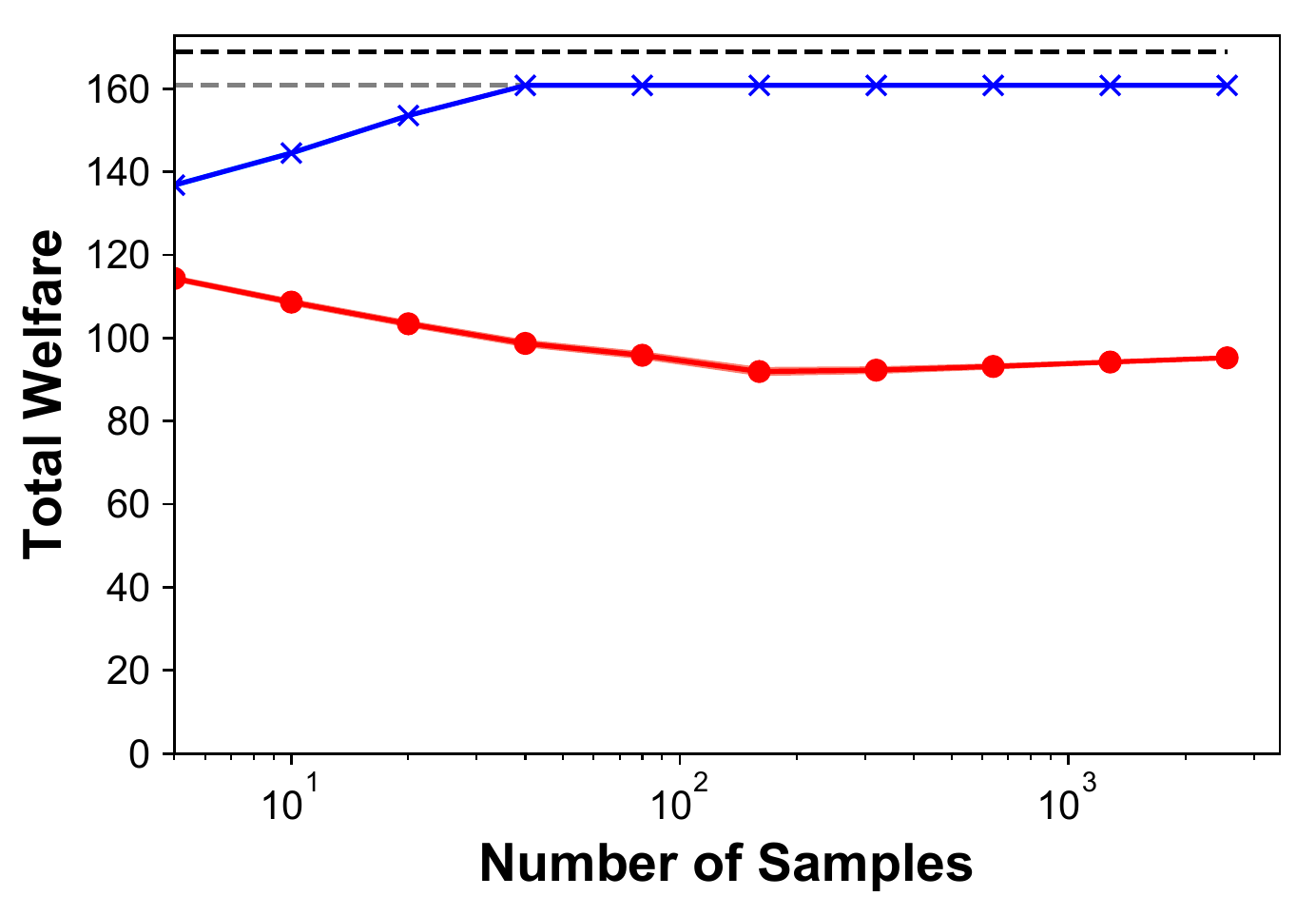}}\hfill
  \subfloat[Uniform Product Distribution\label{fig:additive-case1-dc5}]
  {\includegraphics[width=.25\linewidth]{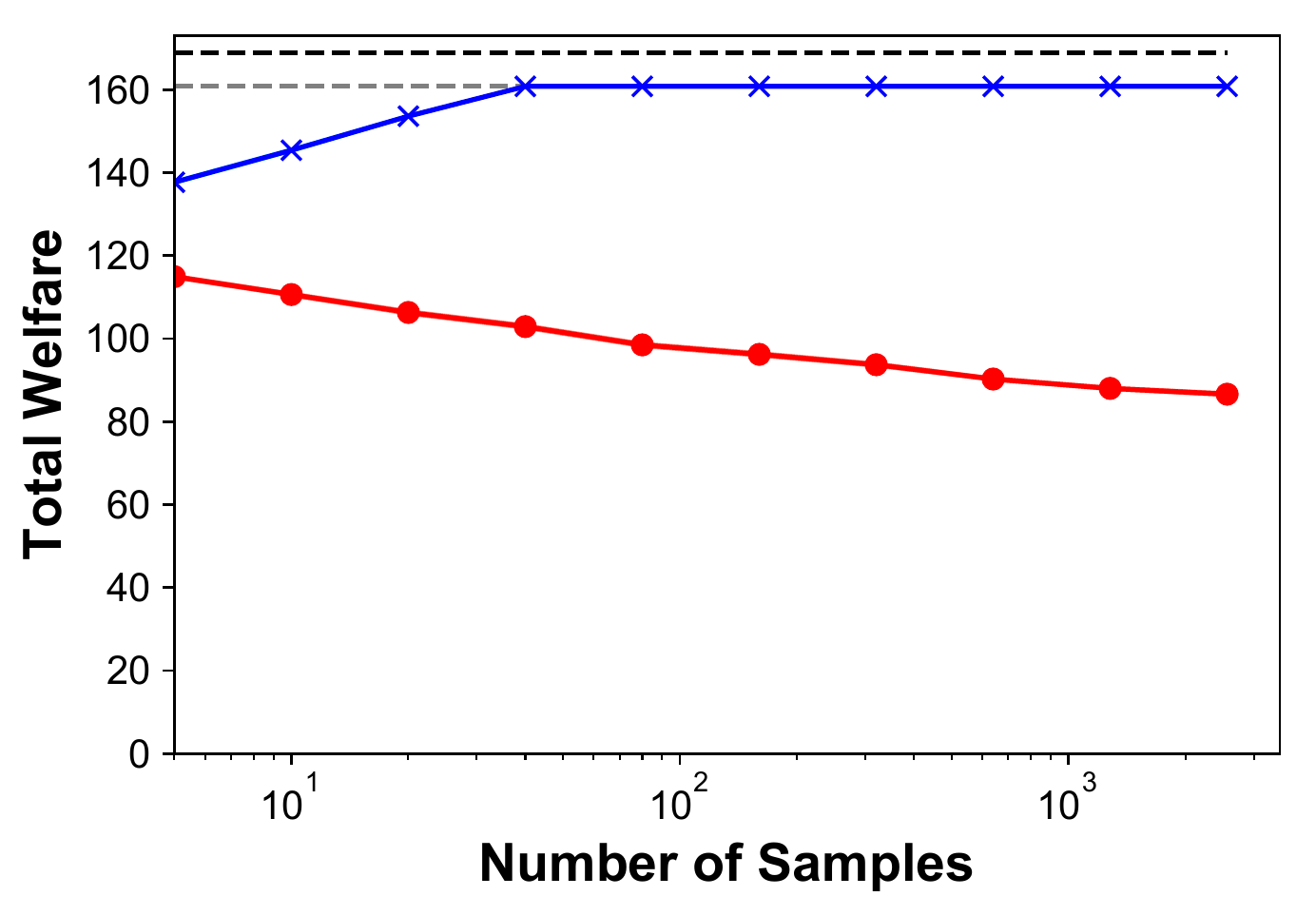}}\hfill
\caption{Additive Markets (Sellers' Market): Total welfare vs. \# of samples for various distributions. \textbf{DLE} represents  directly learned equilibria and \textbf{ILO} represents  indirectly learned outcomes. \textbf{OPT} represents the welfare-maximizing allocation and \textbf{OPTEQ} represents the welfare of the equilibrium allocation when goods are divisible.}
\label{fig:additive-case1-util}
\end{figure}

For all the markets we examine, we see a similar pattern emerge. With a large number of samples, the direct learning approach achieves more than $60\%$ of the welfare that the indirect learning approach achieves (see Figure \ref{fig:additive-case1-util}) inspite of the assumption that goods are indivisible while ensuring consistency with respect to the data. The performance of the direct learning approach decreases marginally as the size of the samples in the dataset increases. 

We further note that welfare first decreases with the number of samples and then increases. This is mainly because the number of burnt goods first increases and then decreases with the samples. Upon closer inspection, these two graphs seem to mirror one another, as one rises when the other falls (see Figure \ref{fig:additive-case1-util-burn}). This indicates that the major cause of the loss of welfare is due to the burning of goods. The trend in the number of burnt goods with respect to the number of samples, however, warrants further discussion. We believe this  occurs due to the fact that the algorithm allocates complete samples to players (when all the samples have the same size). When the size of each sample is constant, the algorithm will allocate disjoint samples to different players and then burn any samples partially but not completely intersecting with the set of allocated goods in order to maintain consistency. When the number of samples is low, we have fewer disjoint samples, so we do not allocate too many goods. Therefore, we are more likely to have bundles that partially intersect with the set of allocated goods and therefore, the algorithm burns a lot of goods in the initial stages. However, at the later stages, when the number of samples is much higher, we are likelier to have disjoint bundles and so we allocate more goods. Therefore, it is less likely to have samples that partially intersect with the set of allocated goods and likelier to have samples that are a subset of the set of allocated goods. As a result of this, the number of burnt goods decreases.

\begin{figure}[ht]
\center
\subfloat[Welfare vs \# of samples\label{subfig:additive-case1-welfare}]
  {\includegraphics[width=.3\linewidth]{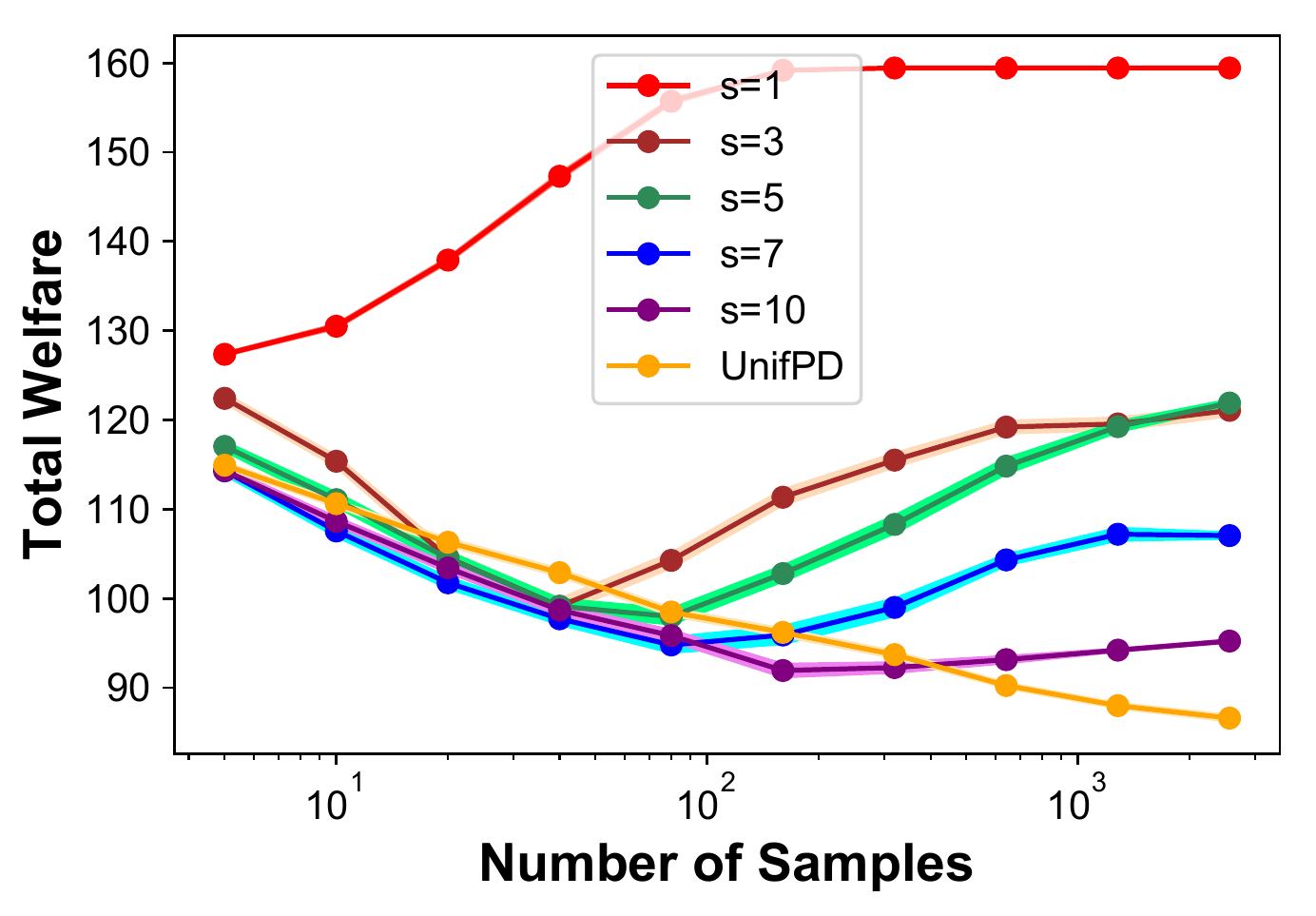}}
\subfloat[Number of Burnt Goods vs. \# of Samples \label{subfig:additive-case1-burn}]
  {\includegraphics[width=.3\linewidth]{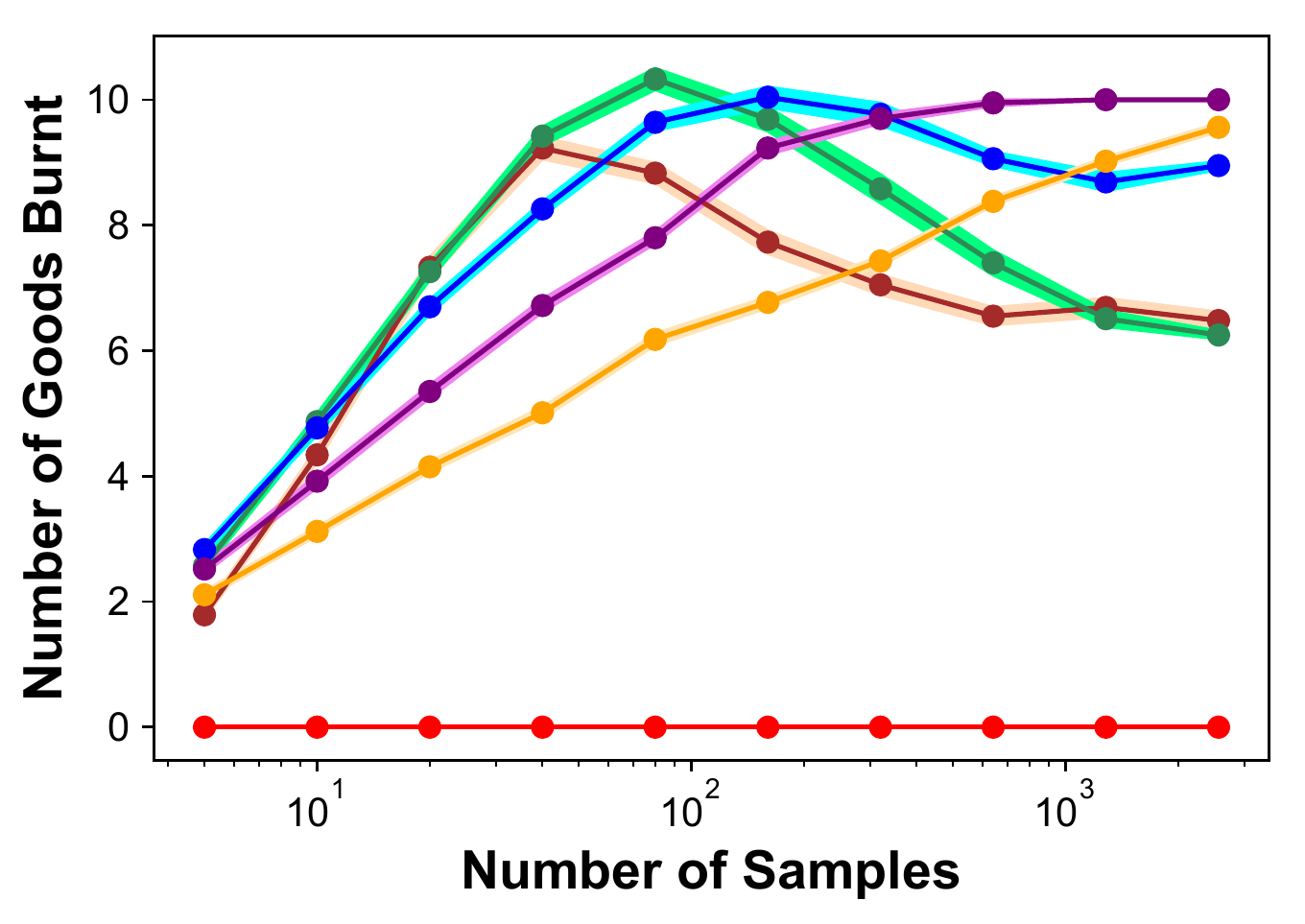}}\hfill
\caption{Additive Markets (Sellers' Market): (Left) Total Welfare vs. \# of samples for various distributions and (Right) \# of Burnt Goods vs. \# of Samples for various distributions.}
\label{fig:additive-case1-util-burn}
\end{figure}

\subsection{Submodular Markets}\label{subsec:submodular-expts}
Similar to additive markets, we use the setup described at the beginning of this section to generate PAC equilibria for all the different market sizes and distributions using Algorithm~\ref{algo:submod-valuations}. For each market and distribution, we evaluate our algorithms for three Threshold values $\Th \in \{3, 5, 10\}$. To the best of our knowledge, there exists no efficient algorithm to compute an equilibrium for these valuations even when goods are divisible, so we compare the efficiency of our algorithm to the optimal welfare allocation which can be computed using an integer linear program. Similar to Section \ref{subsec:additive-expts}, we also evaluate our learning algorithm in terms of its market inconsistency, as measured by the empirical loss and inefficiency with respect to number of goods burnt by the algorithm.

The performance of our algorithm is similar to that of additive markets with two key differences. First, when the number of samples is low, the total welfare is also quite low. When the number of samples is low, there are a lot of leftover goods when the algorithm terminates and both algorithms allocate all the leftover goods to Player $1$ at a price of $0$. While this works and results in a large welfare increase when valuations are additive, when valuations are submodular and there is a threshold value, this has very little effect since the marginal gain drops to $0$ very quickly. Secondly, as the size of each sample grows, the welfare decreases much more sharply to the extent that when all the samples have a size of $10$, the welfare is roughly a fourth of the welfare when all the samples have a size of $3$. This is again because of the Threshold value. When the threshold value is $3$, samples of size $3$ have nearly additive valuations. Therefore, these samples have a much larger utility per unit size than samples of size $10$. Since we allocate complete samples (in the case where all the samples have the same size), when all the samples have size $3$, we can allocate many more samples with high value and the resultant allocation has a much higher welfare. Note that this is also the reason that samples of size $5$ come very close to samples of size $3$ in performance when the threshold value is increased to $5$.

\begin{figure}[ht]
\center
\subfloat[Threshold 3\label{subfig:submod-case1-cutoff3}]
  {\includegraphics[width=.3\linewidth]{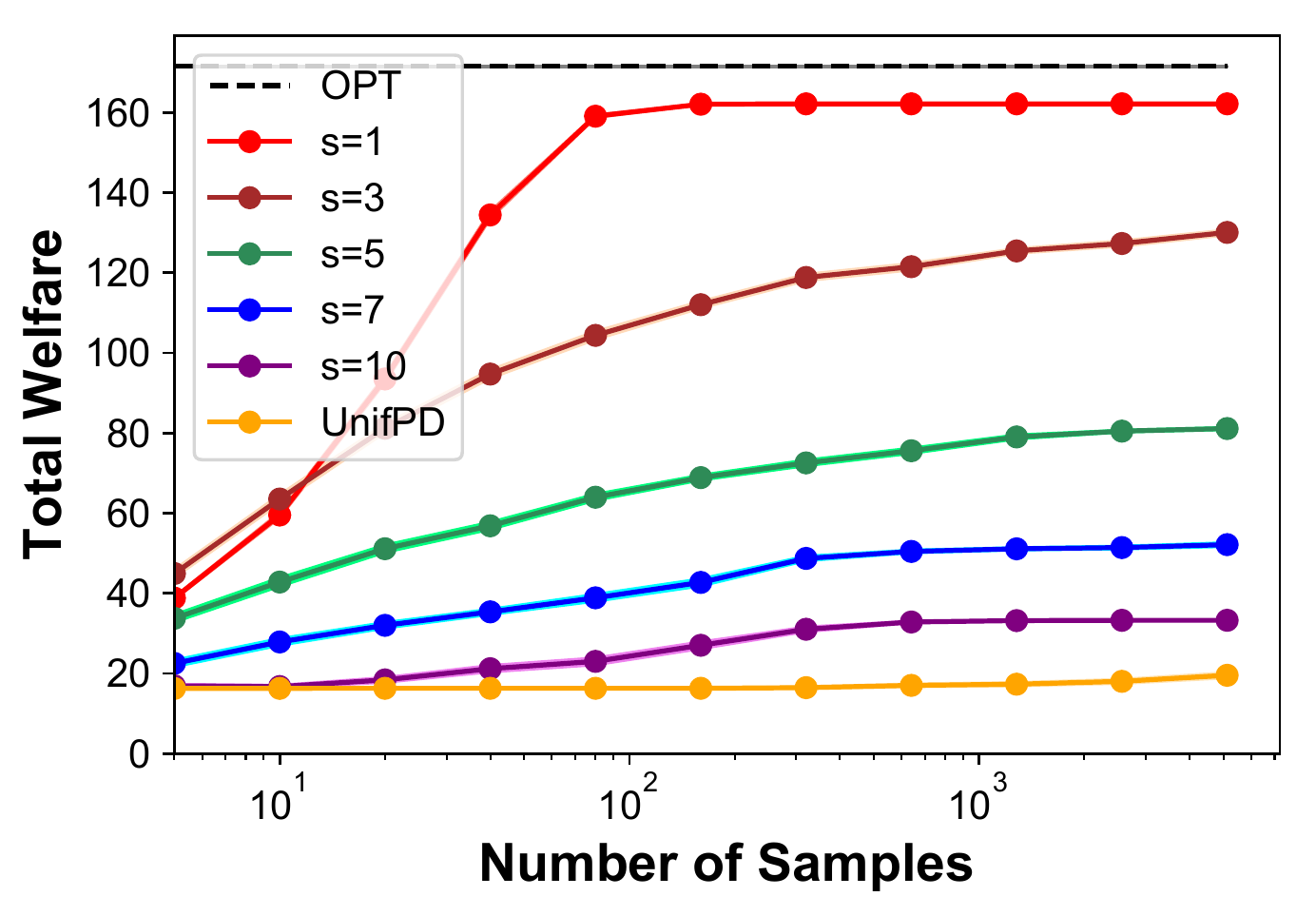}}
 \subfloat[Threshold 5\label{subfig:submod-case1-cutoff5}]
  {\includegraphics[width=.3\linewidth]{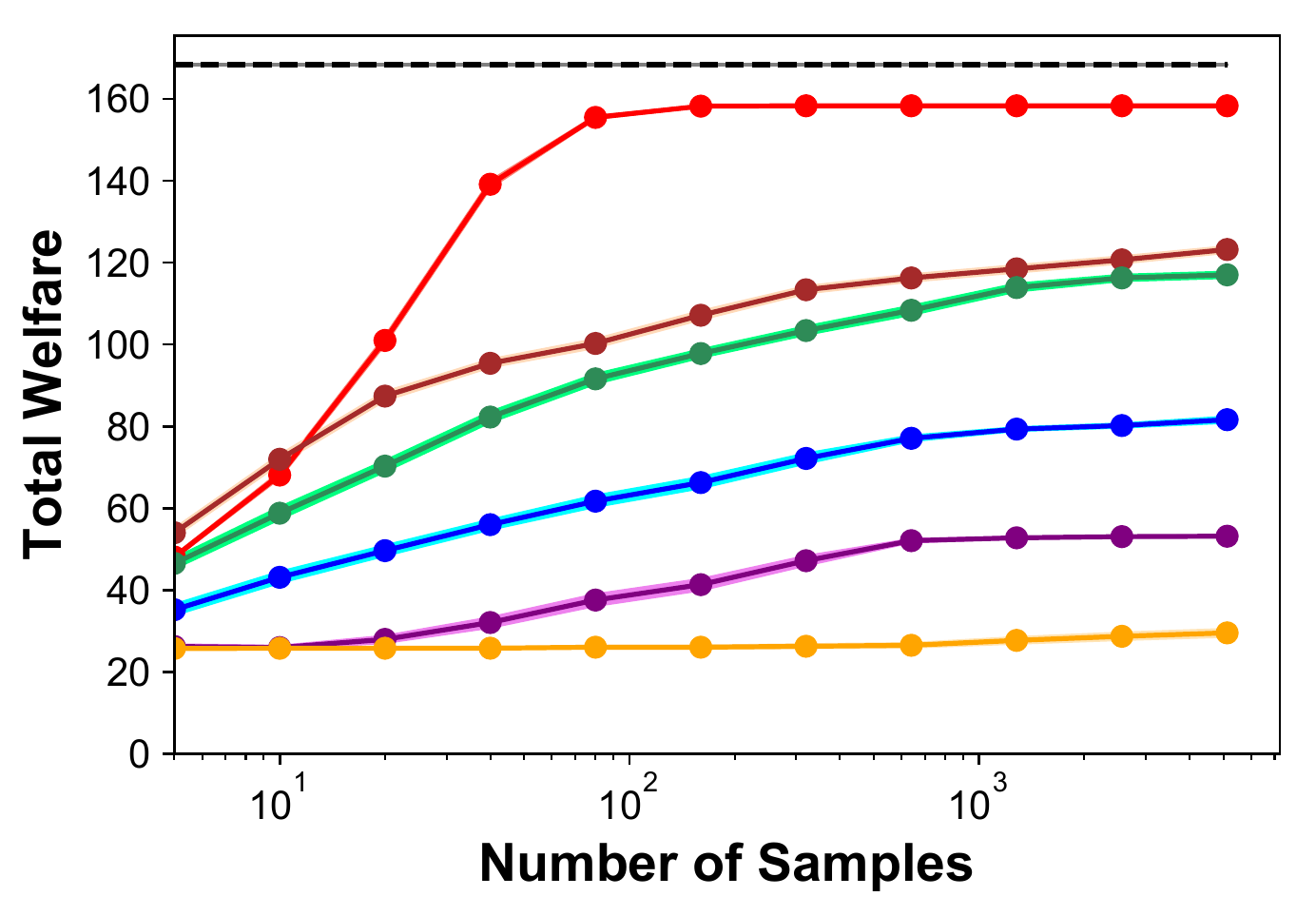}}\hfill
\caption{Submodular Markets (Sellers' Market): Total welfare vs. \# of samples for various distributions. Each line corresponds to a specific distribution. \textbf{OPT} represents the welfare-maximizing allocation.}
\label{fig:submod-case1-util}
\end{figure}

\subsection{Empirical Loss Analysis}
The empirical loss for additive and submodular markets converges to $0$ in roughly $200$ samples for all the different market sizes and distributions we consider (see Figure \ref{fig:add-submod-loss}). We also see a similar trend when we compare empirical loss versus the size of the samples in the dataset: as the size increases, the empirical loss decreases. Similar to unit demand markets, we attribute this to the fact that it is harder to violate consistency with samples of large size as opposed to smaller samples.

\begin{figure}[ht]
\center
\subfloat[Submodular Markets (Threshold 5)\label{subfig:submod-losscase1-cutoff5}]
  {\includegraphics[width=.3\linewidth]{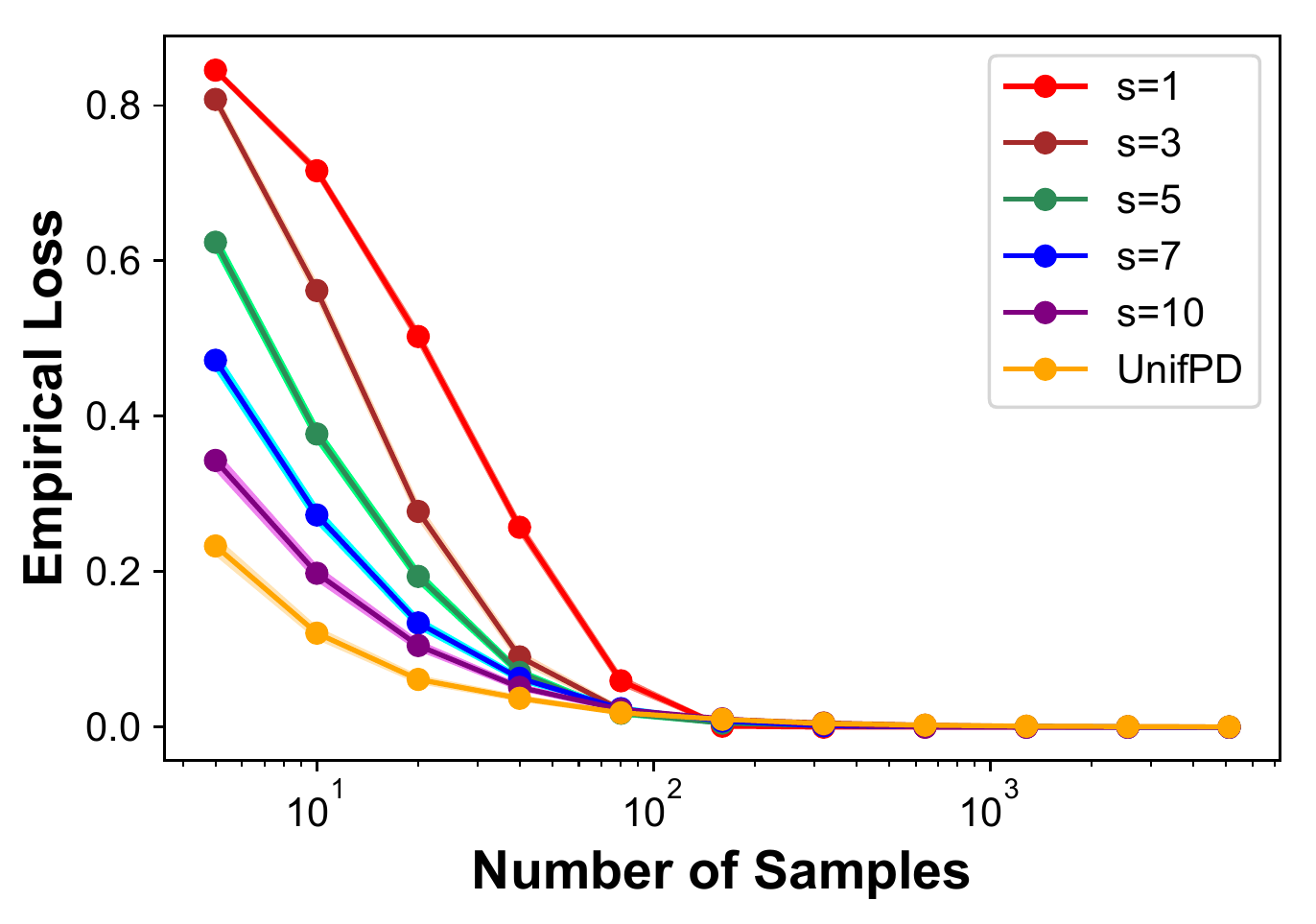}}
 \subfloat[Additive Markets\label{subfig:additive-losscase1}]
  {\includegraphics[width=.3\linewidth]{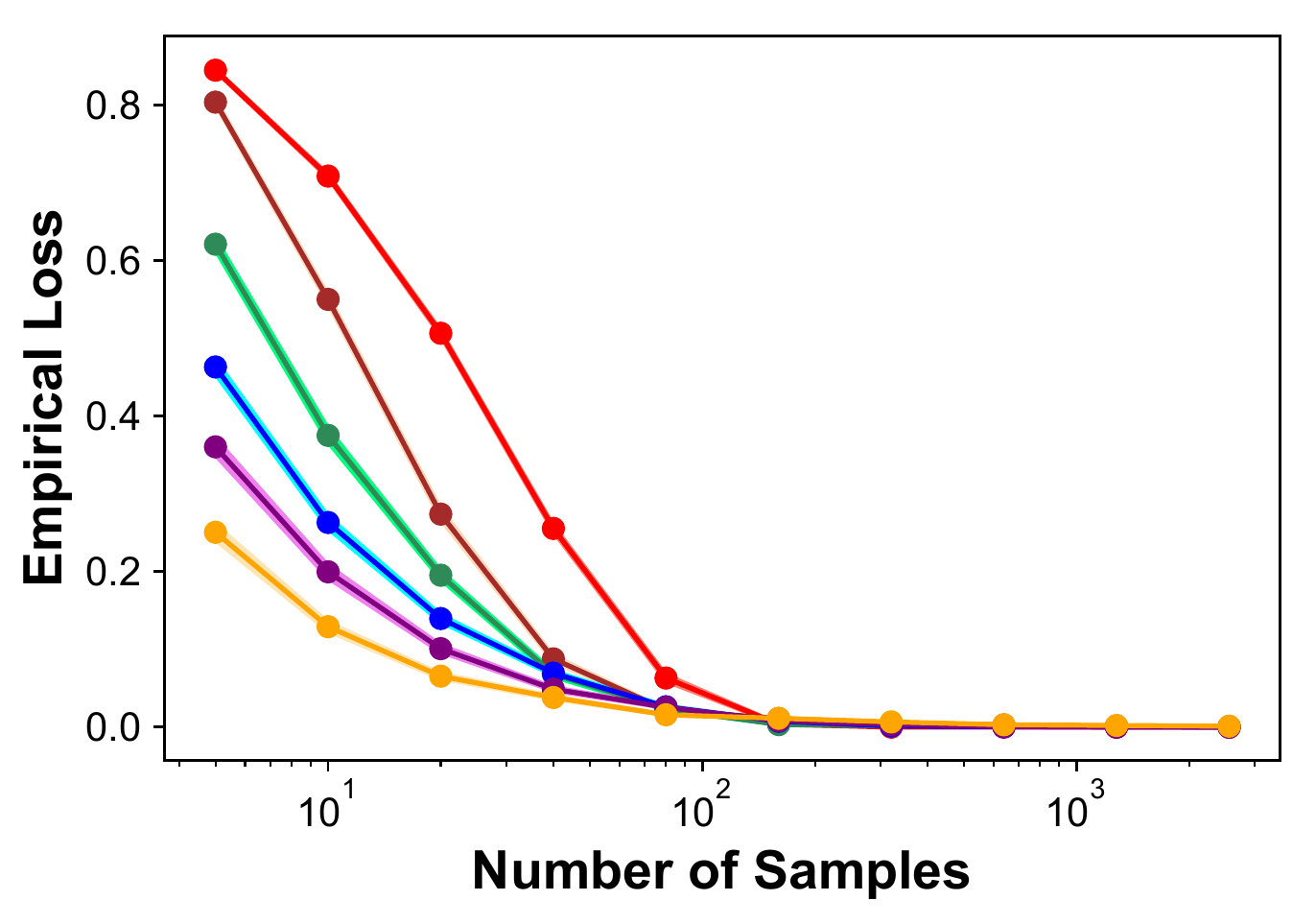}}\hfill
\caption{Sellers' Market: Empirical Loss vs. \# of samples for various distributions.}
\label{fig:add-submod-loss}
\end{figure}
\section{Conclusions and Future Work}\label{sec:conclusion}
This work shows the benefit of directly learning equilibrium states, instead of learning utility functions, and calculating equilibria states from them. We deal with several valuation function families, and in all of them show algorithms to produce a PAC-approximation, with our results being tight, i.e., no better approximation can be guaranteed.

Examining more realistic-seeming settings allows us to see that this approach carries with it several advantages and several drawbacks. The direct algorithm is more ``careful'', avoiding making significant mistakes (which is what ensures its theoretical guarantees), but the indirect approach in unit-demand utilities, performs risky moves in low-information settings (few samples) that can pay off, particularly in cases (as in our simulations), when the difference in valuations between items is not large. In more elaborate settings, involving more complex utility functions, we can see that our method finds equilibria, when having a decent-sized sample space, and the equilibrium quality compares well with the \emph{divisible} case, indicating we find at least a solid baseline for equilibrium in the indivisible case.

We believe that this work is the tip of the iceberg in showing how PAC learning can help in reaching economic, game-theoretic results, directly from the data, without using the data to construct intermediate steps (such as learning utility functions). Plenty of problems are still open -- from expanding results to a larger family of functions (XOS, gross substitutes), to further type of results (e.g., other desirable states beyond equilibria).


\bibliography{abb,learningmarketrefs}
\bibliographystyle{theapa}

\appendix
\onecolumn


\section{Missing Proofs from Section~\ref{sec:unitdemand}}\label{apdx:unit-demand}

\lembestegm*
\begin{proof}
We prove this result by induction on $i$. 
Since player 1's budget, $b_1$, is the highest, in any equilibrium allocation they should be allocated item $g_1^* = \argmax_{g'\in \cal G_1} v_1(g')$. Note that since players have strict preferences over items, $g_1^*$ is unique. If player 1 is not allocated $g_1^*$, then either it is unallocated (and has a price of 0), in which case player 1 demands it, contradicting that it is an equilibrium. Otherwise, it is allocated to another player $j$ whose budget is $b_j < b_i$, in which case the price of $g_1^*$ is less than $b_1$, and player 1 demands it.
Assume that the claim holds for players $1,\dots,i-1$, and consider player $i$.
If $g_i^* \neq \argmax_{g\in \cal G_i} v_i(g)$ then $\argmax_{g\in \cal G_i} v_i(g)$ is allocated to some other player $b_{i'}$ where $i'>i$ which is impossible because then player $i$ can afford the good $\argmax_{g\in \cal G_i} v_i(g)$. The above argument shows that any market equilibrium should assign $g_i^*$ to player $i$ and $g_i^*$ is the best possible good which can be assigned to player $i$, which implies that for any equilibrium allocation $\cal A$, $v_i(A_i) = v_i(g_i^*)$. Therefore the social welfare for any equilibrium is $\sum_{i\in N}v_i(g_i^*)$    
\end{proof}

\lemprodDexpbound*
\begin{proof}
Let $\hat{g}_i$ be the most valued good for player $i$ in $\bigcup \limits_{S\in \cal S^t_i }S$; we set $\alpha = \min_{g \in G} \Pr_{S \sim \cal D}(g \in S)$ and $\beta = \max_{g \in G} \Pr_{S \sim \cal D}(g \in S)$. 

At the $t$-th iteration of the \textbf{\textit{while loop}} in Algorithm~\ref{algo:unit-demand-market} for player $i$; $B_i^t \neq \{\hat{g}_i\}$ if and only if there exists another good $g{'} \ne \hat{g}_i$ which appears in all samples which contain $\hat{g}_i$, and does not appear in samples which do not contain $\hat{g}_i$ in $\cal S_i^t$. When this happens, $\{\hat{g}_i,g'\}\subseteq B_i^t$.

This event occurs with an exponentially low probability.
The probability that a good $g' (\ne \hat{g}_i)$ is present or absent together with $\hat{g}_i$ in a sample $S\in \cal S_i^t$ is 
\begin{align}
\Pr_{S \sim \cal D}(g' \in S) \Pr_{S \sim \cal D}(\hat{g}_i \in S)  &+\notag\\ \bigg (1 - \Pr_{S \sim \cal D}(g' \in S) \bigg ) \bigg (1 - \Pr_{S \sim \cal D}(\hat{g}_i \in S) \bigg ) &\label{eq:prob-together}
\end{align}
The upper bound on \eqref{eq:prob-together} for any product distribution is $\gamma^2 + (1-\gamma)^2$ where $\gamma = \min(\alpha, 1 - \beta)$.
When given $\ge k^2$ samples in $\cal S^t_i$, the probability that a good $g' (\ne g^*)$ is present or absent together with $\hat{g}_i$ in all samples is 
\begin{equation}
\le(\gamma^2 + (1-\gamma)^2)^{k^2}= \me^{-k^2 \log \big (  \frac{1}{\gamma^2 + (1-\gamma)^2} \big )}\label{eq:bound-gamma}
\end{equation}
Equation~\eqref{eq:bound-gamma} is $< \me^{-k}$ only when 
$\gamma \in \left(\frac12 \pm\frac{\sqrt{2\me^{-1/k}-1}}{2} \right)$; in particular
\begin{align*}
\min(\alpha, 1 - \beta) > \frac{1}{2} - \frac{\sqrt{2\me^{-1/k}-1}}{2}
\end{align*} 
This implies that 
\begin{align}
\frac{1}{2} - \frac{\sqrt{2\me^{-1/k}-1}}{2} < \Pr_{S \sim \cal D}(g \in S) < \frac{1}{2} + \frac{\sqrt{2\me^{-1/k}-1}}{2} \label{eq:prob-interval}
\end{align}

Since this is true for all goods, using the union bound, the probability that $B_i^t \ne \{\hat{g}_i\}$ is  $\le (k-1)\cdot \me^{-k} \le \me^{-k/2}$. This completes the proof.

Note that this lemma holds when the probabilities of sampling each good are given by \eqref{eq:prob-interval}. This is a larger interval than the one given in the statement of the lemma (where the lower bound is doubled). The smaller interval in the statement of the lemma exists solely to make the proof of Theorem \ref{thm:prod_improved_bound} easier to understand. 
\end{proof}

\propproddistunitD*
\begin{proof}
Assume that player $i$ gets good $g_i^*$ in the optimal equilibrium allocation; let $\alpha = \min_{g \in G} \Pr_{S \sim \cal D}(g \in S)$ and let $\beta$ be $\max_{g \in G} \Pr_{S \sim \cal D}(g \in S)$. Using a similar argument to that in Theorem~\ref{thm:prod_improved_bound}, the probability that Algorithm~\ref{algo:unit-demand-market} assigns $g_i^*$ to player $i$ for $i\leq \max\{\log n,\log k\} /\log (\frac{1}{1-\beta})$ is at least 

\begin{equation*}
  1- \frac{ 2\max\{\log n,\log k\}^2}{\log^2\big(\frac{1}{ (1-\beta)}\big)}\me^{-\frac{k}{4}}   
\end{equation*}

The second part of the proof uses the above result to show efficiency bounds. 
The efficiency ratio $\ER_v(\cal A)$ is 
\begin{align*}
    &\frac{\sum_{i=1}^{n}v_i(A_i)}
{\sum_{i=1}^{n}v_i(A^*_i)} \ge\\
& \frac{\sum_{i=1}^{\log n/\log (\frac{1}{1-\beta})}v_i(A^*_i)}
{\sum_{i=1}^{\log n/\log (\frac{1}{1-\beta})}v_i(A^*_i) + \sum_{i=\log n/\log (\frac{1}{1-\beta})+ 1}^{n}v_i(A^*_i)} 
\end{align*}

Let us assume that the minimum utility achieved by any player among the first $\log n/\log (\frac{1}{1-\beta})$ players is $c$. 
This makes the $\ER_v(\cal A)$,
\begin{align}
& \ge \frac{(\log n/\log (\frac{1}{1-\beta}))c}{(\log n/\log (\frac{1}{1-\beta}))c + \sum_{i=\log n/\log (\frac{1}{1-\beta})+ 1}^{n}v_i(A^*_i)} \notag   
\end{align}
For the remaining players, the optimal utility is bounded by $\rho c$ since anything higher would violate the equilibrium condition. This is because if any remaining player (say $i^*$) receives a bundle with value $> \rho c$, then some player (say $i$) with a higher budget who currently has a value of $c$ for their allocated bundle will have a value $> c$ for the bundle allocated to player $i^*$. 
This violates the equilibrium condition since player $i$ can afford $A_{i^*}$ and strictly prefers $A_{i^*}$ to their allocation. 
This implies that the efficiency is 
\begin{align}
& \ge \frac{(\log n/\log (\frac{1}{1-\beta}))c}{(\log n/\log (\frac{1}{1-\beta}))c + (n-\log n/\log (\frac{1}{1-\beta}))\rho c}\ \notag \\   
& = \frac{(\log n/\log (\frac{1}{1-\beta}))}{\rho n + (1-\rho)\log n/\log (\frac{1}{1-\beta})} \notag 
 > \frac{\log n}{\rho n \log (\frac{1}{1-\beta})} \notag 
\end{align}
\end{proof}

\section{Missing Proofs from Section \ref{sec:singleminded}}\label{apdx:single-minded}

\thmSMEfficiencyNP*
\begin{proof}
We use a reduction from the NP-Complete problem SET PACKING:
\begin{quote}
    {\em Given a collection $\cal C = \{C_1, C_2, \dots , C_n\}$ of finite sets, all of which are a subset of a universal set $U = \{e_1, e_2, \dots e_m\}$, and a positive integer $K \le n$, does C contain at least $K$ mutually disjoint sets? .}
\end{quote}
Given a collection $\cal C$, a universal set $U$ and an integer $K$, construct a market with $n$ players $N= \{1,2,\dots,n\}$ and $m+n$ goods $G = \{g_1, g_2, \dots, g_{m+n}\}$ where $g_1 = e_1, g_2 = e_2, \dots, g_m = e_m$. Let each player have an arbitrary non-zero budget $b_i$ and desired set $\Des_i = C_i \cup \{g_{m+i}\}$. We show that there exists an equilibrium with total welfare at least $K$ if and only if $\cal C$ has a disjoint collection of $K$ sets. 

If an equilibrium with total welfare at least $K$ exists, then there are at least $K$ players who receive their desired set. This means that the $C_i$'s for all the players who receive their desired must be disjoint; otherwise, the equilibrium allocation would not be feasible. Therefore, there are at least $K$ sets in $\cal C$ which are mutually disjoint.  

If there are at least $K$ sets which are mutually disjoint, we can construct an equilibrium as follows:

Assume w.l.o.g. the sets $C_1, C_2, \dots, C_K$ are mutually disjoint. Furthermore, assume w.l.o.g. that they are maximal i.e. there is no other set $C_i \in \cal C$ which can be added to $\{C_1, C_2, \dots, C_K\}$ to create a set of $K+1$ mutually disjoint sets. For each $i \in \{1,2,\dots,K\}$, assign each good in the bundle $C_i$ a non-zero price such that the total price is equal to $b_i$ and price good $g_{m+i}$ at zero. For each $i \in \{K+1,K+2, \dots, n\}$, assign each good $g_{m+i}$ a price equal to $b_i$. For all the goods whose prices have not been defined so far, set them to zero. Now, allocate the first $K$ players their desired set and for every player $i \in \{K+1, K+2, \dots, n\}$, allocate the good $g_{m+i}$. Lastly, assign all the remaining goods to player $n$. This allocation has a total welfare at least $K$ since $K$ players get their desired set.
It is also easy to verify that the above allocation is an equilibrium since any player $i \in \{K+1, K+2, \dots, n\}$ (who do not get their desired set) cannot afford their desired set. This is because $C_i$ intersects with another set $C_j$ such that $j \in \{1, 2, \dots, K\}$. Otherwise, this would violate the maximality assumption. Therefore, $C_i$ has a non-zero price and the bundle $C_i \cup \{g_{m+i}\}$ has a price strictly greater than $b_i$.
This concludes our proof.
\end{proof}

\thmSMinfobounds*
\begin{proof}
Consider a market with $n$ players and $k$ goods. Define a set of single minded valuation function profiles $\cal V'$ as follows: the desired set of each player consists of only one good. This good is referred to as the desired good. Furthermore, let no two players in the top $\min\{n, k\}$ players budget wise have the same desired good. 

Define the budget vector $\{b_1, b_2, \dots, b_n\}$ as any budget vector such that $b_1 > b_2 > \dots > b_n$. Let us call this vector of budgets $\vec{b'}$.

Now, suppose the only sample we have is the set of goods $G$ ($\cal S = \{G\}$) and $v_i(G) = 1$ for all $i \in N$. This sample set is consistent with all the valuation function profiles in $\cal V'$.

Note that for any valuation profile $v \in \cal V'$, the best equilibrium allocation is where the top $\min\{n, k\}$ players get their desired good. This allocation gives us a total value of $\min\{n, k\}$.

Suppose that allocation $\cal A$ allocates all the goods to one player. The maximum total utility that $\cal A$ can guarantee is $1$. This allocation gives us an efficiency 
\begin{align*}
    \min_{v \in \cal V} ER_v(\cal A) \le \frac{1}{\min\{n, k\}} 
\end{align*}
since the maximum utility that the optimal equilibrium allocation can get among all the valuation function profiles consistent with $\cal S$ is lower bounded by $\min\{n, k\}$.

If this is not the case and $\cal A$ allocates goods to more than one player, then we show that the maximum utility that $\cal A$ can guarantee is $0$. Let $\cal A$ allocate non-empty bundles to players in $\{i_1, i_2, \dots, i_{n'}\}$. Therefore, the bundles $\{A_{i_1}, A_{i_2}, \dots, A_{i_{n'}}\}$ are non-empty. There exists a valuation function in $\cal V'$ such that the desired good of $i_1$ is in $A_{i_2}$, the desired good of $i_2$ is in $A_{i_3}$ and so on till finally, the desired good of $i_{n'}$ is in $A_{i_1}$. All the players in $\{i_1, i_2, \dots, i_{n'}\}$ have different desired goods here implying that all the players in $\{i_1, i_2, \dots, i_{n'}\}$ which are in the top $\min\{n, k\}$ budget wise players have different desired goods. For those players in the top $\min\{n, k\}$ budget wise who are not allocated any goods, we can set their desired good such that no two players in the top $\min\{n, k\}$ budget wise have the same desired good. This valuation profile is in $\cal V'$ and is consistent with $\cal S$. The optimal equilibrium utility in this case is non-zero trivially and therefore the efficiency guaranteed by this allocation is $0$.

This means, given the set of samples and the set of budgets as defined above, we cannot guarantee an efficiency greater than $\frac{1}{\min\{n, k\}}$. This means that 
\begin{align*}
    \min_{v \in \cal V} \max_{\cal A} \min_{\cal S \subseteq 2^G, b \in \cal B} \ER_v(\cal A) \le
    \min_{v \in \cal V} \max_{\cal A} \min_{\cal S = \{G\}, \vec{b} = \vec{b'}} \ER_v(\cal A) \le \frac{1}{\min\{n, k\}}
\end{align*}
This concludes the proof.
\end{proof}

\section{Missing Proofs from Section \ref{sec:additive}}\label{apdx:additive}


\thmadditiveinfobounds*
\begin{proof}
Consider a market with $n$ players and $k$ goods. We divide this proof into two parts.

\textbf{When $n \ge k$:}
Define a set of additive valuation function profiles $\cal V'$ as follows: each player has one good for which $v_i(\{g\}) = b_i$ and every other good has value $0$ for this player. We refer to the good with non-zero valuation as the favourite good of player $i$. Also, let no two players in the top $k$ players budget wise have the same favourite good. This set of valuations satisfies our budget normalisation condition. 

Define the budget vector $\{b_1, b_2, \dots, b_n\}$ as follows: for every player $b_i = b_1 - \delta_i$ where $\delta_1 = 0$, $0 < \delta_2 < \delta_3 < \dots < \delta_n$ and $\delta_n = \frac{\delta b_1}{k}$. Let us call this vector of budgets $\vec{b'}$.

Now, suppose the only sample we have is the set of goods $G$ ($\cal S = \{G\}$) and $v_i(G) = b_i$ for all $i \in N$. This is consistent with all the valuation function profiles in $\cal V'$.

Note that for any valuation profile $v \in \cal V'$, the best equilibrium allocation is where the top $k$ players get their favourite good. This allocation gives us a total value of 
\begin{align}
    \sum_{i \in N} v_i(A^*_i) &= \sum_{i = 1}^{k} b_i \notag \\
    &= k b_1 - \sum_{i = 1}^{k} \delta_i \notag \\
    &\ge k b_1 - \sum_{i = 1}^{k} \delta_n \notag \\
    &= k b_1 - \delta b_1 \label{eqn:additive-info-bound-optimal-utility}
\end{align}

Suppose that allocation $\cal A$ allocates all the goods to one player. The maximum total utility that $\cal A$ can guarantee is $b_1$ and this arises when the entire bundle is allocated to player $1$. Allocating the entire bundle to any other player will give us a strictly lower utility since all other players have a lower budget. This allocation gives us an efficiency 
\begin{align*}
    \min_{v \in \cal V} ER_v(\cal A) \le \frac{b_1}{k b_1 - \delta b_1} = \frac{1}{k - \delta}
\end{align*}
since the maximum utility that the optimal equilibrium allocation can get among all the valuation function profiles consistent with $\cal S$ is lower bounded by Equation \ref{eqn:additive-info-bound-optimal-utility}.

If this is not the case and $\cal A$ allocates goods to more than one player, then we show that the maximum utility that $\cal A$ can guarantee is $0$. Let $\cal A$ allocate non-empty bundles to players in $\{i_1, i_2, \dots, i_{n'}\}$. Therefore, the bundles $\{A_{i_1}, A_{i_2}, \dots, A_{i_{n'}}\}$ are non-empty. There exists a valuation function in $\cal V'$ such that the favourite good of $i_1$ is in $A_{i_2}$, the favourite good of $i_2$ is in $A_{i_3}$ and so on till finally, the favourite good of $i_{n'}$ is in $A_{i_1}$. All the players in $\{i_1, i_2, \dots, i_{n'}\}$ have different favourite goods here implying that all the players in $\{i_1, i_2, \dots, i_{n'}\}$ which are in the top $k$ players budget wise have different favourite goods. For those players in the top $k$ budget wise who are not allocated any goods, we can set their favourite good such that no two players in the top $k$ budget wise have the same favourite good. This valuation profile is in $\cal V'$ and is consistent with $\cal S$. The optimal equilibrium utility in this case is non-zero trivially and therefore the efficiency guaranteed by this allocation is $0$.

This means, given the set of samples and the set of budgets as defined above, we cannot guarantee an efficiency greater than $\frac{1}{k - \delta}$. This means that 
\begin{align*}
    \min_{v \in \cal V} \max_{\cal A} \min_{\cal S \subseteq 2^G, b \in \cal B} \ER_v(\cal A) \le \min_{v \in \cal V} \max_{\cal A} \min_{\cal S = \{G\}, \vec{b} = \vec{b'}} \ER_v(\cal A) \le \frac{1}{k - \delta}
\end{align*}
\textbf{When $n < k$}:
Let $G'$ be a subset of $G$ such that $|G'| = n$. Define a set of additive valuation function profiles $\cal V'$ as follows: each player has one good in $G'$ for which $v_i(g) = b_i$ and every other good in $G'$ has value $0$ for this player. We refer to this good with non-zero valuation as the favourite good of player $i$. Also, let no two players have the same favourite good. Each of the goods in $G \setminus G'$ is valued by exactly one player in $N$ at a value equal to their budget. Note that it is not necessary for all the goods in $G\setminus G'$ to be valued by the same player. All the valuations in $\cal V$ satisfy the budget normalisation property.

Define the budget vector $\{b_1, b_2, \dots, b_n\}$ as follows: for every player $b_i = b_1 - \delta_i$ where $\delta_1 = 0$, $0 < \delta_2 < \delta_3 < \dots < \delta_n$ and $\delta_n = \frac{\delta b_1}{k}$. Let us call this vector of budgets $\vec{b'}$.

Now, suppose the only sample we have is the set $G'$ ($\cal S = \{G'\}$) and $v_i(G') = b_i$ for all $i \in N$. This is consistent with all the valuation function profiles in $\cal V'$.

Note that for any valuation profile $v \in \cal V'$, the best equilibrium allocation is where all the players get their favourite good and the goods in $G \setminus G'$ are given to the only player who values them at a non-zero value. 

This allocation gives us a total value of 
\begin{align}
    \sum_{i \in N} v_i(A^*_i) &\ge \sum_{i = 1}^{k} b_n \notag \\
    &= k b_1 - k \delta_n \notag \\
    &= k b_1 - \delta b_1 \label{eqn:additive2-info-bound-optimal-utility}
\end{align}
Before we prove the highest utility a consistent allocation can guarantee, we first show that no allocation can guarantee any utility from any good in the set $G \setminus G'$ when the valuation function profile is in $\cal V'$. If the allocation allocates all the goods in $G \setminus G'$ to one player (say $i$), there exists a valuation function profile with the same set of favourite goods where all the goods in $G \setminus G'$ is valued by some player $j \ne i$. If this is not the case and the allocation allocates the good in $G \setminus G'$ to multiple players (say $\{j_1, j_2, \dots, j_{n'}\}$), then there exists a valuation function profile in $\cal V$ with the same favourite goods such that all the goods given to $j_2$ are valued by $j_1$, all the goods given to $j_3$ are valued by $j_2$ and so on till finally, all the goods given to $j_1$ are valued by $j_{n'}$. Either way, there exists a valuation function for which no good in $G \setminus G'$ provides any value. Therefore, we only need to look at the utility guaranteed by goods in $G'$.

Now, suppose that allocation $\cal A$ allocates all the goods in $G'$ to one player. The maximum total utility that $\cal A$ can guarantee is $b_1$ and this arises when the entire bundle is allocated to player $1$. This is because $G'$ guarantees a utility of $b_1$ and $G \setminus G'$ cannot guarantee a non-zero utility. Allocating the entire bundle to any other player will give us a strictly lower utility since all other players have a lower budget. This allocation gives us an efficiency 
\begin{align*}
    \min_{v \in \cal V} ER_v(\cal A) \le \frac{b_1}{k b_1 - \delta b_1} = \frac{1}{k - \delta}
\end{align*}
since the maximum utility that the optimal equilibrium allocation can get among all the valuation function profiles consistent with $\cal S$ is lower bounded by Equation \ref{eqn:additive2-info-bound-optimal-utility}.

If this is not the case and $\cal A$ allocates goods in $G'$ to more than one player, then we show that the maximum utility that $\cal A$ can guarantee is $0$. Let $\cal A$ allocate non-empty subsets of $G'$ to players in $\{i_1, i_2, \dots, i_{n'}\}$. Therefore, the bundles $\{A_{i_1}, A_{i_2}, \dots, A_{i_n'}\}$ are non-empty. There exists a valuation function in $\cal V'$ such that the favourite good of $i_1$ is in $A_{i_2}$, the favourite good of $i_2$ is in $A_{i_3}$ and so on till finally, the favourite good of $i_{n'}$ is in $A_{i_1}$. All the players in $\{i_1, i_2, \dots, i_{n'}\}$ have different favourite goods. For those players who are not allocated any goods, we can set their favourite good such that no two players have the same favourite good. Furthermore, we can choose a valuation function in $\cal V$ with these favourite goods such that no utility is guaranteed by the goods in $G \setminus G'$. The optimal equilibrium utility in this case is non-zero trivially and therefore the efficiency guaranteed by this allocation is $0$.

This means, given the set of samples and the set of budgets as defined above, we cannot guarantee an efficiency greater than $\frac{1}{k - \delta}$. This means that 
\begin{align*}
    \min_{v \in \cal V} \max_{\cal A} \min_{\cal S \subseteq 2^G, b \in \cal B} \ER_v(\cal A) \le 
    \min_{v \in \cal V} \max_{\cal A} \min_{\cal S = \{G'\}, \vec{b} = \vec{b'}} \ER_v(\cal A) \le \frac{1}{k - \delta}
\end{align*}
This concludes the proof.
\end{proof}

\end{document}